\documentclass[a4paper,11pt]{article}

\usepackage[dvipsnames]{xcolor}  
\newcommand{\RC}[1]{\textcolor{red}{#1}}  
 
 \usepackage[english]{babel}  
 \usepackage[utf8]{inputenc}
 \usepackage{cite}

 \usepackage[totalwidth=16.5cm,totalheight=23.5cm]{geometry}
 
\usepackage{import}

\usepackage{amsmath,amssymb,amsthm,mathtools,mathrsfs}
\usepackage{bm,faktor}

  \newcommand{\ga}{\alpha}
  \newcommand{\bga}{\bm{\alpha}}
  
  \newcommand{\gb}{\beta}

  \newcommand{\Ch}{\hat{C}}					
   
  \newcommand{\bCt}{\bm{\tilde{C}}}  
   
  \newcommand{\bbC}{\mathbb{C}}
  \newcommand{\cC}{\mathcal{C}}	
  \newcommand{\bC}{\bm{C}}

  \newcommand{\gC}{\Gamma}	
  
  \newcommand{\Dt}{\tilde{D}}				
    
  \newcommand{\gd}{\delta}					
  \newcommand{\gD}{\Delta}
  \newcommand{\bgD}{\bm{\Delta}}
  \newcommand{\gDh}{\hat{\Delta}}

  \newcommand{\eps}{\epsilon}				
   \newcommand{\beps}{\bm{\epsilon}}

  \renewcommand{\bf}{\bm{f}}				
  
  \newcommand{\fh}{\hat{f}}

  \newcommand{\bphi}{\bm{\phi}} 		
  
   \newcommand{\bgg}{\bm{g}}				
  \newcommand{\bGG}{\bm{G}}  
  	
  \newcommand{\ggh}{\hat{g}}
  \newcommand{\GGh}{\hat{G}}

  \newcommand{\hh}{\hat{h}}	
  \newcommand{\bh}{\bm{h}}		
  
  \newcommand{\scrI}{\mathscr{I}}
  
  \newcommand{\bK}{\bm{K}}	
  
  \newcommand{\lh}{\hat{l}}					
  \newcommand{\bl}{\bm{l}}	
  					
  \newcommand{\cL}{\mathcal{L}}	
  \newcommand{\gl}{\lambda}
  \newcommand{\gL}{\Lambda}
  
  \newcommand{\ellh}{\hat{\ell}}
   \newcommand{\bell}{\bm{\ell}}

  \newcommand{\Mb}{\overline{M}}			
  \newcommand{\mb}{\bar{m}}
  \newcommand{\Mh}{\hat{M}}

  \newcommand{\bmu}{\bm{\mu}}
  \newcommand{\cM}{\mathcal{M}}

  \newcommand{\nh}{\hat{n}}	
  \newcommand{\Nh}{\hat{N}}
  \newcommand{\Nt}{\tilde{N}}	   
  \newcommand{\bn}{\bm{n}}					
  \newcommand{\bN}{\bm{N}}

  \newcommand{\cO}{\mathcal{O}}				
  \newcommand{\bgo}{\bm{\omega}}
    
  \newcommand{\go}{\omega}

  \newcommand{\gO}{\Omega}
  \newcommand{\gOh}{\hat{\Omega}}

  \newcommand{\Ph}{\hat{P}}					
    
  \newcommand{\cP}{\mathcal{P}}

  \newcommand{\bpsi}{\bm{\psi}}
  
  \newcommand{\Rh}{\hat{R}}

  \newcommand{\bbR}{\mathbb{R}}
  \newcommand{\gr}{\rho}

  \newcommand{\bgs}{\bm{\sigma}}			
  \newcommand{\bgsh}{\bm{\hat{\sigma}}}
  \newcommand{\gs}{\sigma}
  \newcommand{\gS}{\Sigma}

  \newcommand{\bT}{\bm{T}}	 
  \newcommand{\cT}{\mathcal{T}}				
  \newcommand{\gt}{\tau}
  \newcommand{\bgt}{\bm{\tau}}

  \newcommand{\gth}{\theta}					
  
  \newcommand{\bgth}{\bm{\theta}}
  
  \newcommand{\uh}{\hat{u}}					

  \newcommand{\gU}{\Upsilon}
  \newcommand{\bgU}{\bm{\Upsilon}}				
  
  \newcommand{\bV}{\bm{V}}
  \newcommand{\bVh}{\hat{\bm{V}}}

  \newcommand{\Vh}{\hat{V}}

  \newcommand{\wb}{\bar{w}}					
  \newcommand{\Wh}{\hat{W}}

  \newcommand{\xh}{\hat{x}}  				
  				
  \newcommand{\bX}{\bm{X}}					
  \newcommand{\bxi}{\bm{\xi}}  
  
  \newcommand{\Yh}{\hat{Y}}					
  \newcommand{\yh}{\hat{y}}

  \newcommand{\bY}{\bm{Y}}		
  \newcommand{\bYh}{\hat{\bm{Y}}}		
 
  \newcommand{\gz}{\zeta}
  \newcommand{\gzb}{\bar{\zeta}}

\newcommand{\from}{\colon}
\newcommand{\xto}[2][]{\xrightarrow[#1]{#2}}

\DeclareMathOperator{\Iso}{Iso}

\DeclareMathOperator{\SO}{SO}

\newcommand{\Quotient}[2]{\faktor{#1}{#2}}

\newcommand{\covD}{\nabla}

\newcommand{\covDh}{\hat{\nabla}}

\newcommand{\parD}{\partial}
\newcommand{\parDh}{\hat{\partial}}
\newcommand{\LieD}{\mathcal{L}}
\newcommand{\intD}{\lrcorner}

\newcommand{\So}[1]{\Gamma\left[ #1 \right]}
\newcommand{\Co}[1]{\mathcal{C}^{\infty}\left( #1 \right)}

\DeclareMathOperator{\End}{End}
  
\renewcommand{\S}{\bm{\text{S}}}				

\newcommand{\Mtx}[1]{\begin{pmatrix} #1	\end{pmatrix}}

\newtheorem{Proposition}{Proposition}[section]
\newtheorem{Theorem}{Theorem}[section]

\theoremstyle{definition}
\newtheorem{Definition}{Definition}[section]

\numberwithin{equation}{section}

\newcommand{\M}{X}
\newcommand{\MNull}{\scrI}
\newcommand{\MRiem}{M}
\newcommand{\MConf}{\widetilde{M}}
\newcommand{\gConf}{\bm{g}}

\newcommand{\hConf}{\bm{h}}

\newcommand{\CSoL}[1]{\Gamma_{#1}\left[ L \right]}

\usepackage{comment}
\usepackage{version}
\usepackage{tcolorbox}

\usepackage[debug==false]{hyperref}

\excludeversion{PrintVersion}
 \date{\begin{PrintVersion}
 \today 
 		\begin{center}
 			\RC{This is the version of the file posted on arxiv on the 27/08/20 (accepted for publication in J.Math.Phys)}
 		\end{center}
 \end{PrintVersion}} 
 
\begin{document} 
	\pagestyle{plain}
	\title{Asymptotic Shear and the Intrinsic Conformal Geometry of Null-Infinity}
	\author{Yannick Herfray \footnote{Yannick.Herfray@ulb.ac.be}\\ 
		{\small \it Département de Mathématiques,
			Université Libre de Bruxelles, } \\ {\small \it CP 218, Boulevard du Triomphe, B-1050 Bruxelles, Belgique. }}
	\maketitle
\begin{abstract}
In this article we propose a new geometrization of the radiative phase space of asymptotically flat space-times: we show that the geometry induced on null-infinity by the presence of gravitational waves can be understood to be a generalisation of the tractor calculus of conformal manifolds adapted to the case of degenerate conformal metrics. It follows that the whole formalism is, by construction, manifestly conformally invariant. We first show that a choice of asymptotic shear amounts to a choice of linear differential operator of order two on the bundle of scales of null-infinity. We refer to these operators as Poincaré operators. We then show that Poincaré operators are in one-to-one correspondence with a particular class of tractor connections which we call ``null-normal'' (they generalise the normal tractor connection of conformal geometry). The tractor curvature encodes the presence of gravitational waves and the non-uniqueness of flat null-normal tractor connections correspond to the ``degeneracy of gravity vacua'' that has been extensively discussed in the literature. This work thus brings back the investigation of the radiative phase space of gravity to the study of (Cartan) connections and associated bundles. This should allow, in particular, to proliferate invariants of the phase space.
\end{abstract}
\maketitle

\section{Introduction: motivations and main results}

\subsection{The ``radiative phase space'' of asymptotically flat space-times}
 
 The geometry of asymptotically flat space-times has been traditionally studied from two complementary point of views: On the one hand, the Bondi-Sachs formalism, which was originally designed in \cite{bondi_gravitational_1962,sachs_gravitational_1962} to clarify the physics of gravitational waves, makes use of an especially well chosen set of coordinates in a neighbourhood of null-infinity (see \cite{barnich_aspects_2010,madler_bondi-sachs_2016,ruzziconi_asymptotic_2020} for reviews and modern presentations). In this context, properties of asymptotically flat space-times are studied by performing an asymptotic expansion in the ``radial coordinate'' $ r^{-1}= \gO$.  This is especially convenient since one can work in a very explicit way and it is mainly in this form that the recent developments \cite{strominger_bms_2014,he_bms_2015,strominger_gravitational_2016,strominger_lectures_2018} on the infrared structure of the S-matrix of perturbative quantum gravity and related memory-effects have been studied. On the other hand, the classical results obtained in the Bondi-Sachs coordinates were given a coordinate-free description in the work of Penrose \cite{newman_approach_1962,penrose_zero_1965,penrose_conformal_2011} (see also \cite{tamburino_gravitational_1966,tafel_comparison_2000} for comparisons between the two formalisms). This invariant point of view is essential to be able to precisely phrase questions of global nature (see \cite{frauendiener_conformal_2004,friedrich_geometric_2015,frauendiener_conformal_2018} for modern reviews). From this perspective, asymptomatically flat space-times are understood to be particular case of conformally compact Lorentzian manifold. Practically, the ``physical'' space-time is taken to be conformally isometric to the interior of a compact manifold $\Mb$ ( the ``unphysical'' space-time) with boundary $\parD \Mb = \MNull$ (``null-infinity''). The ``radial coordinate'' $\gO$ then plays the role of a boundary defining function for $\MNull$ (Generically the boundary $\parD \Mb$ will be disconnected, we are here really interested by the behaviour of the metric near one of the connected components $\MNull$).
  
 An especially interesting (and classical) question in this context is to understand the ``asymptotic'' or ``radiative'' phase space of gravity i.e the subset of the phase-space corresponding to gravitational waves. In the Bondi-Sachs coordinates this amounts to the choice of fields parametrizing the very first terms in the expansion:
 \begin{equation}\label{Introduction: Bondi-Sachs metric}
 g_{\mu\nu} = \frac{1}{\gO^2}\left( 2 d\gO du + h_{AB} + \gO C_{AB} + \cO\left(\gO^2\right)\right).   
 \end{equation}
 Here $h_{AB}$ is a two-by-two non-degenerate symmetric tensor and $C_{AB}$ is a two-by-two symmetric trace-free (with respect to $h_{AB}$) symmetric tensor. In this sense these two fields parametrize the radiative phase space of gravity and it was indeed the realisation that $C_{AB}$  could effectively be interpreted as the flux of energy at null infinity through Bondi's mass loss formula which was one of the most (rightfully) celebrated result of BMS's work. 
 From the conformal compactification point of view, and since general relativity intrinsically is a geometrical theory, one naively expect the ``radiative''  phase space to correspond to a, reasonably standard, geometric moduli space.  Considering that $h_{AB}$ and $C_{AB}$ appear at very low order in the Bondi-Sachs expansion, it is tempting to think that they should amount to a particular geometrical structure on $\MNull$ and indeed one can show \cite{geroch_asymptotic_1977,ashtekar_radiative_1981, ashtekar_a._symplectic_1981,ashtekar_symplectic_1982,ashtekar_asymptotic_1987,ashtekar_geometry_2015,ashtekar_null_2018} that the leading order terms (i.e the first two) in this expansion equip $\MNull$ with a ``universal null-infinity structure''
 \begin{Definition}{Universal null-infinity structure}\mbox{}\label{Introduction: Def, universal null-infinity structure}
 	
 We will say that a $n$-dimensional manifold $\MNull$ is equipped with a universal null-infinity structure if
 \begin{itemize}
 	\item  it is the total space of a trivial bundle $\MNull \xto{\pi} \gS$ over an $(n-1)$-dimensional orientable manifold $\gS$,
 \end{itemize}
and it is equipped with
 \begin{itemize}
 	\item a conformal class $[h_{AB}]$ of metric on $\gS$ i.e $h_{AB} \sim \go^2 h_{AB}$ with $\go \in \Co{\gS}$
 	\item an equivalence class of (nowhere vanishing) vertical vector fields $(n^a , h_{AB}) \sim (\go^{-1} n^a , \go^2 h_{AB})$ with $d\pi^A{}_b\; n^b~=~0$.
 \end{itemize}
Our convention is that lower case Latin indices are abstract indices for tensors on $\MNull$ while upper case Latin indices are abstract indices for tensors on $\gS$. We call null-infinity manifold a manifold $\MNull$ equipped with a universal null-infinity structure.
 \end{Definition}	
\noindent This can be directly obtained by restricting the suitable tensors to $\MNull$:
 \begin{align}
h_{AB} &= \gO^2 g_{\mu\nu} \big|_{\gO =0}, &  n^a &= \gO^{-2} \; d\gO_{\nu} g^{\mu\nu}\big|_{\gO =0}.
\end{align} There is however nothing unique about the boundary defining function $\gO$, one could just have well have considered instead $\go \gO$ and this fact leads to the appearance of equivalence classes in the definition.

So far so good, in particular the symmetry group of this universal structure is known \cite{ashtekar_geometry_2015, duval_conformal_2014, duval_conformal_2014-1} to be the celebrated BMS group 
\begin{equation}
BMS\left( \MNull \to \gS , [h_{AB}] , [n^a]\right) = \Co{\gS} \rtimes  Conf\left( \gS , [h_{AB}] \right)   
\end{equation}
where $Conf\left( \gS , [h_{AB}] \right)$ is the space of conformal isometries of $\gS$ and $\Co{\gS}$ parametrize ``super-translations'' along the fibres. 

Things get more intricate when considering the sub-leading order in the Bondi-Sachs expansion i.e the ``asymptotic shear'' $C_{AB}$, which is known to encode the radiative information of asymptotically flat 4D space-times. One might consider the idea that $C_{AB}$ induces a tensor on $\MNull$, however there is once again nothing unique about the coordinates used to represent the metric \eqref{Introduction: Bondi-Sachs metric} which means that $C_{AB}$ is subject to what might appear as strange transformation rules: 
If one applies the change of coordinates
\begin{align}
	\gOh &=\go \gO + \cO\left( \gO^2 \right) & \uh &= \go\left( u-\xi\right)  + \cO\left( \gO \right) & \xh^A &= x^A + \cO\left( \gO\right)   
\end{align}
(with $\go$, $\xi\in \Co{\gS}$ and the sub-leading terms in $\gO$ chosen to be such that this change of coordinates preserve the form of the metric \eqref{Introduction: Bondi-Sachs metric}) we obtain
\begin{align}\label{Introduction: C transformation rules}
	\hh_{AB} &= \go^2 h_{AB}\nonumber \\ 
	\nh^a &= \go^{-1} n^a \\
	\Ch_{AB} & = \go\Big(\;C_{AB} -2\; \covD_A \covD_B \big|_0 \xi -2\; \uh \; \covD_A \covD_B \big|_0 \go^{-1} \Big) \nonumber
\end{align}
where $\big|_0$ stands for ``trace-free part of''. The complicated transformation rules for the asymptotic shear should leave no doubt that $C_{AB}$ induces on $\MNull$ nothing like a tensor. 

In the beautiful pieces of work \cite{geroch_asymptotic_1977,ashtekar_radiative_1981,ashtekar_a._symplectic_1981,ashtekar_symplectic_1982,ashtekar_asymptotic_1987} (see \cite{ashtekar_geometry_2015,ashtekar_null_2018} for modern reviews) huge steps were realised in the geometrization of this radiative phase space. It was understood that the asymptotic shear could be interpreted as the coordinates of a ``covariant derivative'' $D$ on $\MNull$ compatible with the ``universal'' structure, $D h_{AB}=0$, $Dn^a =0$. Due to the fact that the metric on null-infinity is degenerate these conditions do not determined the ``connection'' uniquely, rather we obtain an affine space modelled on trace-free symmetric tensor on $\gS$. The beauty of the construction resides in the fact that presence of gravitational waves then correspond to the non-vanishing of the ``curvature'' of this ``connection''. In particular, the non-uniqueness of Minkowski vacuum was very elegantly interpreted as the existence of a whole family of ``flat connections'' (i.e ``connections'' with vanishing ``curvature'') with the BMS group acting transitively on this space of vacuum configurations (each of the stabilisers then corresponding to a different copy of the Poincaré group).

Despite the undeniable progress that this work represented we however believe that there is still room for improvement because of the following unappealing features : On the one hand we had to put the words ``connection'' and ``curvature'' between inverted comas for, due to conformal invariance, one is really here working with an equivalence class of connections and indeed from the conformal compactification point of view this is quite clear that there is no invariant notion of affine connection defined on $\MNull$. It follows that there is nothing straightforward about the related notion of curvature, in fact the construction appear as rather unnatural from the point of view of $\MNull$ itself. More generally, it is a really non-trivial guess game to be able to construct any kind of invariants from these equivalence class of connections. Since the whole point of a geometrical approach is to be able to work invariantly this state of affair is quite unsatisfactory. Another sign that this approach misses some part of the underlying geometrical structure is that ``flat connections'', i.e those associated with the absence of gravitational waves, are only remotely related with the existence of four-parameter families of ``good-cuts'' \cite{newman_heaven_1976,hansen_r._o._metric_1978,ko_theory_1981,adamo_generalized_2010,adamo_null_2012} which have been understood to be a crucial feature of flat space-times. 

The aim of this article is to propose an intrinsic description of the geometry of null-infinity avoiding the drawbacks that we just listed. As a bonus, our formalism will be versatile enough to allow to work in a generic dimension $n\geq2$ for $\MNull$. Our main result is that any $n$-dimensional manifold equipped with a universal null-infinity structure as in definition \ref{Introduction: Def, universal null-infinity structure} can be equipped with a compatible \emph{null-normal tractor connection}. For $n+1=4$ these connections are not unique but rather form an affine space modelled on symmetric trace-free tensors on $\gS$ (here $n=3$ is the dimension of the null-infinity manifold, see the core of the article for results in other dimensions). As opposed to previous works, these tractor connections are \emph{bona fide} (Cartan) connection on $\MNull$. Their curvature (now defined in the standard way) then naturally encodes the gravitational degrees of freedom. These connections naturally act on an $(n+2)$-dimensional vector bundle canonically constructed from the universal structure, the tractor bundle. In this way we obtain a ``tractor calculus'' on $\MNull$ which allows to easily proliferate invariants. Finally, solutions of the good-cuts equations are in directly related with tractors which are covariantly constant on $\MNull$ and therefore the existence of a $(n+1)$-parameters family of good-cuts is straightforwardly related to the vanishing of the curvature (equivalently the existence of $(n+2)$ linearly independent covariantly constant tractors) and thus to the absence of gravitational radiation .

In the rest of this article, and as opposed to the point of view taken in this introduction, the whole geometry will be discussed intrinsically, i.e solely in terms of the intrinsic (tractor) geometry of null-infinity manifold (as in definition \ref{Introduction: Def, universal null-infinity structure}). In a future work \cite{Herfray_2020} we plan to detail how this geometry can be naturally derived from the tractor geometry of asymptotically flat space-times.

As this article was in preparation, the author has become aware of the upcoming work \cite{Korovin_2020} that should share some common features with this present work and \cite{Herfray_2020}. These results have been obtained independently with motivations coming from holographic dualities.

\subsection{Degenerate conformal geometries and tractors}

As we just explained the main motivation for this work is to provide new tools to study the radiative phase space of gravity. Another important motivation, however, is to lay the basis of the investigation of the geometry of degenerate conformal manifolds through the means of Cartan connections.

Our main inspiration for starting this work indeed came from the conformal (more generally parabolic) geometry literature \cite{bailey_thomass_1994,sharpe_differential_1997,cap_parabolic_2000,cap_parabolic_2009} and especially recent works of R.Gover and collaborators \cite{gover_almost_2005,gover_conformal_2007,gover_almost_2010,curry_introduction_2018,gover_boundary_2014,gover_poincare-einstein_2015,gover_calculus_2016} which started the systematic investigation of the geometry of conformally compact manifolds through the use of tractors: this is especially satisfying as results can be then be stated in manifestly conformal invariant forms. Most of the difficulties appearing in the work \cite{geroch_asymptotic_1977,ashtekar_radiative_1981, ashtekar_a._symplectic_1981,ashtekar_symplectic_1982,ashtekar_asymptotic_1987,ashtekar_geometry_2015,ashtekar_null_2018} can indeed be traced back to the inherent complication of working with an equivalence class of metric: each metric representative defines its own Levi-Civita connection and this makes the construction of conformal invariants a non-trivial task. In the conformal geometry literature, two main sets of tools have been developed to circumvent this inherent difficulty: These are the ambient metric of Fefferman and Graham \cite{fefferman_conformal_1985,fefferman_ambient_2012} and tractor calculus \cite{bailey_thomass_1994} (the two being in fact closely related see \cite{cap_standard_2003}). They both allow to study conformal geometry in an \emph{essentially manifestly conformally invariant} way. In effect, the tractor calculus of conformal geometry acts as an equivalent of Ricci calculus for Riemannian geometry. We here briefly recall the basics of tractor calculus and describe the kind of modifications which are needed to adapt to the situation where the conformal metric is degenerate.

Let $\left( M , [h_{AB}] \right)$ be a $n$-dimensional conformal (Riemannian, i.e of Euclidean signature) manifold. The starting point is to convert the ``equivalence class of metric'' description in a more geometrical picture making use of the ``bundle of scale'' $L \to M$. If $x$ is a point on $M$ then the fibre $L_x$ at $x$  correspond to all possible choice of representative $h_{AB} \in [h_{AB}]$ at $x$. Accordingly, each choice of metric representative correspond to a choice of section of $L$. Choosing a representative $h \in [h_{AB}]$ then amounts to choosing a trivialisation of $L$ i.e amounts to work in a particular coordinate system: Conformal invariants are the geometrical objects that do not depend on (or transform covariantly under) this choice of coordinate. The tractor bundle $\cT \to M$ is then a $(n+2)$-dimensional vector bundle canonically constructed from the 2-jet bundle of $L$ and is naturally equipped with a metric of signature $\left(n+1,1\right)$. 

Just like the starting point of Ricci calculus is the fact that their is a unique torsion-free metric-compatible connection on the tangent bundle, the starting point of tractor calculus is the realisation that for $n>2$ there is a unique \emph{normal} connection on the tractor bundle compatible with the metric. Here ``normal'' refers to a natural set of conditions that one needs to impose on the curvature to obtain unicity (this plays a similar role in spirit to torsion-freeness for the Levi-Civita connection). This ``normal tractor connection'', canonically constructed from $[h_{AB}]$, really is a Cartan connection modelled on the conformal sphere \begin{equation}\label{Introduction: conformal sphere}
	\left(\S^n , [h_{\S^n}] \right) = \Quotient{\SO\left(n+1 , 1\right) }{Iso\left(n\right)  \rtimes\bbR}
\end{equation}
where $Iso\left(n\right)\rtimes\bbR \subset \SO\left(n+1 , 1\right) $ is the stabiliser of a null line in $\bbR^{(n+1,1)}$.

By construction the tractor bundle and the tractor connection do not rely on a choice of representative $h_{AB} \in [h_{AB}]$, rather each choice of such representative amounts to picking a particular ``splitting'' of the tractor bundle. More precisely a section of $L$ (equivalently a choice of representative $h_{AB} \in [h_{AB}]$) gives an isomorphism
\begin{equation}
\cT \to L \oplus \left( TM \otimes L^{-1} \right) \oplus L^{-1}.
\end{equation}
The equivalence class of Levi-Civita connections associated with the equivalence class of metrics $[h_{AB}]$ then correspond to an equivalence class of ``components'' describing the (invariant) tractor connection in each of these splittings. We here purposely use the same type of terminology as in \cite{geroch_asymptotic_1977,ashtekar_radiative_1981, ashtekar_a._symplectic_1981,ashtekar_symplectic_1982,ashtekar_asymptotic_1987,ashtekar_geometry_2015,ashtekar_null_2018} to emphasise the similitude: in this article we wish to suggest that the equivalence class of connections used in these works is best thought as an equivalence class of components for an (invariant) tractor connection.

 The comparison of null-infinity with conformal geometry in dimension $n>2$ however breaks here for there is no such a thing as the normal Cartan connection for a degenerate conformal geometry. On the other hand, a very illuminating parallel can be drawn with two-dimensional conformal geometry.

When $n=2$, a strange phenomenon appear: the unicity of normal tractor connections fails to be true. Rather there is an infinite family of compatible normal connections forming an affine space over quadratic differentials. This is directly related to the fact that the local conformal group for Riemann surfaces are holomorphic transformations rather than $PSL(2,\bbC)$ transforms. The exact geometrical structure needed to regain unicity of normal tractor connections has been extensively described in \cite{calderbank_mobius_2006,burstall_conformal_2010} under the name of Möbius structures:
\begin{Definition}{Möbius structure}\mbox{}\\ \label{Introduction: Def, Mobius structure}
	Let $\left( M, [h_{AB}]\right)$ be a conformal manifold. A compatible Möbius operator
	 \begin{equation}
	\cM \from \So{L}~\to~\So{S^2_0\; T^*M \otimes L}
	\end{equation} is defined to be a linear differential operator of order two such that, in the trivialisation of $L$ given by a choice of representative $h_{AB} \in [h_{AB}]$ it takes the form
\begin{equation}\label{Introduction: Mobius operator}
\ell \mapsto \left( \covD_A \covD_B\big|_0 -\frac{1}{2} N_{AB}  \right)  \ell.
\end{equation}
where $\ell$ are the coordinates of a section $\bell \in \So{L}$, $\covD_A$ is the Levi-Civita connection of $h_{AB}$,  $\covD_A \covD_B \big|_0$ is the trace-free part of the Hessian and $N_{AB}$ is a trace-free symmetric tensor. A Möbius structure $\left( M, [h_{AB}] , \cM\right)$ is a two-dimensional conformal manifold equipped with a compatible Möbius operator. The definition might seem to depend on a choice of representative $h_{AB} \in [h_{AB}]$ but it actually does not: if $\cM$ takes the form \eqref{Introduction: Mobius operator} for $h_{AB} \in [h_{AB}]$ it will have the same form for any other $\hh_{AB} \in [h_{AB}]$ with a non-trivial transformation rule $N_{AB} \mapsto \Nh_{AB}$.
\end{Definition}

Möbius operators form an affine space modelled on trace-free symmetric tensor (i.e quadratic differential) and by the results from \cite{calderbank_mobius_2006,burstall_conformal_2010} are in one-to-one correspondence with choice of normal tractor connections modelled on the conformal two-sphere \eqref{Introduction: conformal sphere}. What is more flat tractor connections (corresponding to holomorphic quadratic differentials) are in one-to-one correspondence with complex projective structure. In other terms a choice of Möbius operator effectively reduces the (infinite dimensional) group of holomorphic transformations to $PSL\left(2,\bbC\right)$:
\begin{equation}
\text{holomorphic maps} \quad \xto{\text{Möbius operator}} \quad  PSL\left(2,\bbC\right).
\end{equation}

 This vividly resemble the situation at null infinity where the BMS group is reduced to the Poincaré group by a choice of asymptotic shear,
 \begin{equation}
 \text{BMS group} \qquad \xto{\text{Poincaré operator}} \quad  \Iso(n,1).
 \end{equation}
  We will prove there is indeed a precise sense in which both phenomenon are related: just like two-dimensional conformal geometry is not rigid enough to allow for a description in terms of Cartan geometry (while Möbius structures are) degenerate conformal geometry needs to be complemented by an extra geometric structure (which we call a Poincaré structure) to become rigid. More precisely, while Möbius structures are equivalent to choices of normal Cartan connection modelled on the conformal sphere \eqref{Introduction: conformal sphere}, we will show that Poincaré structures are equivalent to choices of ``null-normal'' Cartan connection modelled on a realisation of null-infinity as an homogeneous space,
  \begin{equation}\label{Introduction: homogenous model for null-infinity}
  \MNull_{(n,1)} = \Quotient{ \Iso\left(n,1\right) }{Carr\left(n\right) \rtimes\bbR}
  \end{equation}
  where $Carr\left(n\right)\subset \SO\left(n+1,2\right)$ is the stabiliser of two (non-parallel) null vectors in $R^{(n+1,2)}$ (alternatively $Carr\left(n\right)$ is the Carroll group from \cite{levy-leblond_nouvelle_1965,duval_carroll_2014}).

Before we come to a more detailed description of these results, let us already point the following  shortcoming of our work: From the point of view that we just discussed, it is clear that the definition \ref{Introduction: Def, universal null-infinity structure} of ``universal null-infinity structures'' assumes to much. Apart from the topological requirements (which can be easily lifted) the main assumption here is that one assumes the metric to be invariant along the integral line of the vector fields. A more satisfying starting point would be to only assume a degenerate conformal metric $[h]$ together with a vertical vector field $[n]$ without any extra requirement (this structure has been discussed as a ``conformal Carroll structure'' in \cite{duval_carroll_2014,duval_conformal_2014,duval_conformal_2014-1}). One would then expect the invariance of $[h]$ along $[n]$ to only appear as part of the integrability condition for the Cartan connection. In this article, we however restrict ourselves to the ``universal null-infinity structure'' of definition \ref{Introduction: Def, universal null-infinity structure} : Since conformal boundaries of asymptotically flat space-times must have an induced structure of this type (this is imposed by Einstein's equations) this is all one really needs for discussing the physics of gravitational waves. On the other hand, and as far as a systematic description of degenerate conformal geometries is concern, this article can be thought as laying down the basis for future work.

\subsection{Summary of the main results}

Let $\left(\MNull \to \gS , \bh_{ab} , \bn^a \right)$ be a $n$-dimensional null-infinity manifold as in definition \ref{Introduction: Def, universal null-infinity structure}. In line with the standard practice in the conformal geometry literature we introduce the bundle of scales $L\to \MNull$ and the conformal class of metric $[h_{ab}]$ is now thought as a section $\bh_{ab}$ of $S^2T^*\MNull \otimes L^2$ while the equivalence class of vector field amounts to a section $\bn^a$ of $T\MNull \otimes L^{-1}$. See the preceding discussion for a heuristic definition or \cite{curry_introduction_2018} for a gentle introduction, for completeness all concepts of conformal geometry which are not standard in the mathematical physics community will also be reviewed in the bulk of this article.

A first, elementary, result is that null-infinity manifolds are naturally equipped with a vertical derivative on $L$ that we will note $\covD_{\bn} \from L \to \bbR$. This allows to define the set of scales with constant vertical derivatives $\CSoL{k}$ as
\begin{equation}
\bgs \in \CSoL{k} \qquad \Leftrightarrow \qquad \covD_{\bn} \bgs =k.
\end{equation}

In usual conformal geometry $\left( \gS , \bh_{AB} \right)$ sections of $L_{\gS}$ serve as trivialisations. In the context of null-infinity manifold $\left(\MNull \to \gS , \bh_{ab} , \bn^a \right)$ we really have two bundles $\MNull \to \gS$ and $L_{\gS} \to \gS$ and the equivalent of the trivialisations of conformal geometry will be given by well-adapted trivialisations:
\begin{Definition}\mbox{}
	
	A well-adapted trivialisation $\left(\bgs ,u \right)$ is a choice of trivialisation of $L_{\gS}\to \gS$ (given by a nowhere vanishing section $\bgs \in \So{L_{\gS}}$) together with a trivialisation of $\MNull\to \gS$ (given by $u\in \Co{\MNull}$ such that $u \times \pi \from \MNull \to \bbR \times \gS$ is a trivialisation) satisfying the compatibility relation: $\bn^a du_a =\bgs^{-1}$.
\end{Definition}
This definition has the advantage to fit with the notation of the gravitational wave literature in the BMS formalism where one typically work with a fixed coordinates system given by $u$ and a fixed representative $h_{AB} = \bgs^{-2} \bh_{AB}$, the compatibility condition then ensures that ``$\parD_u = \bgs \bn^a$''. One can however prove that, in this context, a choice of well adapted trivialisation is equivalent to choosing two independent sections of $L_{\gS}\to \gS$ and $\MNull \to \gS$ (since $\MNull \to \gS$ is not assumed to be a vector bundle this is not an immediate equivalence, rather it makes use of the extra structure given by $\bn^a$.)

With this preliminaries, one can define strong (or ``radiative'') null-infinity structures
\begin{Definition}{Strong null-infinity structure}\mbox{}\label{Introduction: Def,Poincare operator}\\
	Let $\left( \MNull \to \gS, \hConf_{ab}, \bn^a\right)$ be a null-infinity manifold. A compatible Poincaré operator \begin{equation}
	\cP \from~\CSoL{k}~\to~\So{S^2_0\; T^*\gS \otimes L}
	\end{equation} is defined to be a linear differential operator of order two such that, in a well adapted trivialisation $\left ( \bgs , u \right)$, it takes of the form 
	\begin{equation}\label{Introduction: Poincare operator}
	\cP(\bl)_{AB} =  \bgs \left( \covD_A \covD_B \big|_0+ \frac{1}{2}[C_{AB} , \parD_u ]  \right) l
	\end{equation}
	where $l= \bgs^{-1} \bl$, ``$\parD_u$'' stands for $\bgs\LieD_{\bn}$, $\covD_A$ is the ``horizontal derivative'' given by $u$ and the Levi-Civita connection of $h_{AB} = \bgs^{-2}\hConf_{ab}$,  $\covD_A \covD_B \big|_0$ is the trace-free part of the Hessian, $C_{AB}$ is a trace-free symmetric tensor and $[C_{AB} , \parD_u ]l := C_{AB}\parD_ul - (\parD_u C_{AB})l$ .
	
	 A strong\footnote{We refrain to call these ``radiative'' structures in the spirit of \cite{ashtekar_geometry_2015} for they only describe gravitational radiations on 3-dimensional null-infinity manifolds.} null-infinity structure $\left( \MNull \to \gS, \hConf_{ab}, \bn^a ,\cP\right)$ is a universal null-infinity structure together with choice of compatible Poincaré operator.
\end{Definition}
\begin{Proposition}\mbox{}
	\begin{itemize}
\item Definition \ref{Introduction: Def,Poincare operator} does not depend on a choice of well-adapted trivialisation: if $\cP$ takes the form \eqref{Introduction: Poincare operator} in  $\left( \bgs ,u \right)$ then it will have the same form in any other well-adapted trivialisation $\left(\bgsh , \uh \right)$ with a transformation rule $C_{AB} \mapsto \Ch_{AB}$ given by the next point.

\item Well-adapted trivialisation $\left(\bgsh , \uh \right)$ are parametrised by two functions $\go \in \Co{\MNull}$ and $\xi \in \Co{\MNull}$ such that $\left(\bgsh = \go^{-1}\bgs , \uh = \go\left( u- \xi \right) \right)$ and the transformation rules for $C_{AB}$ are the one given by \eqref{Introduction: C transformation rules}. 

\end{itemize}
\end{Proposition}
We will in fact prove the following
\begin{Theorem}{Asymptotic shear and Poincaré operator}\mbox{}\label{Introduction: Prop,Asymptotic shear and Poincare operator}\\
	Choices of asymptotic shear for an asymptotically flat space-times of dimension $n+1 \geq 4$ are in one-to-one correspondence with choices of Poincaré operators on the null-infinity manifold at the conformal boundary.
\end{Theorem}
Note that, when restricted to sections in $\CSoL{0}$  (i.e such that $\parD_u l =0$) the Poincaré operator \eqref{Introduction: Poincare operator} gives a differential operator formally similar to Möbius operators \eqref{Introduction: Mobius operator} (with $N_{AB} = \parD_u C_{AB}$). One can show that when $\parD_u\parD_u C_{AB}=0$ this indeed defines a Möbius operator on $\left(\gS , \bh_{AB}\right)$.

Apart from the nice relations to physical quantities, the introduction of Poincaré operators will be justified by their equivalence with a certain class of Cartan connections. These are Cartan connection modelled on the realisation of flat null-infinity $\MNull_{(n,1)}$ as the homogenous space \eqref{Introduction: homogenous model for null-infinity}. In the bulk of this paper we will review in details this homogenous space construction. For this introduction, this will however be enough to say that flat null-infinity $\MNull_{(n,1)}$ can be embedded in a null hyper-surface of $\bbR^{(n+1,2)}$. The tractor bundle of a null-infinity manifold is then an infinitesimal version of this embedding:
\begin{Proposition}{Tractor bundle of a null-infinity manifold}\mbox{}\label{Introduction: Prop,Tractor bundle of a null-infinity manifold}\\	
	Let $\left(\MNull \to \gS , \bh_{ab}, \bn^a \right)$ be a null-infinity manifold as in definition \ref{Introduction: Def, universal null-infinity structure}. It is canonically equipped with a $(n+2)$-dimensional vector bundle $\cT \to \MNull$, the ``tractor bundle''. This bundle comes with the following structure
	\begin{itemize}
		\item a degenerate metric $g_{IJ}$ with one-dimensional kernel and signature $\left(n,1\right)$, 
		\item a preferred null section $\bX^I \in \So{\cT \otimes L}$ i.e $g_{IJ} \bX^I \bX^J  =0$,
		\item a preferred section $I^I \in\So{\cT}$ such that $g_{IJ} I^J =0$ and $X_I I^I \neq 0$.
	\end{itemize}
What is more the ``reduced tractor bundle'' defined as the quotient $\cT / I$ is canonically isomorphic to the pull-back of $\cT_{\gS}$, the standard tractor bundle of $\left(\gS ,\bh_{AB} \right)$,
\begin{equation}
\cT / I = \pi^* \left( \cT_{\gS} \right).
\end{equation}	
\end{Proposition}
\noindent Well-adapted trivialisations $\left(\bgs, u\right)$ splits the tractor bundle of a null-infinity manifold in the same way that trivialisations of $L$ splits the standard tractor bundle of conformal geometry:
\begin{Proposition}
	A well-adapted trivialisation $\left(\bgs , u \right)$ gives an isomorphism
\begin{equation}\label{Introduction: splitting}
 \cT \to L \oplus \left( \pi^*\left(T\gS\right)\otimes L^{-1} \right) \oplus  L^{-1} \oplus \bbR,
\end{equation}
 with the last term in this direct sum corresponding to the degenerate direction of the tractor metric.
\end{Proposition}

We are now in position to state our main result.
\begin{Theorem}{Poincaré operators and Null-normal connections}\mbox{}\label{Introduction: Thrm, null-normal connection}\mbox{}\\
	Let $\left(\MNull \to \gS , \bh_{ab} , \bn^a \right)$ be a null-infinity manifold of dimension $n$.	
	\begin{itemize}
		\item If $n =3$, choices of null-normal tractor connection are in one-to-one correspondence with choices of Poincaré operator.
		\item If $n\geq4$, choices of null-normal tractor connection are in one-to-one correspondence with choices of Poincaré operator inducing the canonical Möbius structure on $\left(\gS , \bh_{AB}\right)$.		
	\end{itemize}
Here ``null-normal'' is a natural condition on the curvature that generalise the normality condition of conformal geometry and will be discussed in details in the bulk of the paper. By ``canonical Möbius structure'' we mean the one given by the Schouten tensor, in particular in dimension $n\geq4$ null-normal connections restricted to the reduced tractor bundle must be the pull-back of the normal Cartan connection on $\left(\gS , \bh_{AB}\right)$.
\end{Theorem}
\noindent These tractor connections can also be directly understood to be Cartan connections for the homogenous space \eqref{Introduction: homogenous model for null-infinity}. In particular flatness gives a local identification with the model.

Altogether, theorem \ref{Introduction: Prop,Asymptotic shear and Poincare operator} and theorem \ref{Introduction: Thrm, null-normal connection} means that, in the physically relevant dimension $n+1=4$, a choice of asymptotic shear (characterising gravitational radiations) precisely correspond to a choice of null-normal tractor connection. 

The reader might fear that this result is too abstract for practical purpose. However, and even thought we believe that one of the main advantage of the formalism presented here is its conceptual clarity, any choice of well-adapted trivialisation will actually enable to work very concretely. As an example, the following proposition can be taken to be a practical definition for null-normal tractor connections (this is here given for $n+1=4$, see the core of the paper for results in other dimensions).
\begin{Proposition}\label{Introduction: Proposition, Null-Normal Cartan connection}
	Let $\left(\MNull \to \gS , \bh_{ab} , \bn^a \right)$ be a $3$-dimensional null-infinity manifold . Then, in the splitting \eqref{Introduction: splitting} given by a well-adapted trivialisation $\left(\bgs, u\right) $ we have
	\begin{equation}
		D_bY^I = \Mtx{ \covD_b & -\gth_{bC}  & 0 & 0  \\
			- \xi_b{}^A &	\covD_b & \gth_b{}^A & 0  \\
			0 &  \xi_{bC} & \covD_b & 0\\
			-\psi_b & -\frac{1}{2} C_{bC} & du_b	& \covD_b	
		} \Mtx{ Y^{+} \\ Y^C \\  Y^{-} \\ Y^u  }
	\end{equation}
	where  $\gth^B_b \from T \MNull \to T\MNull/\bn$ is the canonical projection, $\covD$ is the tensor product of the Levi-Civita connection of $h_{AB} = \bgs^{-2}\bh_{AB}$ with the connection on $L$ given by the scale $\bgs$ and	
		\begin{align}
			C_{bA}&=C_{AB}\;\gth^B_b,& \xi_{bA} &= \left(\frac{1}{2}\parD_u C_{AB} - \frac{R}{4} h_{AB}  \right)\gth^B_b,&  \psi_b &=\frac{1}{4}R \; du_b -\frac{1}{2}\covD^C C_{BC} \;\gth^B_b.
		\end{align}
Here $R$ is the scalar curvature of $h_{AB}$ and ``the asymptotic shear''  $C_{AB}$ is a trace-free symmetric tensor. 
\end{Proposition}

 In particular, by covariant differentiations and contractions of the curvature tensor of this connection one can in principle construct non-trivial invariants in an very explicit way.

\paragraph{The conformal boundary of 3D asymptotically flat space-times}\mbox{}

In the bulk of this article we will also treat the case $n=2$, corresponding to the conformal boundary of a three-dimensional asymptotically flat space-times. The geometry of the null-boundary then also relates very naturally to the physics: a choice of null-normal Cartan connection amounts to a choice of 3D ``Mass aspect'' and ``Angular momentum aspect''. In particular choices of 3D ``Mass aspect'' correspond to a choice of (generalised) Laplace structure as defined in \cite{calderbank_mobius_2006,burstall_conformal_2010}. Even thought things are really close in spirit to the higher dimensional cases, the details are however significantly different and it would take us too far to describe them in this introduction.

\subsection{Organisation of the article}

In order to make this article self-contained we first review the elements of conformal geometry that will be needed in the rest or the article. We then take some time describing the flat model i.e the realisation of null-infinity as an homogeneous space for the Poincaré group. This is essential since tractors will be modelled on this homogeneous space. We then review the ``universal'' or ``weak'' structure of null-infinity in a form that will be suited to describe our results. The ``strong'' (or ``radiative'') structure of null-infinity is related to a choice of Poincaré operators and the geometry of these operator is described, along the way we review essential results on Möbius structures. We then come to the tractor bundle of a null-infinity manifold: we first define the bundle and describe its properties before discussing tractor connections and their equivalence with Poincaré operators. Finally we discuss gravity vacua i.e flat null-normal tractor connections on $\bbR \times \S^{n-1}$ and transition between two such vacua.

\setcounter{tocdepth}{2}
\tableofcontents

\section{Elements of conformal geometry}

We here review standard elements of conformal geometry, see \cite{curry_introduction_2018} for a nice introduction. This will serve mainly to fix the conventions that we will use in the following sections.

\subsection{Conformal manifolds and the bundle of scales}

\subsubsection{Bundle of scales and abstract index conventions}
Let $\M$ be a $d$-dimensional manifold. The bundle of $1$-densities $|\bigwedge|\M$ is the real line bundle associated to the frame bundle of $\M$ with respect to the representation $\mathrm{M} \mapsto |det(\mathrm{M})|^{-1}$ of $GL(d)$. This bundle is always trivial. If $\M$ is orientable, $1$-densities coincide with $n$-forms but on non-orientable manifold $1$-densities (not $d$-forms) are the right type of objects needed for integration. 

The \emph{bundle of scales} is defined by taking the dual of the positive d-root $L = \left(|\bigwedge|X\right)^{-\frac{1}{d}}$. We will take $L^+ \subset L$ to be the bundle of positive scales. In what follows however we will only consider orientable manifolds and make the identification \begin{equation}\label{Conformal geometry: tautological isomorphism for scales}
L = \left(\gL^{d}\; T^* \M \right)^{-\frac{1}{d}} .
\end{equation}
This is only for convenience and all results extend straightforwardly to non-orientable manifolds.

Everywhere in this article we will use an extended version of Penrose's abstract indices notation \cite{penrose_spinors_1984} where weighted-valued sections are represented by bold letters e.g $\ga_a \in \So{T^*\M}$, $\bf \in \So{L^k}$, $\bV^a \in \So{T \M\otimes L}$.

\subsubsection{Conformal manifolds}
We will say that $(\M, \gConf_{ab} )$ is a \emph{conformal manifold} if it is equipped with a non-degenerate symmetric bilinear form $\gConf_{ab}$ with values in $L^2$, i.e $\gConf_{ab}$ is a section of $S^2 T^*\M \otimes L^2 $. The volume form $\bmu_{(\gConf)}{}_{a_1 ... a_{d}}$ is then a section of $\bigwedge{}^d\; T^*\M\otimes L^{d}$ and therefore gives an isomorphism $L^{-d} \to \bigwedge{}^d \; T^*\M$. We will always suppose $\gConf_{ab}$ to be such that this isomorphism is the tautological one \eqref{Conformal geometry: tautological isomorphism for scales} (this can always be achieved by multiplying $\gConf_{ab}$  by the correct function).

Therefore, if $V^a$ is a tangent vector at $x \in\M$, its ``length'' (evaluated with the conformal metric) is $	\sqrt{|\gConf_{ab} V^a V^b|} \in L_x$. One can think of this in physical terms : generically there is no natural scale for length measurements and one needs to choose an arbitrary unit system (centimetres, or inches) to convert a measurement into a number. The situation is similar here, only once we make a choice of scale $\bgs \in L_x$ can we convert our length into a number $\bgs^{-1}\sqrt{|\gConf_{ab} V^a V^b|} \in \bbR$. Conformal geometry corresponds to the curved version of this everyday experience in that scales have to be chosen independently at every point of the manifold $\M$.

\subsection{Choice of scale}

\subsubsection{Choice of trivialisation}

An essential difficulty in conformal geometry is that there is no canonical covariant derivatives on the tangent bundle. One way to deal with this problem is to work in a given trivialisation. This amounts to picking a global (nowhere zero) section $\bgs \in \So{L}$: if $\bf \in \So{L^k}$ is a section of $L^k$, its coordinate is then given by $\bgs^{-k} \bf \in \Co{\M}$. In particular choosing a scale $\bgs \in \So{L}$ amounts to a choice of representative metric $g_{ab}= \bgs^{-2} \gConf_{ab} \in \So{S^2\; T^* \M}$ in the ``conformal class'' $\gConf_{ab}\in \So{S^2\; T^* \M \otimes L^2}$. 

Had we considered another trivialisation $\bgsh \in \So{L}$ such that $\bgsh = \bgs \gO^{-1}$ we would however have obtained
\begin{align}
	\gConf_{ab} &= \bgsh^{2} \ggh_{ab} = \bgs^{2} g_{ab} &\Rightarrow& &\qquad \ggh_{ab} &= \gO^2 g_{ab}\\	
	\bf &= \bgsh^{k} \fh = \bgs^{k} f &\Rightarrow& & \qquad \fh&= \gO^k f.\nonumber
\end{align}

\subsubsection{Weyl connection associated with a choice of scale}

Once we have made a choice of scale $\bgs \in \So{L}$, we can use the Levi-Civita connection $\covD^{(\gs)}$ of $g_{ab} = \bgs^{-2}\bgg_{ab}$ to differentiate tensors. A choice of scale $\bgs \in \So{L}$ also defines a connection $\covD^{(\gs)}$ on $L^k$ as
\begin{equation}\label{Conformal geometry: connection on L}
	\covD^{(\gs)}_{a} \bl = \bgs^{k} d_a\left( \bgs^{-k}\bl \right).
\end{equation}
In what follows, when the context clearly suggests that we have a preferred scale we will simply write $\covD$ for the above connections: E.g if $U^a$, $\go_a$ and $\bl$ are sections of $TX$ , $T^*X$ and $L^k$ respectively, we note their covariant derivatives $\covD_a U^b$, $\covD_a \go_b$ and $\covD_a \bl$.

However, since we are interested by conformal invariants, we will have to check at each step that our statements do not depend on this choice of trivialisation, i.e that we could have used another section $\bgsh$ and obtain equivalent results. If $\bgsh = \gO^{-1} \bgs$ is any other scale and $\covDh$ the associated connection, we have the transformation rules
\begin{align}\label{Conformal geometry: Connection transformation rules}
\covDh_a \bl &= \covD_a \bl + k \gU_a \bl \nonumber\\
\covDh_a U^b &= \covD_a U^b +\gU_a U^b - U_a \gU^b + U^c \gU_c \;\gd^b_a \\
\covDh_a \go_b &= \covD_a \ga_b -\gU_a \go_b - \go_a \gU_b + \gU^c \go_c \; g_{ab} \nonumber
\end{align}
where $\gU_a = \gO^{-1} \; d_a\gO$.

\subsubsection{Curvature tensors}\label{ss: Curvature tensors}

Whenever one has a preferred scale $\bgs$ and therefore a preferred metric $g_{ab} = \bgs^{-2}\bgg_{ab}$ it will be useful to define the Schouten tensor $P_{ab}$ as
\begin{align}\label{Conformal geometry: Schouten tensor}
P_{ab} &= \frac{1}{d-2}\left( R_{ab} -\frac{R}{2(d-1)} h_{ab} \right)\\
       &= \frac{1}{d-2}  R_{ab}\big|_{0} +\frac{R}{2d(d-1)} g_{ab} \nonumber
\end{align}
where $R_{ab}$ is the Ricci tensor and $R$ the Ricci scalar. Here and everywhere in this article $|_0$ will also indicates the ``trace-free part of'' the tensor. We also define the trace of the Schouten tensor as
\begin{equation}\label{Conformal geometry: Schouten tensor trace}
P = \frac{1}{2(d-1)} R.
\end{equation}
It will also be important for us that, even thought the Schouten tensor is only defined in dimension $d \geq 3$, its trace makes sense in any dimensions $d \geq 2$.
If $\bgsh= \gO^{-1} \bgs$ is any other scale, we have the transformation rules
\begin{equation}\label{Conformal geometry: Schouten tensor transformation rules}
\Ph_{ab} = P_{ab} - \covD_a \gU_b + \gU_a \gU_b -\frac{1}{2} \gU^2 g_{ab}.
\end{equation}

In any dimension $d \geq 3$, the Weyl tensor is obtained as
\begin{equation}
W^a{}_{bcd} = R^a{}_{bcd} - 2 P^a{}_{[c} \; g_{d]b} - 2 g^a{}_{[c} \; P_{d]b}
\end{equation}
and is conformally invariant
\begin{equation}
\Wh^a{}_{bcd} = W^a{}_{bcd}.
\end{equation}

Finally the Cotton tensor is defined in any dimension  $d \geq 3$ as
\begin{equation}\label{Conformal geometry: Cotton tensor, definition}
C_{ab}{}^c = 2 \covD_{[a} P_{b]}{}^c
\end{equation}
it follows from Bianchi identity $\covD_{[a} R_{bc]de} =0$ that we have the identities
\begin{align}\label{Conformal geometry: Cotton tensor, identities}
\left(d-3\right) C_{ab}{}^c &= \covD_{d} W^{dc}{}_{ab}, &
C_{ba}{}^b &= \covD_b P_a{}^b - \covD_a P =0.
\end{align}

\section{The flat model}\label{s: The flat model}

In this section we review in full details the conformal compactification of Minkowski space. We do so in such a way that the homogeneous space structure is manifest: The essential point here is to emphasis that flat null-infinity naturally is an homogenous space for the Poincaré group. Readers which are familiar with the homogenous space structure of null-infinity can therefore safely skip this section.

 This description of flat null-infinity will serve as the model for the geometry discussed in the rest of this article. In particular, in this presentation flat null-infinity is naturally embedded in a null hyper-surface $I^{\perp}$ of $R^{(n+1,2)}$ (where $n$ is the dimension of flat null-infinity as a manifold). The tractor bundle of a null-infinity manifold will be modelled on this null hyper-surface.

\vspace{0,5cm}
\noindent Here $d=n+1$ is the dimension of Minkowski space.

\subsection{The conformal compactification of Minkowski space}

\subsubsection{Minkowski space}

Let us consider $\bbR^{(d,2)}$ as a vector space equipped with the flat metric $q$ of signature $\left(d,2\right)$. We pick a basis on this vector space such that any point $V \in \bbR^{(d,2)}$ is written as
\begin{equation}\label{The flat model: R^{d+2} coordinates}
	V = \begin{pmatrix}
		U \\ V^{i} \\ W
	\end{pmatrix}
\in \bbR^{d+2}
\end{equation}
with $Y^i \in \bbR^{d}$ and the inner-product $q$ is given by 
\begin{equation}\label{The flat model: metric on R^d+2}
	q = 2\; dU d W+ \eta_{ij}\; dV^i dV^j
\end{equation}
where $\eta_{ij}$ is the flat metric of signature $\left(d-1,1\right)$. 

Let $L^+$ be the null cone in $\bbR^{(d,2)}$ passing through the origin
\begin{align}
	L^+&=\left\{V \in \bbR^{(d,2)}|\; V^2=0 \right\}.
\end{align}
Let the ``infinity tractor'' $I \in \bbR^{(d,2)}$ be a choice of null vector, $I^2=0$ and let us adapt our coordinates such that $I$ is given by
\begin{equation}
	I = \begin{pmatrix}
		0 \\
		0\\
		1
	\end{pmatrix}.
\end{equation}
Minkowski space $\MRiem^{(d-1,1)}$ can then be isometrically embedded in $L^{+}$ as
\begin{equation}
	\MRiem^{(d-1,1)} = \left\{ V \in L^{+} \;|\; V.I =1 \right\}.
\end{equation}
In coordinates, this is given by
\begin{equation}\label{The flat model: Minkowski isometric embeding}
	\begin{array}{ccc}
		\bbR^{(d-1,1)} & \to & 	\MRiem^{(d-1,1)} \\
		x^i & \mapsto & \begin{pmatrix}
			1 \\
			x^i\\
			-\frac{1}{2}|x|^2
		\end{pmatrix}.		
	\end{array}
\end{equation}
The stabiliser of $I$ in $\SO(2,d)$ is isomorphic to the Poincaré group $\Iso(d-1 , 1) = \bbR^{d}~\rtimes~\SO\left(d-1,1\right)$ and can be parametrised as
\begin{equation}\label{The flat model: Iso(d-1,1) action}
	\begin{pmatrix}
		1 & 0 & 0 \\
		-\;r^i & m^i{}_j & 0 \\
		\frac{1}{2} r^k r_k & r_k m^k{}_j & 1
	\end{pmatrix}
\end{equation}
where $r^i$ is in $\bbR^{d}$, $m^i{}_j$ is in $\SO\left(d-1,1\right)$ and lower case Latin indices are raised and lowered with the flat metric $\eta_{ij}$. The action of \eqref{The flat model: Iso(d-1,1) action} on \eqref{The flat model: Minkowski isometric embeding} then gives the usual action of the Poincaré group on Minkowski space.

\subsubsection{Conformal compactification}

Since $\bbR^+$ acts on $L^+ \subset \bbR^{(d,2)}$ by multiplication, the null cone $L^+$ is the total space of a $\bbR^+$-principal bundle
 \begin{equation}
L^+ \to \MConf^{(d-1,1)}
\end{equation} over the projectivised null cone
 \begin{equation}
	\MConf^{(d-1,1)} \coloneqq L^+/\bbR^+ \simeq \S^{d-1} \times \S^1.
\end{equation}
Each section of $L^+ \to \MConf^{(d-1,1)}$ gives a different metric on $\MConf^{(d-1,1)}$, obtained by pull-back of the metric \eqref{The flat model: metric on R^d+2}. All these metrics are conformally related to the round metric on $\S^{d-1} \times \S^1$ and consequently $\MConf^{(d-1,1)}=\S^{n-1} \times \S^1$ comes equipped with the conformal round metric $\gConf^{(n-1,1)}$. Finally, $L^+ \to \MConf^{(d-1,1)}$ identifies with the bundles of positive scales over $\MConf^{(d-1,1)}$.

It follows from this discussion that we have a conformal embedding of Minkowski space into the projectivised null-cone $\MConf^{(d-1,1)}$,
\begin{equation}\label{The flat model: Minkowski embeding}
	\begin{array}{ccc}
		\bbR^{(d-1,1)} & \to & 	\MConf^{(d-1,1)} \\
		x^i & \mapsto & \begin{bmatrix}
			1 \\
			x^i\\
			-\frac{1}{2}|x|^2
		\end{bmatrix}
	\end{array}
\end{equation}
(the squared brackets indicate homogeneous coordinates on the projective space).

\subsubsection{Bondi-Sachs coordinates in a neighbourhood of null-infinity}

We now take `` future null infinity'' (or simply ``null-infinity'' in what follows) $\scrI_{(d-1,1)}$ to be
\begin{align}\label{The flat model: null-infinity definition}
	\scrI_{(d-1,1)} &= \left\{ V \in \MConf^{(d-1,1)} |\; V.I = 0 \; \text{and s.t} \; V^i\; \text{is future directed}\right\}.
\end{align}
It follows from the preceding sections that $\scrI_{(d-1,1)}$ is part of the boundary of (the conformal compactification of) Minkowski space in the null cone $L^+ \subset \bbR^{(d,2)}$. Bondi-Sachs coordinates are convenient to describe a neighbourhood of $\scrI_{(d-1,1)}$ in $L^+$:
\begin{equation}\label{The flat model: BMS coordinates}
	 \begin{bmatrix}
	\bgs\begin{pmatrix}
		\gO \\ 1+ u\gO\\ \gth^a \\ u +\frac{1}{2}u^2\gO
		\end{pmatrix} 
	\end{bmatrix}
	\in \MConf^{(d-1,1)}
\end{equation}
where $\gO\in \bbR^+$, $u\in \bbR$ while $\gth^a \in \bbR^{d-1}$ parametrizes the $d-2$ sphere, $\sum_{a =1}^{d-1} (\gth^a)^2 = 1$. In this chart, null-infinity $\scrI_{(d-1,1)}$ is given by $\gO=0$. The conformal metric $\gConf^{(d-1,1)}$ induced on $\MConf^{(d-1,1)}$ from the ambient metric \eqref{The flat model: metric on R^d+2} then reads,
\begin{align}\label{The flat model: conformal BMS metric}
	\gConf^{(d-1,1)} = \bgs^2 \Big( 2d\gO du -\gO^2 du^2 + d\gth^2  \big).
\end{align}
In particular, the isometric embedding of Minkowski space is given by $\bgs = \gO^{-1}$.

\subsection{Homogenous space structure of null-infinity}\label{ss: Homogenous space structure of null-infinity}

From now-one we take $d=n+1$. This is such that, everywhere in this article $n$ is the dimension of a null-infinity manifold.

\subsubsection{Adapted coordinates on \texorpdfstring{$I^{\perp}$}{I^perp}}

In order to match our tractor notation, it will be useful to think of $\scrI_{(n,1)}$ as a sub-manifold of the null hyper-surface $I^{\perp} = \{V \in\bbR^{(d,2)} \big| V.I=0 \}$: The tractor bundle of null-infinity will indeed be a vector bundle whose fibres are modelled on $I^{\perp}$.

  As a null hyper-surface, $I^{\perp} = \bbR^{n+2}$ is fibred by null-lines and therefore the total space of a line bundle $I^{\perp} \to \bbR^{n+1}$ over $\bbR^{n+1}$.  We will pick adapted coordinates on $I^{\perp}$: If $V^I$ is in $I^{\perp}$, these are such that
\begin{equation}\label{The flat model: Iperp coordinates}
V^I = \Mtx{V^i \\ u} \in I^{\perp}=\bbR^{n+2}
\end{equation}
where $V^i \in \bbR^{n+1}$ parametrize the space of null lines foliating $I^{\perp}$. The induced metric (from \eqref{The flat model: metric on R^d+2}) is then the degenerate metric
\begin{equation}\label{The flat model: Iperp metric}
	q\big|_{I^{\perp}} = \eta_{ij} dV^i dV^j,
\end{equation}
with degenerate direction spanned by the preferred vertical vector field, 
\begin{equation}\label{The flat model: Iperp vector field}
	n = \parD_{u}.
\end{equation}

We now come to null-infinity $\scrI_{(n,1)}$ itself. It is embedded in $I^{\perp}$ as 
\begin{equation}
	\scrI_{(n,1)} =\Quotient{ \left\{ V \in I^{\perp} \big| V^2 = 0  \; \text{and s.t} \; V^i\; \text{is future directed} \right\} }{ \bbR^+}.
\end{equation}
We have the following homogeneous coordinates on  $\scrI_{(n,1)} = \bbR \times \S^{n-1}$,
\begin{equation}\label{The flat model: homogenous coordinates on scrI1}
	 \begin{bmatrix}
	 	 \bgs\begin{pmatrix}
	 	1 \\ \gth^a \\ u
	 		\end{pmatrix}
	 \end{bmatrix}
\end{equation}
where $u \in \bbR$ and $\gth^a \in \bbR^{(n)}$ parametrise the $(n-1)$-sphere, $ \sum^{n}_{a = 1 }(\gth^a)^2 = 1 $. These amounts to taking $\gO =0$ in the Bondi-Sachs coordinates \eqref{The flat model: BMS coordinates}.

\subsubsection{Weak (or universal) structure of \texorpdfstring{$\scrI_{(n,1)}$}{scrI}}

As a sub-manifold of $I^{\perp}$, null-infinity $\scrI_{(n,1)}$ has the structure of a line bundle \begin{equation}
\scrI_{(n,1)} \to \S^{n-1} 
\end{equation}(with adapted coordinates $\left( u , \gth^a \right) $). 

It also inherits a degenerate conformal structure $\gConf_{\scrI} \in \So{S^2 T^*\scrI \otimes L^2_{\scrI}}$ (obtained by the restriction of \eqref{The flat model: Iperp metric} ). In our set of coordinates,
\begin{equation}
	\gConf_{\scrI} = \bgs^2\; d\gth^2
\end{equation}
This metric then descends to the conformally round metric on $\S^{n-1}$ and, all in all, $\scrI$ has the structure of a line bundle $\scrI \to \S^{n-1}$ over the $(n-1)$-dimensional conformal round sphere.

The structure inherited by $\scrI_{(n,1)}$ as a sub-manifold of $\MConf^{(n,1)}$ does not stop here. Recall that $I^{\perp}$ also comes with a preferred vector field \eqref{The flat model: Iperp vector field}, it induces on $\scrI_{(n,1)}$ a weighted vertical vector field $\bn_{\scrI} \in \So{ T\scrI \otimes L^{-1}_{\scrI}}$. In our adapted coordinate system,
\begin{equation}
	\bn_{\scrI} = \bgs^{-1}\; \parD_{u}.
\end{equation} 
Altogether $\left( \scrI_{(n,1)} \to \S^{n-1}, 	\gConf_{\scrI}, 	\bn_{\scrI} \right) $ form the weak (or universal) structure of $\scrI_{(n,1)}$.

\subsubsection{Homogeneous space structure}

The Poincaré group $\Iso\left(n,1\right)$ acts linearly on $I^{\perp}$ as
\begin{equation}\label{The flat model: Iso(d-1,1) action on scrI}
	\begin{pmatrix}
		m^i{}_j & 0 \\
		r_k m^k{}_j & 1
	\end{pmatrix} 
\end{equation}
with $r^i \in \bbR^{n+1}$ and $m^i{}_j \in \SO\left(n,1\right)$. 
Accordingly, this induces an action of the Poincaré group on null-infinity. This action is transitive and therefore $\scrI_{(n,1)}$ naturally is an homogeneous space
\begin{equation}\label{The flat model: scrI as an homogeneous space}
	\scrI_{(n,1)} = \Quotient{\Iso\left(n,1\right)}{Carr\left(n\right)  \rtimes \bbR}
\end{equation} 
Where we denote by $Carr\left(n\right)$ the subgroup of $\SO\left(n+1,2\right)$ stabilizing two (non parallel) orthogonal null vectors, we will soon see that it isomorphic to the Carroll group from \cite{levy-leblond_nouvelle_1965,duval_carroll_2014}

In order to have a concrete realisation of $Carr\left(n\right)  \rtimes \bbR \subset \Iso\left(n,1\right)$ we take null coordinates on $\bbR^{(n,1)}$, i.e we pick a basis of $\bbR^{(n,1)}$ such that
\begin{equation}\label{The flat model: null coordinates on R^d}
	V^i = 
		\begin{pmatrix}
		V^{+}\\
		V^A \\
		V^{-}		
	\end{pmatrix} 
\in \bbR^{n+1}
\end{equation}
where $V^A \in \bbR^{n-1}$ and the metric reads
\begin{equation}
	\eta_{ij} dV^i dV^j = 2dV^{+}dV^{-} + h_{AB} dV^A dV^B
\end{equation}
(with $h_{AB}$ the flat metric on $\bbR^{n-1}$).
Together with our coordinates \eqref{The flat model: Iperp coordinates} this parametrizes $I^{\perp}$ as:
\begin{equation}\label{The flat model: null coordinates on Iperp}
	V^I = 
	\begin{pmatrix}
		V^{+}\\
		V^A \\
		V^{-}\\
		V^{u}	
	\end{pmatrix} 
	\in I^{\perp}=\bbR^{n+2}.
\end{equation}
Accordingly, we have the stereographic coordinates on $\scrI_{(n,1)}$:
\begin{equation}\label{The flat model: stereographic coordinates in scrI}
\begin{bmatrix} 1 \\ y^A \\ -\frac{1}{2} y^2 \\ u \end{bmatrix} \in \scrI_{(n,1)}
\end{equation}
where $u\in \bbR$ is the coordinate along the fibres of $\scrI_{(n)} \to \S^{n-1}$ and $y^A \in \bbR^{(n-1)}$ are stereographic coordinates on $\S^{n-1}$.

Let $X$ be a null vector in $I^{\perp}$ such that $X$ and $I$ are orthogonal (but not proportional). We can always adapt our set of coordinates \eqref{The flat model: null coordinates on Iperp} such that
\begin{equation}
	X^I = 	\begin{pmatrix}
		0\\
		0 \\
		1\\
	0		
	\end{pmatrix} 
\end{equation}
Then $Carr\left(n\right)  \rtimes \bbR$, the subgroup of $\Iso\left(n,1\right)$ stabilizing the \underline{line} generated by $X$, is parametrised as
\begin{equation}\label{The flat model: Carr(n) action on scrI}
	\begin{pmatrix}
		\gl & 0 & 0 &0 \\ 
		- t^A & m^A{}_B & 0 & 0 \\
		\gl^{-1} \; \frac{1}{2}\; t^C t_C & \gl^{-1}\;t_C m^C{}_B & \gl^{-1} & 0 \\
		f & r_{C} m^C{}_B & 0 & 1
	\end{pmatrix}
\end{equation}
where $f$ and $\gl$ are in  $\bbR$, $r^A$ and $t^A$ are elements of $\bbR^{n-1}$ and $m^A{}_B$ is a matrix in $\SO\left(n-1\right)$. In particular $Carr\left(n\right)=\bbR^n\rtimes \Iso\left(n-1\right) $, the subgroup of $\SO\left(n+1,2\right)$ stabilising two orthogonal (non-parallel) null vectors is obtained by taking $\gl =1$. In this form one can also see that (taking $\gl=1$) the action of \eqref{The flat model: Carr(n) action on scrI} on the quotient $I^{\perp}/X$ directly identifies with the representation of the Carroll group used in \cite{duval_carroll_2014}. 

We close this subsection with a few remarks related to the construction that will be presented in the rest of this article:
The tractor bundle over a null-infinity manifold is naturally an associated bundle for $Carr\left(n\right)  \rtimes \bbR$ in the representation given by \eqref{The flat model: Carr(n) action on scrI}: 
Fibres of the tractor bundle are modelled on $I^{\perp}$ and ``well-adapted trivialisations'' will be shown to give coordinates of the form \eqref{The flat model: null coordinates on Iperp}.  Finally, null-normal tractor connections will we be a class of Cartan connections modelled on \eqref{The flat model: scrI as an homogeneous space}.

\subsubsection{Good cuts}

By construction the action of the Poincaré group \eqref{The flat model: Iso(d-1,1) action on scrI} on $\scrI_{(n,1)}$ must preserve the (weak) structure of null-infinity: $\left( \scrI_{(n,1)} \to \S^{n-1}, \gConf_{\scrI}, \bn_{\scrI} \right)$. This is because this structure was obtained on $\scrI_{(n,1)}$ solely from its realisation as a sub-manifold and quotient in $\bbR^{(n+1,2)}$. 

However the subgroup of diffeomorphism preserving $\left( \scrI_{(n,1)} \to \S^{n-1}, 	\gConf_{\scrI}, 	\bn_{\scrI} \right) $ is well-known \cite{ashtekar_geometry_2015, duval_conformal_2014, duval_conformal_2014-1} to be the (infinite-dimensional) BMS-group and therefore the fact that null-infinity is an homogeneous space for the Poincaré group (rather than the BMS-group) points to the fact that is is equipped with more structure.

This extra structure is a set of ``good cuts'' \cite{newman_heaven_1976,hansen_r._o._metric_1978,ko_theory_1981,adamo_generalized_2010,adamo_null_2012} i.e a preferred set $H$ of sections (or cuts) of $\scrI_{(n,1)} \to \S^{n-1}$. It will be convenient to represent these in stereographic coordinates:
\begin{equation}
	y^A \mapsto \begin{bmatrix} 1 \\ y^A \\ -\frac{1}{2} y^2 \\ u\left(y\right) \end{bmatrix}.
\end{equation}
The section $ y^a \mapsto u\left(y\right) =0$ is then an example of good-cut and all others are obtained by the action of translations (realised by taking $m^i{}_j =\gd^i{}_j$ in \eqref{The flat model: Iso(d-1,1) action on scrI}): If $r^i=\left(r^{+} , r^A , r^{-}\right)$ is a vector of $\bbR^{(n,1)}$ expressed in null coordinates \eqref{The flat model: null coordinates on R^d} it sends the cuts  $ y^A \mapsto u =0$ to
\begin{equation}
	y^A \mapsto u = r^{-} + r^A y_A + r^{+} \left(- \frac{1}{2}y^2 \right).
\end{equation}

Good cuts have the following interpretation: The image of $y^A \mapsto u\left(y\right) =0$ is the intersection of null-infinity $\MNull_{(n,1)}$ with the null cone emanating from the origin of Minkowski space $M^{(n,1)}$ (this makes sense because null-conformal geodesics are conformally invariants). As we act with the group of translations we span all the points in Minkowski space-times and good-cuts correspond to all the possible intersections of the associated null-cones.

Alternatively, the good-cuts can be obtained as the space of solutions to the ``good-cuts equations'' on $\S^{(n-1)}$. In stereographic coordinates, these reads
\begin{equation}
	\covD_A \covD_B u\big|_{0}=0.
\end{equation}

Zeros of Poincaré operators will generalise these good cuts for a generic null-infinity manifold.

\section{Universal (or weak) structure of null-infinity}

We here review from \cite{geroch_asymptotic_1977,ashtekar_radiative_1981, ashtekar_a._symplectic_1981,ashtekar_symplectic_1982,ashtekar_asymptotic_1987,ashtekar_geometry_2015,ashtekar_null_2018,duval_conformal_2014,duval_conformal_2014-1} the ``weak''\footnote{We here use the terminology ``weak'' rather than ``universal'': The terminology ``universal structure'' should probably be kept to refer to the particular case of a ``weak structure'' over the conformal sphere. The weak/strong terminology also nicely fits with the one from \cite{duval_conformal_2014,duval_conformal_2014-1,duval_carroll_2014}.  } structure of null-infinity (and its symmetry group). We do so with an insistence on the (degenerate) conformal geometry of this structure. We also derive some elementary results that will be useful later on.

\subsection{Null-infinity manifold}
\subsubsection{Conformal Carroll manifolds}

Let $\MNull$ be a $n$-dimensional manifold and $L \to \MNull$ its bundle of scales. We will first need the notion of conformal Carroll manifolds (taken from \cite{duval_conformal_2014, duval_conformal_2014-1})
\begin{Definition}{Conformal Carroll manifolds}\mbox{}\\
A conformal Carroll structure $\left(\scrI, \hConf_{ab}, \bn^a \right)$ on a $n$-dimensional manifold $\scrI$ consists of
\begin{itemize}
	\item a degenerate conformal metric $\hConf_{ab} \in \So{S^2T^*\scrI \otimes L^2}$ with one dimensional kernel
	\item a weighted vector field $\bn^a \in \So{T\MNull \otimes L^{-1}}$ generating this kernel, $\bn^a \hConf_{ab} =0$.
\end{itemize}
A conformal Carroll manifold is a manifold equipped with a conformal Carroll structure.
\end{Definition}

\subsubsection{Vertical connection on L}\label{sss: vertical connection on L}

Let $\left(\MNull, \hConf_{ab}, \bn^a \right)$ be a conformal Carroll manifold. Let $V =\bbR\bn \subset T\MNull$ be the $1$-dimensional distribution (the ``vertical distribution'') given by the kernel of $\hConf_{ab}$. Let us consider the $(n-1)$-dimensional vector bundle $T\MNull / V$ obtained by taking the quotient, we will use upper-case latin indices $A,B, etc$ as abstract indices on this bundle. By construction, $\hConf_{ab}$ induces a non-degenerate inner-product $\bh_{AB}$ on $T\MNull / V$. Let $\left(T\MNull / V \right)^*$ be the dual of $T\MNull / V$, the ``volume form'' $\bmu_{(\hConf)}{}_{A_{1}...A_{n-1}}$ of $\hConf_{AB}$ then is a section of $ \bigwedge^{n-1}\left(T\MNull / V \right)^*  \otimes L^{n-1}$.

From this few remarks we can define a ``vertical connection on $L$''
 \begin{equation}
	\covD \from \So{V \otimes L} \to \So{L}.
\end{equation}
In order to construct this connection explicitly, let us take $\bl \in \So{L}$ a section of $L$ and note that $ \bl \left( \bmu_{(\hConf)}\right)^{-\frac{1}{n-1}}$ is a section of $ \left(\bigwedge^{n-1}\left(T\MNull / V \right)^* \right)^{-\frac{1}{n-1}} $. If $v \in \So{V}$, one defines the covariant derivative $\covD_v \bl$ through
\begin{equation} 
\cL_v \left ( \bl \left( \bmu_{(\hConf)}\right)^{-\frac{1}{n-1}} \right ) =  \covD_v \bl \; \left( \bmu_{(\hConf)}\right)^{-\frac{1}{n-1}}
\end{equation}
and one can check that it satisfies the property of a connection (see also below for a local coordinate description).

Since $\bn^a \in \So{V\otimes L^{-1}}$, one has $\covD_{\bn}\bl \in \Co{\MNull}$ and one can define sections with constant vertical derivative to be such that $\covD_{\bn}\bl = cst$.
\begin{Definition}
	We will say that a section $\bl \in \So{L}$ has constant vertical derivative $k\in \bbR$ and write $\bl \in \CSoL{k}$ if and only if $\covD_{\bn} \bl = k$,
	\begin{equation}\label{Weak structure of null-infinity: covariantly constant sections}
\bl \in \CSoL{k} \qquad \Leftrightarrow \qquad \covD_{\bn} \bl = k.
	\end{equation}
\end{Definition}

Sections with zero vertical derivative act as distinguished trivialisation of $L$ and for this reasons, will be very important: if $\bgs \in \CSoL{0}$ and $\bl \in \So{L}$ with $l= \bgs^{-1}\bl$, we must have
\begin{equation}
	\covD_{\bn} \bl = \bgs\; \bn^a dl_a,
\end{equation}
and in particular $\bl \in \CSoL{k}$ if and only if 
\begin{equation}
	\bgs\; \bn^a dl_a = k.
\end{equation}

 If $\MNull \to \gS$ is a fibre bundle, scales in $\CSoL{0}$ are very natural for another reason: they correspond to scales of the base.
\begin{Proposition}\label{Weak structure of null-infinity: Prop, covariantly constant sections isomorphism}
Let $\MNull$ be the total space of a fibre bundle $\MNull \to \gS$ with fibres given by the integral lines of $\bn$ there is then a canonical isomorphism
 \begin{equation}\label{Weak structure of null-infinity: covariantly constant sections isomorphism}
\CSoL{0} \simeq \So{L_{\gS}}
\end{equation}
\end{Proposition}
\begin{proof} 	
By definition, the map $\bl \mapsto \bl^{-(n-1)} \bmu_{(\hConf)}$ takes sections with zero vertical derivatives to sections of $\pi^* \left( \bigwedge^{n-1}T^*\gS \right))$ whose Lie derivative in the vertical direction vanishes and thus with sections which are the pull-back of sections of $\bigwedge^{n-1}T^*\gS$. Since we have the tautological identification $L_{\gS} = \left( \bigwedge^{n-1}T^*\gS \right)^{-\frac{1}{n-1}}$ this gives the required isomorphism.
\end{proof} 

\paragraph{Local coordinates description}\mbox{}

For the convenience of the reader, we here provide a local coordinate description of the above construction. This can be safely skipped if one feels at ease with the abstract description. 

Let us pick up local coordinates on $\MNull$, $\left(u , y^A \right) \in \bbR \times \bbR^{n-1}$, together with a section $\bgs \in \So{L}$. We take these coordinates to be adapted to the conformal Carroll structure i.e
\begin{equation}\label{Weak structure of null-infinity: adapted coordinates conditions}
	\bgs^{-2} \hConf_{ab} = h_{AB}\; dy^A dy^B , \qquad \bgs \bn^a = \parD_u,
\end{equation}
\begin{equation}
\bgs^{-(n-1)} \bmu(\hConf){}_{A_{1} ... A_{n-1}}	= \sqrt{det(h)} \; d^{n-1} y
\end{equation}
where $h_{AB}$ is a symmetric $(n-1)$-dimensional matrix..

Consider $\bl \in \So{L}$ a section of $L$ and $l = \bgs^{-1} \bl \in \cC^{\infty}\left[\MNull\right]$ its representative in the trivialisation $\bgs$. Its vertical derivative reads
\begin{equation}\label{Weak structure of null-infinity: covariant derivative in coordinates}
	\covD_{\bn} \bl = \left( \parD_u  -\frac{1}{n-1} \frac{\parD_u \sqrt{det(h)}}{\sqrt{det(h)}}  \right) l.
\end{equation}
In particular, $\bgs$ is covariantly constant (i.e $\bgs \in \CSoL{0}$), if and only if $\parD_u det(h) =0$. 

\subsubsection{Null infinity manifolds}

Let $\left( \MNull, \hConf_{ab} , \bn^a \right)$ be a conformal Carroll manifolds. Making use of the Bondi-Sachs expansion one could try to take this manifold as some initial data for an asymptotically flat space-times. However, not all conformal Carroll manifolds can be the conformal boundary of an asymptotically flat space-time, rather Einstein equations (to order minus one in the expansion in the boundary defining function $\gO$) impose that the degenerate conformal metric must be independent of the ``u'' coordinates. In this context it is also very natural, see \cite{ashtekar_geometry_2015}, to require that $\MNull$ is foliated by the null lines generated by $\bn^a$ such that $\MNull \to \gS$ is a fibre bundle (with $\gS$ the space of null lines). This considerations justify to introduce the following definition:

\begin{Definition}{Null-infinity manifold}\mbox{}\label{Weak structure of null-infinity: Def, null-infinity manifold}\\
	A \emph{weak null-infinity structure} $\left(\MNull \to \gS , \hConf_{ab} ,\bn^a \right)$ consists of a fibre bundle $\MNull \to \gS$ over an $(n-1)$-dimensional manifold $\gS$ together with a conformal Carroll structure $\left(\MNull, \hConf_{ab} ,\bn^a \right)$ on the total space satisfying the compatibility conditions:
	\begin{itemize}
		\item $\bn^a$ is tangent to the vertical direction of $\MNull \to \scrI$, $d\pi^A_b \;\bn^b=0 $
		\item $\hConf_{ab}$ is the pull-back of a conformal metric $\hConf_{AB}$ on $\gS$, $\hConf_{ab} = \pi^*\left(\hConf_{AB}\right)$.
	\end{itemize}
	A null-infinity manifold is a line bundle equipped with a weak null-infinity structure.
	
	 In the following we will also suppose that $\MNull \to \gS$ is trivial and $\gS$ is orientable. (This is just for convenience as most results are local.) 
	As a convention, we will use lower-case Latin indices of the beginning of the alphabet as abstract indices on $\MNull$ and upper-case Latin indices of the beginning of the alphabet as abstract indices on $\gS$. 
\end{Definition}
Some remarks are in order about the second point in this definition. Strictly speaking $\pi^* \hConf_{AB}$ is a section of $S^2 T^*\MNull \otimes (\pi^* L_{\gS})^2$ however, making use of the isomorphism given by proposition \ref{Weak structure of null-infinity: Prop, covariantly constant sections isomorphism}, we can turn it into a section of $S^2 T^*\MNull \otimes L^2$ with zero vertical derivative. In other words, the second compatibility condition can be rephrased as $\covD_{\bn} \hConf_{ab} = 0$ where $\covD_{\bn}$ stands for the combined action of Lie derivative and the vertical covariant derivative on $L$. Finally, let us note for concreteness that if $\bgs \in \So{L}$ is a scale and $\left(u, y^A\right) \in \bbR \times \bbR^{n-1}$ are local coordinates on (an open subset of) $\MNull \to \gS$ satisfying \eqref{Weak structure of null-infinity: adapted coordinates conditions} then the second condition in definition \ref{Weak structure of null-infinity: Def, null-infinity manifold} reads 
\begin{equation}
\parD_u h_{AB} = \frac{2}{n-1} \frac{\parD_u \sqrt{det(h)}}{\sqrt{det(h)}} h_{AB}.
\end{equation}
Which is the constraint familiar from the Bondi-Sachs formalism see e.g \cite{barnich_aspects_2010}.

\subsection{Well-adapted trivialisations}\label{ss: Well adapted trivialisations}

From now on we will make a systematic use of the isomorphism \eqref{Weak structure of null-infinity: covariantly constant sections isomorphism} and identify $\bgs \in \So{L_{\gS}}$, its pull back $\bgs \in \So{\pi^* L_{\gS}}$ and the section $\bgs \in \CSoL{0}$ given by proposition \ref{Weak structure of null-infinity: Prop, covariantly constant sections isomorphism}.

\subsubsection{Well-adapted trivialisations}\label{sss: Well adapted trivialisations}
\begin{Definition}{Well-adapted trivialisations}\label{Weak structure of null-infinity: Def, well-adapted trivialisation}\mbox{}\\
	We will say that $\left(\bgs , u \right) \in \So{L} \times \Co{\MNull}$ is a \emph{well-adapted trivialisation} for a null-infinity manifold $\left( \MNull \xto{\pi} \gS, \bh_{ab} , \bn^a \right)$ if and only if $\bgs$ is vertically constant, $\bgs \in \CSoL{0} \simeq \So{L_{\gS}}$, and $u \in \cC^{\infty}\left(\MNull\right)$ defines a trivialisation
	\begin{equation}
		\left(u, \pi \right) \from
		\begin{array}{ccc}
			\MNull & \to & \bbR \times \gS
		\end{array}
	\end{equation}
	such that
	\begin{equation}
		\bgs\,\bn^a du_a= 1 .
	\end{equation}
	In particular both $\bgs$ and $u$ must be globally\footnote{We suppose that $\MNull \to \gS$ is trivial but all results extends to non-trivial bundle by taking $u$ to be a local trivialisation.} defined on $\MNull$.
\end{Definition}
Well-adapted trivialisation in the above sense will be very useful for they allow for a direct comparison with the literature written in the BMS formalism: there, a conjoint choice of representative $h_{AB} = \bgs^{-2}\bh_{AB}$ and trivialisation $u\from \MNull \to \bbR$ is always assumed. Accordingly, when a choice of well-adapted trivialisation $\left( \bgs , u\right)$ is clearly understood we will abuse notation and write $\parD_u \coloneqq \bgs \LieD_{\bn}$ for the Lie derivative in the vertical direction. 
For the reason that we just stated, definition \ref{Weak structure of null-infinity: Def, well-adapted trivialisation} will be our working definition, it will however be useful to keep in mind the following equivalent definition:
\begin{Proposition}\label{Weak structure of null-infinity: Prop, equivalence of well-adapted trivialisation1}
	A choice of well-adapted trivialisation $\left( \bgs, u \right)$ is equivalent to a choice of sections $\left( \bl_{0}, \bl_{1}\right)$ of $\CSoL{0} \oplus \CSoL{1}$. The isomorphism is given by
	\begin{equation}
	\left( \bl_{0}, \bl_{1}\right) \mapsto \left( \bgs = \bl_{0} , u = \bgs^{-1} \bl_{1} \right).
	\end{equation}
\end{Proposition}
\begin{proof}
	As discussed in the previous section, vertically covariantly constant section $\bl_{0} \in \CSoL{0}$ are canonically identified  $\bl_{0} = \pi^* \bgs$ with section $\pi^*\bgs$ of $\pi^* L_{\gS}$ which are the pull-back of a section of $L_{\gS}$. On the other hand, if $\bl_{1}$ is a section of $L$ and $u = \bgs^{-1}\bl_{1}$ its coordinates then $\bl_{1} \in \CSoL{1}$ if and only if $\covD_{\bn} \bl_{1} = \bgs\; \bn^a  du_a =1$.
\end{proof}

An elementary fact that will play a crucial role for us is the equivalence between elements of $\CSoL{1}$ and sections of $\MNull \to \gS$.
\begin{Proposition}\label{Weak structure of null-infinity: Prop, equivalence of well-adapted trivialisation2}
	Sections of $\MNull \to \gS$ are in one-to-one correspondence with sections of $L$ in $\CSoL{1}$:
	
	 If $\bGG\from \MNull \to L$ is a section of $L$ in $\CSoL{1}$ there is a unique section $s_{\bGG} \from \gS \to \MNull$ such that $\bGG \circ s_{\bGG} =0$. The other way round, a section $\bGG$ of $\CSoL{1}$ is uniquely defined by its zero-set and this zero-set always define a section of $\MNull \to \gS$.	In a well-adapted trivialisation $\left( \bgs , u \right)$ the isomorphism is given by
	\begin{equation}
		\begin{array}{ccc}
			\So{\MNull} & \to & \CSoL{1} \\ \\
			u\circ s_{\bGG}\left|\begin{array}{ccc}
				\gS &\to& \bbR\\
				x &\mapsto& G 
			\end{array}\right. & \mapsto &
			\bGG = \bgs \left( u-G \right).
		\end{array}  
	\end{equation}
\end{Proposition}
\begin{proof}
Let $\bGG \from \MNull \to L$ be a section of $\CSoL{1}$ and let $\left(\bgs , u \right)$ be a well-adapted trivialisation. Then $\bGG - \bgs u$ must have zero vertical derivative, $\bGG - \bgs u \in \CSoL{0}$ and therefore $\bGG = \bgs \left( u - G \right)$ where $G \in \Co{\gS}$ is a function on $\gS$. Let $s \from \gS \to \MNull $ be a section of $\MNull \to \gS$ and let $S = u\circ s \from \gS \to \bbR$ be its coordinates in the trivialisation given by $u$. We must have
\begin{equation}
	\bGG \circ s = \bgs \left(S-G \right)
\end{equation}
and therefore the section $s_{\bGG}$ obtained by taking $S=G$ parametrises the zero set of $\bGG$, $\bGG \circ s_{\bGG}=0$.	
\end{proof}

Finally making use of both proposition \eqref{Weak structure of null-infinity: Prop, equivalence of well-adapted trivialisation1} and \eqref{Weak structure of null-infinity: Prop, equivalence of well-adapted trivialisation2} we obtain our last equivalent definition of well-adapted trivialisations.
\begin{Proposition}\label{Weak structure of null-infinity: Prop, equivalence of well-adapted trivialisation3}
	A choice of well-adapted trivialisation $\left(\bgs , u \right)$ on a null-infinity manifold is equivalent to choosing simultaneously a section of $L_{\gS} \to \gS$ and a section of $\MNull \to \gS$.
\end{Proposition}
\begin{proof}
	By proposition \eqref{Weak structure of null-infinity: Prop, equivalence of well-adapted trivialisation1} a well-adapted trivialisation is equivalent to two sections $\bl_{0} \in \CSoL{0}$ and $\bl_{1} \in \CSoL{1}$. By proposition \eqref{Weak structure of null-infinity: Prop, covariantly constant sections isomorphism} $\bl_{0}$ is equivalent to a section of $L_{\gS}$ and by proposition \eqref{Weak structure of null-infinity: Prop, equivalence of well-adapted trivialisation1} $\bl_{1}$ is equivalent to a section of $\MNull \to \gS$ which concludes the proof.
\end{proof}
This last definition is especially satisfying from a geometrical point of view.  The main advantage of definition \ref{Weak structure of null-infinity: Def, well-adapted trivialisation} is however the fact that it exactly correspond to the ``gauge choices'' which are typically made in the literature written in terms of BMS coordinates. In this article we are interested in geometrical quantities which do not depend on these choices. 

\subsubsection{Change of well-adapted trivialisation}\label{sss: Change of well-adapted trivialisation}

Let $\left( \bgs , u\right)$ be a well adapted trivialisation and let $\left(\bgsh , \uh\right)$ be any other. Then there must exists two functions $\gO$ and $\xi$ on $\gS$ such that
\begin{align}\label{Weak structure of null-infinity: well-adapted trivialisation transformation rules}
	\bgsh = \gO^{-1}\bgs, \quad \uh = \gO\left(u - \xi\right).
\end{align}
This either directly follows from the definition of well-adapted trivialisations (in particular from $\bgsh \bn^a d\uh_a =1$) or from proposition \ref{Weak structure of null-infinity: Prop, equivalence of well-adapted trivialisation1} and proposition \ref{Weak structure of null-infinity: Prop, equivalence of well-adapted trivialisation2}.

What will be important for us are the transformation rules for sections $\bl$ of $\CSoL{k}$:
 Let $\bl \in \CSoL{k}$ with $k\neq 0$ then $\bl - \bgs ku$ must be in $\CSoL{0}$ and there thus exists a function $\ell \in \Co{\gS}$ on $\gS$ such that
\begin{equation}
	\bl = \bgs \left( ku + \ell \right) 
\end{equation}
The same reasoning holds for $\left( \bgsh , \uh \right)$ i.e there exists $\ellh \in \Co{\gS}$ such that
\begin{equation}
	\bl = \bgsh \left( k \uh + \ellh \right) 
\end{equation}
and we have the transformation rule
\begin{equation}\label{Weak structure of null-infinity: constant sections transformation rules}
		\ell \mapsto \ellh = \gO\left( \ell + k \xi \right).	
\end{equation}

\subsection{The BMS group}\label{sss: Conformal Carroll Symmetries}

The symmetry group of a null-infinity manifold $\left( \MNull \to \gS , \bh_{ab}, \bn^a \right) $ is the subgroup of automorphism $\Phi$ of $\MNull \to \gS$ preserving both $\bh_{ab}$ and $\bn^a$. If $\bgs \in \So{L_{\gS}}$ is choice of trivialisation and $h_{ab} = \bgs^{-2}\bh_{ab}$, $n^a = \bgs \bn$ are the associated representatives, this amounts to
\begin{equation}
\Phi^*h_{ab} = \gO^2 h_{ab}, \quad \Phi^{-1}_*n^a = \gO^{-1}n^a
\end{equation}
with $\gO^{-1} = \bgs^{-1} \Phi^{*}\bgs$ a function on $\MNull$.  We will call this group a BMS group and write $BMS\left(\MNull \to \gS, \hConf, \bn\right)$. 

\begin{Proposition}\label{Weak structure of null-infinity: Prop, BMS symmetry}
	Let $\left( \bgs , u\right)$ be a well-adapted trivialisation. The BMS group $BMS\left(\MNull \to \gS, \hConf, \bn\right)$ is realised by automorphism $\Phi$ of $\MNull$ such that
	\begin{align}\label{Weak structure of null-infinity: BMS action}
	\phi^* h_{AB} = \gO^2\; h_{AB}, \qquad u \circ \Phi = \gO\left(u - \xi\right)
	\end{align}
	with $\gO$ and $\xi$ two functions on $\gS$ and where $\phi = \pi \circ \Phi \circ \pi^{-1} \from \gS \to \gS$ is the diffeomorphism on the base induced by $\Phi$.
\end{Proposition}
\begin{proof}
Since $\Phi \from \MNull \to \MNull$ is an automorphism of $\MNull \to \gS$ and $\bh_{ab} = \pi^*\bh_{AB}$, it is clear that $\Phi^*h_{ab} = \gO^2 h_{ab}$ with $\gO^{-1} = \bgs^{-1} \Phi^* \bgs $ (now a function on $\gS$) if and only if the induced map on the base $\phi \from \gS \to \gS$ satisfies $\phi^* h_{AB} = \gO^2\; h_{AB}$. Then $ \phi^{-1}_*\left( \parD_u \right ) = \gO^{-1} \parD_u$ can be rewritten as $\parD_u\left( u\circ \Phi^{-1} \right) = \gO^{-1} \circ \phi^{-1}$ and this implies that there exists $\xi$ a function on $\gS$ such that $u \circ \Phi = \gO\left(  u - \xi\right)$.
\end{proof}
It follows from the above proposition that all diffeomorphisms $\Phi \from \MNull \to \MNull$ such that
\begin{equation}
\pi \circ \Phi \circ \pi^{-1} = Id_{\gS}, \qquad u \circ \Phi = u - \xi
\end{equation}
form a subgroup $\cT$ of the BMS group. These are the so-called ``super-translations''. Clearly, once we fix a well-adapted trivialisation, super-translations are parametrised by functions on $\gS$. There might or might not be any further symmetry depending on whether $\left(\gS , \hConf \right)$ admits conformal diffeomorphisms ($f \from \gS \to \gS$ such that $f^* h =  \gO^{2} h$).  Each of these conformal transformations indeed parametrise, in a fixed well-adapted trivialisation, a symmetry:
\begin{equation}
\pi \circ \Phi \circ \pi^{-1} = f, \qquad u \circ \Phi = \gO u.
\end{equation}
All in all we have,
\begin{equation}
BMS\left(\MNull \to \gS, \hConf, \bn\right)  = \cT \rtimes Conf(\gS, \hConf).
\end{equation}
where $Conf(\gS)$ is the group of conformal diffeomorphisms of $\gS$.
If one further suppose that $\left(\gS , \hConf \right) = \left(\S^{n-1} , \hConf_{\S^{n-1}} \right)$ is the conformal $(n-1)$-sphere then the resulting symmetry is the celebrated BMS group from \cite{bondi_gravitational_1962,sachs_asymptotic_1962}, $BMS_{n+1}= \cT \rtimes SO\left(n,1\right)$. 

It now follows from proposition \ref{Weak structure of null-infinity: Prop, BMS symmetry} that the action of the BMS group \eqref{Weak structure of null-infinity: BMS action} effectively sends a well-adapted trivialisation to another:

\begin{Proposition}\label{Weak structure of null-infinity: Prop, BMS symmetry2}
The BMS group naturally acts on the space of a well-adapted trivialisations via pull-back:
\begin{equation}
\left(\bgs , u \right) 	\quad \mapsto \quad \left( \bgsh = \Phi^*\bgs , \uh =  \Phi^* u \right)
\end{equation}
where $\Phi^*\bgs = \gO^{-1}\bgs$ and $\Phi^* u= \gO\left( u -\xi\right) $.
\end{Proposition}

\section{Radiative (or strong) structure of null-infinity}\label{s: Radiative (or strong) structure of null-infinity}

In this section we present our first main set of results: We define the ``strong''\footnote{We try to refrain to use the term ``radiative'' structure too loosely since strong null-infinity structures will only correspond to gravitational radiations on \underline{3-dimensional} null-infinity manifolds.} structure of null-infinity as a choice of Poincaré operator and show that the transformation rules for the tensor parametrising this operator are the same as those of the asymptotic shear of the BMS or Newman-Penrose formalism. We then show that the zeros of this operator are related to good-cuts. For completeness we also review from \cite{calderbank_mobius_2006,burstall_conformal_2010} how Möbius structures are related to complex projective structures.

\subsection{Strong null-infinity structures \texorpdfstring{($n\geq3$)}{}}

\subsubsection{Asymptotic shear and Poincaré Operators}

\begin{Definition}{Poincaré operator}\mbox{}\label{Strong structure of null-infinity: Def,Poincare operator}\\
Let $\left( \MNull \to \gS, \hConf_{ab}, \bn^a\right)$ be a null-infinity manifold. A compatible Poincaré operator $\cP \from \CSoL{k} \to \So{S^2_0\; T^*\gS \otimes L}$ is defined to be a linear differential operator of order two such that, in a well adapted trivialisation $\left ( \bgs , u \right)$, it takes the form 
	\begin{equation}\label{Strong structure of null-infinity: Poincare operator}
	\cP(\bl)_{AB} =  \bgs \left( \covD_A \covD_B \big|_0+ \frac{1}{2}[C_{AB} , \parD_u ]  \right) l
	\end{equation}
	where $l= \bgs^{-1} \bl$, ``$\parD_u$'' stands for $\bgs\LieD_{\bn}$, $\covD_A$ is the ``horizontal derivative'' given by $u$ and the Levi-Civita connection of $h_{AB} = \bgs^{-2}\hConf_{ab}$,  $\covD_A \covD_B \big|_0$ is the trace-free part of the Hessian, $C_{AB}$ is a trace-free symmetric tensor and $[C_{AB} , \parD_u ]l := C_{AB}\parD_ul - (\parD_u C_{AB})l$ .
\end{Definition}
Note that in this definition the trivialisation of $\MNull \to \gS$ given by $u$ is absolutely needed to make sense of the ``horizontal derivative'' $\covD_{A}l$. 
\begin{Definition}{Strong null-infinity structure}\label{Strong structure of null-infinity: Def,strong null-infinty structure}\mbox{} \\
	A \emph{strong null-infinity structure} $\left( \MNull, \hConf_{ab}, \bn^a , \cP \right)$ is a choice of weak null-infinity structure $\left( \MNull \to \gS, \hConf_{ab}, \bn^a \right)$ together with a compatible Poincaré operator $\cP$. 
\end{Definition}
An important remark is that the image of a Poincaré operator does not necessarily have constant vertical derivative, rather 
\begin{equation}\label{Strong structure of null-infinity: DP = ddC}
\covD_{\bn} \cP\left(\bl\right)_{AB}	= - \; \bl \;\frac{1}{2} \parD_u \parD_u C_{AB}.
\end{equation}
Here $\covD_{\bn}$  stands for the combined action of Lie derivative and the vertical covariant derivative on $L$. The above result immediately follows from the fact that since $\bgs \in \CSoL{0}$ and $\bl \in \CSoL{k}$ then $\parD_u h_{AB}=0$ and $\parD_u l =k$ is a constant.
 
Definition \ref{Strong structure of null-infinity: Def,Poincare operator} for Poincaré operators naively seems to depend on a choice of well-adapted trivialisation $\left( \bgs, u \right)$. It actually does not:
\begin{Proposition}
Let $\cP \from \CSoL{k} \to \So{S^2_0\; T^*\gS \otimes L}$ be a differential operator such that in a well-adapted trivialisation $\left( \bgs , u\right)$ it takes the form  \eqref{Strong structure of null-infinity: Poincare operator}. If $\left(\bgsh, \uh\right)$ is any other well-adapted trivialisation then there exists a trace-free symmetric tensor $\Ch_{AB}$ such that
\begin{equation}
	\cP(\bl)_{AB} =  \bgsh \left( \covDh_A \covDh_B \big|_0+ \frac{1}{2}[\Ch_{AB} , \parDh_{\uh} ]  \right) \lh.
\end{equation} 
\end{Proposition}
\begin{Proposition}
Let $\left(\bgs , u\right)$ and $\left(\bgsh = \gO^{-1}\bgs , \uh = \gO\left(u-\xi\right) \right) $ be well-adapted trivialisations. We have the transformation rule
	\begin{equation}\label{Strong structure of null-infinity: C transformation rules} 
	\begin{array}{ccc}
	C_{AB} & \mapsto & \Ch_{AB} = \gO\bigg(C_{AB} -2 \covD_A \covD_B \big|_0 \xi -2  \uh \; \covD_A \covD_B \big|_0 \gO^{-1} \bigg ) .
	\end{array}
	\end{equation}.
	\end{Proposition}

\begin{proof}
	Let $\bl \in \CSoL{k}$ be a scale with constant vertical derivative and let $\left(\bgs, u\right)$ be a well-adapted trivialisation we must have $ l \coloneqq \bgs^{-1}\bl = ku + \ell$ with $\ell \in \Co{\gS}$. We can therefore rewrite  $\cP\left(\bl\right) $ as
	\begin{equation}
		\cP(\bl)_{AB} = \bgs \bigg( \frac{1}{2} \; k\left( C_{AB} - u \parD_u C_{AB}\right ) +  \left( \covD_A \covD_B \big|_0 -  \frac{1}{2}\parD_u C_{AB}  \right) \ell \bigg)
	\end{equation}
	Recall from section \ref{sss: Change of well-adapted trivialisation} that if $\left( \bgsh = \gO^{-1} \bgs , \uh = \gO\left( u - \xi \right )\right)$ is any other well-adapted trivialisation we have $\lh = k \uh +\ellh$ with $\ellh = \gO\left( \ell +k \xi \right)$. Making use of the transformation rules \eqref{Conformal geometry: Connection transformation rules} for the connection, it is then a direct computation to show that we have
	\begin{equation}
		\cP(\bl)_{AB} = \bgsh \bigg( \frac{1}{2} \; k\left( \Ch_{AB} - \uh \parDh_{\uh} \Ch_{AB}\right ) +  \left( \covDh_A \covDh_B \big|_0 - \frac{1}{2} \parDh_u \Ch_{AB}  \right) \ellh \bigg)
	\end{equation}
where $\covDh_A$ is the Levi-Civita connection of $\hh_{AB} = \bgsh^{-2}\hConf_{AB}$ and $\Ch_{AB}$ is given by the transformation rule \eqref{Strong structure of null-infinity: C transformation rules}.	
\end{proof}

From the previous results we obtain immediately the following
\begin{Proposition}
 Poincaré operators form an affine space over $S^2_0 T^*\gS \otimes L$: if $\cP_0$ is a fixed Poincaré operator, any other can be written as
\begin{equation}
	\cP = \cP_0 + \frac{1}{2}[\bCt_{AB} , \covD_{\bn} ] 
\end{equation}
where $\bCt_{AB} \in \So{S^2_0 T^*\gS \otimes L}$ is a genuine tensor.\\ There isn't, however, any distinguished origin to this affine space.
\end{Proposition}

Finally, noting the identity of the transformation rules for the asymptotic shear that can be found in the literature \cite{barnich_finite_2016} and the transformation rules \eqref{Strong structure of null-infinity: C transformation rules} we have our first theorem:
\begin{Theorem}{Asymptotic shear and Poincaré operator}\mbox{}\label{Strong structure of null-infinity: Thrm,Asymptotic shear and Poincare operator}\\
Choices of asymptotic shear for an asymptotically flat space-times of dimension $n+1 \geq 4$ are in one-to-one correspondence with choices of Poincaré operators on the null-infinity manifold at the conformal boundary.
\end{Theorem}

At this stage the reader might wonder if the identity of transformations rules as stated is enough to justify the above result. In what follows we will show that the coefficient $C_{AB}$ appearing in Poincaré operators correspond to a choice of affine connection $D$ compatible with the conformal Carroll structure $D h_{AB}=0$, $Dn^a =0$ and thus match the definition of radiative structure of \cite{geroch_asymptotic_1977,ashtekar_radiative_1981, ashtekar_a._symplectic_1981,ashtekar_symplectic_1982,ashtekar_asymptotic_1987,ashtekar_geometry_2015,ashtekar_null_2018}. This equivalence will thus fully justify theorem \ref{Strong structure of null-infinity: Thrm,Asymptotic shear and Poincare operator}. In a future work \cite{Herfray_2020} we will show how the Poincaré operator can be derived directly from the geometry of asymptotically flat space-times, thus establishing the relationship between both geometries in a more direct way.

\subsubsection{News and Möbius structures}

Let $\cP$ be a Poincaré operator. In this sub-section we want to consider its restriction to sections $\bl \in \CSoL{0}$ which are vertically constant. By proposition \ref{Weak structure of null-infinity: Prop, covariantly constant sections isomorphism} these sections are identified with sections $\bell$ of $L_{\gS}$ via pull-back, $\bl = \pi^* \bell$. Let us write the operator obtained by this restriction as $\cM \from L_{\gS} \to  S^2_0 T^* \gS \otimes L$. In a well-adapted trivialisation,
\begin{equation}
	\cM(\bell)_{AB} = \bgs \left( \covD_A \covD_B \big|_0 - \frac{1}{2} \parD_u C_{AB}  \right) \ell
\end{equation}
and thus $\cM$ is completely parametrised by the ``Bondi News'', $N_{AB} \coloneqq \parD_u C_{AB}$. Under change of well-adapted trivialisation the transformation law for the Bondi News is 
\begin{align}\label{Strong structure of null-infinity: News transformation rules}
	\parDh_{\uh}\Ch_{AB} &= \parD_u C_{AB}  - 2\;\gO\; \covD_A \covD_B \big|_0 \gO^{-1}  \\
	&= \parD_u C_{AB} -2 \left(  \gU_A \gU_B -\covD_A \gU_B \right ) \big|_0. \nonumber
\end{align}

Let us suppose for a moment that $\parD_u \parD_u C_{AB}=0$  (this condition is equivalent to the vanishing of \eqref{Strong structure of null-infinity: DP = ddC} and thus independent of the choice of well-adapted trivialisation). Then $\cM$ is a differential operator on sections of $L_{\gS}$,
\begin{equation}
	\cM \from \left| 
	\begin{array}{ccc}
L_{\gS} & \to & S^2_0 T^* \gS \otimes L_{\gS} \\
\bell & \mapsto & \bgs \left( \covD_A \covD_B \big|_0 - \frac{1}{2} \parD_u C_{AB}  \right) \ell
	\end{array}
\right.
\end{equation}
As discuss in the introduction, the above data $\left(\gS , \hConf , \cM \right)$ is known in the conformal geometry literature as a \emph{Möbius structure} see \cite{calderbank_mobius_2006, burstall_conformal_2010} and operators of the type of $\cM$ as Möbius operators. Möbius operators form an affine space modelled on $S^2_0 T^* \gS$. In dimension $n\geq4$ there is a natural origin to this space given by 
\begin{equation}\label{Strong structure of null-infinity: origins of generalised mobius structure}
\cM^{(0)}\left(\bell\right)_{AB} =  \bgs \left( \covD_A \covD_B + P_{AB}  \right) \big|_0 \ell.
\end{equation}
where $P_{AB}\big|_0  = \frac{1}{n-3}  R_{AB}\big|_{0}$ is the trace-free Schouten tensor of $h_{AB}= \bgs^{-2}\bh_{AB}$. (Note that the transformation rules for the trace-free Schouten tensor \eqref{Conformal geometry: Schouten tensor transformation rules} indeed matches the transformation rules \eqref{Strong structure of null-infinity: News transformation rules} for $-\frac{1}{2} \parD_u C_{AB}$.)

In dimension $n=3$, $\gS$ is a Riemann surface: There is no canonical equivalent to the Schouten tensor in this case and thus no (local) origin to the space of Möbius operators. However in this dimension there is still a preferred Möbius operator of \emph{global} nature: It is given by the uniformization theorem, see \cite{calderbank_mobius_2006,burstall_conformal_2010} for a more detailed discussion. This natural origin of global nature to the space of Möbius structures is also the ``rho tensor'' from \cite{geroch_asymptotic_1977,ashtekar_radiative_1981, ashtekar_a._symplectic_1981,ashtekar_symplectic_1982,ashtekar_asymptotic_1987,ashtekar_geometry_2015,ashtekar_null_2018}.  When needed, will write \eqref{Strong structure of null-infinity: origins of generalised mobius structure} as well for the coordinates of this preferred global Möbius operator in dimension two.

All this justifies to introduce \emph{generalised Möbius operators}.
\begin{Definition}{Generalised Möbius operators}\mbox{}\\
Let $\left( \MNull \to \gS, \hConf_{ab}, \bn^a\right)$ be a null-infinity manifold. A compatible generalised Möbius operator $\cM \from \So{L_{\gS}} \to \So{S^2_0\; T^*\gS \otimes L}$ is defined to be a linear differential operator of order two such that, in a well adapted trivialisation $\left ( \bgs , u \right)$, it takes of the form 
\begin{equation}
\cM\left(\bell\right)_{AB} =  \bgs \left( \covD_A \covD_B\big|_0 -\frac{1}{2} N_{AB}  \right)  \ell.
\end{equation}
	where $\ell= \bgs^{-1} \bell$, $\covD_A$ is the Levi-Civita connection of $h_{AB} = \bgs^{-2}\hConf_{ab}$,  $\covD_A \covD_B \big|_0$ is the trace-free part of the Hessian and $N_{AB}$ is a trace-free symmetric tensor.	
\end{Definition}
Any Poincaré operator as in definition \eqref{Strong structure of null-infinity: Def,Poincare operator} then induces a generalised Möbius operator by restriction to sections with zero covariant derivative $\CSoL{0} = \pi^* L_{\gS}$. Just like Poincaré operators, generalised Möbius operators form an affine space modelled on $\pi^*\left( S^2_0 T^* \gS \right)$. As opposed to Poincaré operators however generalised Möbius operator have a natural origin  for $n \geq3$:
\begin{Proposition}
	In dimension $n\geq3$ any generalised Möbius operator $\cM$ can be written as
\begin{equation}
\cM = \cM^{(0)} - \frac{1}{2} \Nt_{AB}
\end{equation}
 where $\cM^{(0)}$ is the preferred Möbius operator \eqref{Strong structure of null-infinity: origins of generalised mobius structure} and $\Nt_{AB}$ is a section of $\pi^*\left( S^2_0 T^* \gS \right)$. If $n\geq4$, the operator $\cM^{(0)}$ is of local nature while if $n=3$ it is given by the uniformization theorem of Riemann surfaces.
  \end{Proposition}
 The name ``generalised Möbius structure'' is justified by the following.
 \begin{Proposition}   When $\covD_{\bn}\cM = \parD_u N_{AB}$ vanishes (equivalently when $\cL_{\bn}\Nt_{AB} =0$ ), $\cM$ defines a genuine Möbius operator on $\gS$. Generically however, $\covD_{\bn}\cM$ is not zero (and thus $\cM\left(\bell\right)$ is typically not in $\CSoL{0}$) and $\cM$ does not descend to a Möbius operator on $\gS$.
 \end{Proposition}

All this leads to our second theorem
\begin{Theorem}{Bondi-News and Generalised Möbius operators}\mbox{}\\ \label{Strong structure of null-infinity: Thrm, News and Mobius structure}
	Choices of Bondi-News for an asymptotically flat space-times of dimension $n+1 \geq 4$ are in one-to-one correspondence with choices of generalised Möbius operators on the null-infinity manifold at the conformal boundary. Stationary Bondi-News (i.e ``u independent'') are in one-to-one correspondence with genuine Möbius structures on the celestial sphere.
	
	 Depending on ones needs the Bondi news can either be represented, in a well-adapted trivialisation, by $N_{AB} = \parD_u C_{AB}$ or, more invariantly, by a genuine tensor $\Nt_{AB} = \parD_u C_{AB}~+~2 P_{AB}\big|_0$.
\end{Theorem}

Once again this can be justified by the identity of the transformation rules \eqref{Strong structure of null-infinity: News transformation rules} and the transformation rules for the news tensor of the literature \cite{barnich_finite_2016}. See also the discussion after theorem \ref{Strong structure of null-infinity: Thrm,Asymptotic shear and Poincare operator}.
  
 \subsection{Good cuts and the geometry of Poincaré operators  \texorpdfstring{($n\geq3$)}{}}

Let $\left(\MNull \to \gS , \bh_{ab} , \bn^a , \cP\right) $ be a strong null-infinity structure, we note $\cM$ the generalised Möbius operator given by restriction of $\cP$ to $\CSoL{0}$. Let $\left(\bgo , u \right)$ be a well-adapted trivialisation, by proposition \ref{Weak structure of null-infinity: Prop, equivalence of well-adapted trivialisation1} it is equivalent to a couple of sections $\left( \bgo , \bGG  = \bgo u \right)$ in $\CSoL{0}$ and $\CSoL{1}$. We will say that the well-adapted trivialisation is \emph{flat} if the related sections are zeros of the Poincaré operator:
\begin{Definition}{Flat well-adapted trivialisation}\label{Strong structure of null-infinity: Def, flat well-adapted trivialisations}\mbox{}\\
	Let $\left( \bgo, \bGG \right)$ be sections of $\CSoL{0} \oplus \CSoL{1}$. We will say that the well-adapted trivialisation given by proposition \ref{Weak structure of null-infinity: Prop, equivalence of well-adapted trivialisation1} is flat if and only if
 \begin{equation}
\cM\left(\bgo\right)=0, \qquad  \cP\left(\bGG\right)=0.
\end{equation}	
\end{Definition}
\noindent We remark that flat trivialisations might not exists (in fact this is the generic situation for we will see that existence of flat well-adapted trivialisation impose strong constraint on the curvature of the related tractor connection).

 In this section, we investigate the geometry of these flat trivialisations. 
 
  \subsubsection{The generalised good-cut equations}\label{sss: The generalised good-cut equations}
  
 Let us here consider $\bGG\in \CSoL{1}$. Recall from proposition \ref{Weak structure of null-infinity: Prop, equivalence of well-adapted trivialisation2} that $\bGG$ is in one-to-one correspondence with a section $s_{\bGG} \from \gS \to \MNull$ of $\MNull$. Let us choose a well-adapted trivialisation $\left( \bgs , u\right)$, we have
\begin{equation}
\bGG = \bgs \left( u - G\right) 
\end{equation} 
where $G$ is a function on $\gS$. By definition (or proposition \ref{Weak structure of null-infinity: Prop, equivalence of well-adapted trivialisation2}) $G\in \Co{\gS}$ is the coordinate of $s_{\bGG}$ in the trivialisation given by $u$ i.e $G = u\circ s_{\bGG}$. 
\begin{Proposition}\label{Strong structure of null-infinity: Prop, generlised GG eqs}
	We will say that $s_{\bGG}$ is a good-cut or that is ``satisfies the (generalised) good-cut equations for'' for $\cP$ if and only if the associated section $\bGG\in \CSoL{1}$ is a zero of the Poincaré operator. In a well-adapted trivialisation $\left( \bgs , u\right)$:
 \begin{align}\label{Strong structure of null-infinity: generlised GG eqs}
\begin{array}{c}
\cP(\bGG)_{AB} = \bgs \bigg( \frac{1}{2} \; \left( C_{AB} - u \parD_u C_{AB}\right ) -  \left( \covD_A \covD_B \big|_0 - \frac{1}{2}\parD_u C_{AB}  \right) G \bigg) = 0 \\
\Leftrightarrow \\ \\
 \qquad 	\covD_A \covD_B \big|_0 G  = \frac{1}{2}C_{AB}[G] \qquad \text{and} \qquad	\parD_u \parD_u C_{ab}=0.  
\end{array}
\end{align}
\end{Proposition}
In this proposition the vanishing of $\covD_{\bn}\cP = -\frac{1}{2}\parD_u \parD_u C_{ab}$ really is an ``integrability'' condition on $\cP$ i.e a necessary condition for $\cP$ to have any zeros. Before we come to the proof let us clarify that we are not abusing terminology and that for $n=3$ the above really correspond to the good-cut equations familiar from the literature (see \cite{newman_heaven_1976,hansen_r._o._metric_1978,ko_theory_1981,adamo_generalized_2010,adamo_null_2012}): In this dimension $\gS$ is a two-dimensional manifold and the conformal metric $\bh_{AB}$ amounts to a choice of complex structure. Let us pick local complex coordinates $\gz$. Let $\bgs$ be a scale and let us write  $\bgs^{-2} \bh_{AB} =\frac{ 4 d\gz d\gzb}{P^2\left(\gz,\gzb\right)}$ and $C_{AB} = 2\text{Re}\left( \frac{4}{P^2}d\gz d\gz \;\gs^0 \right)$ \footnote{The factor of $\frac{4}{P^2}$ appearing here is such that $\gs^0$ is the coordinate of $C_{AB}$ evaluated in an orthonormal basis i.e $\gs^0 = C_{AB} m^A m^B$ with $m^A \mb_A =1$.}. The first equation of \eqref{Strong structure of null-infinity: generlised GG eqs} can then be rewritten as 
\begin{equation}
\text{Re}\bigg( d\gz d\gz \Big( \; \parD_{\gz} \big( P \parD_{\gz}   \big)G - 2 \gs^0[G] \;\Big) \bigg) =0
\end{equation}
which are the generalised good-cut equations from \cite{adamo_generalized_2010}.

Equations \eqref{Strong structure of null-infinity: generlised GG eqs} are therefore generalisation to arbitrary dimensions of the good-cuts equations. In this sense one can think of the Poincaré operator $\cP$ as selecting among all possible ``cuts'' of $\MNull \to \gS$ the ``good'' ones. We now come to the proof of proposition \ref{Strong structure of null-infinity: Prop, generlised GG eqs}.
\begin{proof}
Let us suppose that $\bGG = \bgs \left( u - G \right) \in \CSoL{1}$ is a zero of the Poincaré operator:
 \begin{equation}
 	\cP(\bGG)_{AB} = \bgs \bigg( \frac{1}{2} \; \left( C_{AB} - u \parD_u C_{AB}\right ) -  \left( \covD_A \covD_B \big|_0 - \frac{1}{2}\parD_u C_{AB}  \right) G \bigg) = 0.
 \end{equation}
 Taking as well-adapted trivialisation $\left( \bgsh = \bgs , \uh = u-G \right)$ we have $\bGG = \bgs \uh$ i.e $ \GGh = 0$ and thus
 \begin{equation}\label{Strong structure of null-infinity: generlised GG eqs 1}
 	\cP(\bGG)_{AB} = \bgsh \;\frac{1}{2}\; \bigg(  \Ch_{AB} - \uh \parD_{\uh} \Ch_{AB} \bigg) =0
 \end{equation}
and, making use of the transformation rules \eqref{Strong structure of null-infinity: C transformation rules}, 
 \begin{equation}\label{Strong structure of null-infinity: generlised GG eqs 2}
 	\Ch_{ab} = C_{ab} - 2 \bigg( \covD_a \covD_b \big|_0 G \bigg ) .
 \end{equation}
Equation \eqref{Strong structure of null-infinity: generlised GG eqs 1} implies that $\Ch_{ab}$ is homogenous degree 1 in $\uh$, $\Ch_{ab} = \uh\; \Ch^{(1)}_{ab}$. Together with \eqref{Strong structure of null-infinity: generlised GG eqs 2} this gives
\begin{equation}
	C_{ab}[u]	= 2 \bigg( \covD_a \covD_b \big|_0 G \bigg ) + (u-G) \Ch^{(1)}_{ab}
\end{equation}
 or equivalently
 \begin{align}\label{Strong structure of null-infinity: generlised GG eqs 3}
 		\parD_u \parD_u C_{ab}=0 ,\qquad C_{ab}[G] = 2\bigg( \covD_a \covD_b \big|_0 G \bigg ).
 \end{align}
 \end{proof}
 
From the proof, one sees that if $\cP$ has a zero then in any well-adapted trivialisation we must have $\covD_{\bn}\cP = -\frac{1}{2}\parD_u \parD_u C_{ab}=0$. Thus a generic Poincaré operator $\cP$ will have no zero at all. On the other hand, if $\bGG \in \CSoL{1}$ and $\bGG'  \in \CSoL{1}$ are two zeros of $\cP$ (i.e define good-cuts) then, by linearity of $\cP$, their difference $\bGG -\bGG' \in \CSoL{0}$ must be a zero of $\cM$:
\begin{Proposition}
	The space of good-cuts for a Poincaré operator $\cP$ form an affine space over the space of zeros of the associated Möbius operator $\cM$.
\end{Proposition}

With this in mind we come to studying zeros of a Möbius operator.
 
\subsubsection{Möbius structures and complex projective coordinates}\label{sss: Mobius structures and complex projective coordinates}
 
 Since our generalised Möbius operators are really just a straightforward extensions of Möbius operators of conformal manifolds to the context of null-infinity manifold all results we will need are already present in the literature, see \cite{bailey_thomass_1994},\cite{calderbank_mobius_2006},\cite{burstall_conformal_2010}. For completeness and the convenience of the reader we review those which are most helpful in this context.

 Let us consider a section $\bgo \in \So{L_{\gS}}$ and let us further suppose that it is a zero of the generalised Möbius operator $\cM$. Let $\left( \bgs , u\right)$ be a well adapted trivialisation and $\go = \bgs^{-1}\bgo$ we must have
 \begin{equation}
	\cM(\bgo)_{AB} = \bgs \left( \covD_A \covD_B \big|_0 - \frac{1}{2} \parD_u C_{AB}  \right) \go =0
\end{equation}
We thus immediately have the following:
 \begin{Proposition}
  A scale $\bgo \in \So{L_{\gS}}$ is a zero of the Möbius operator $\cM$ if and only if, in a generic well-adapted trivialisation $\left( \bgs , u\right)$, we have
 \begin{equation}
 	-\frac{1}{2}\parD_u C_{AB}  = \go^{-1}\;  \covD_A \covD_B \big|_0 \go.
 \end{equation}
where $\go = \bgs^{-1}\bgo$.
 \end{Proposition}
In particular if $\cM$ has a zero we must have $\parD_u\parD_u C_{AB}=0$ in any well-adapted trivialisations.

\paragraph{Möbius operator and Einstein's equations \texorpdfstring{($n\geq4$)}{}}\mbox{}

As we already discussed, when $n\geq4$ there is a natural origin $\cM^{(0)}$ (of local nature) to the space of Möbius operator. Any generalised Möbius structure can then be written as $\cM = \cM^{(0)} -\frac{1}{2}\Nt_{AB}$ where $\Nt_{AB}$ is a genuine trace-free symmetric tensor. In a well-adapted trivialisation $\left( \bgs, u\right)$,
\begin{equation}
	\cM= \left( \covD_A \covD_B + P_{AB} \right) \big|_0  -\frac{1}{2} \Nt_{AB} 
\end{equation}
where $P_{AB}\big|_0 = \frac{1}{n-3} R_{AB}\big|_0$ is the trace-free Schouten tensor of $h_{AB}$. 

Let us suppose that $\bgo$ is a zero of $\cM$. In the well-adapted trivialisation given by $\left( \bgo , u \right)$ we must have
\begin{equation}
	\cM\left(\bgo \right) = \bgo \left( \frac{1}{n-3} R_{AB}\big|_0  -\frac{1}{2} \Nt_{AB} \right)  =0
\end{equation}
and we therefore have
\begin{Proposition}{(from \cite{bailey_thomass_1994}, \cite{burstall_conformal_2010})}
	Let $\cM = \cM^{(0)} -\frac{1}{2}\Nt_{AB}$ be a generalised Möbius operator on a null-infinity manifold of dimension $n\geq4$. Then $\bgo \in \So{L_{\gS}}$ is a zero of $\cM$ if and only if the Ricci tensor $R_{AB}$ of the associated metric $h_{AB} = \bgo^{-2}\bh_{AB}$ satisfies
	\begin{equation}
		R_{AB}\big|_0 = \frac{n-3}{2} \Nt_{AB}.
	\end{equation}
\end{Proposition}

\paragraph{Möbius operator as choices of complex projective structure \texorpdfstring{($n=3$)}{}}\mbox{}

In dimension $n\geq4$ a choice of generalised Möbius operator is not much more than a choice of trace-free symmetric tensor. However when $n=3$ and $\gS$ is a Riemann surface the geometry of Möbius operators is richer: A choice of \emph{integrable} Möbius structure then amounts to a choice of \emph{complex projective structure} on $\gS$. As we already discussed there still exists a preferred Möbius structure of global nature: it corresponds to the complex projective structure given by the uniformization theorem and pull-back of the natural complex projective structure of the model. We here briefly recall from \cite{calderbank_mobius_2006} how this arises.

 Let us pick local complex coordinates $\gz$ and define the (local) scale $\bgs_{(\gz)} \in \So{L_{\gS}}$ as the scale such that $ \bgs_{(\gz)}^{-2} \bh_{AB} = d\gz d\gzb$. In this scale the generalised Möbius operator reads,
\begin{equation}\label{Strong structure of null-infinity: Proof, integrable Mobius operator1}
\cM\left(\bgo\right) = \bgs_{(\gz)}\; \text{Re}\bigg( d\gzb d\gzb \Big( \; \parD_{\gzb} \parD_{\gzb}  - \frac{1}{2}\parD_u C_{\gzb \gzb} \;\Big) \bigg) \go_{(\gz)} 
\end{equation}
where $\bgs_{(\gz)}\go_{(\gz)} = \bgo$. Let $w$ be another set of complex coordinates and $\bgs_{(w)} = |\parD_{\gz}w|^{-1} \; \bgs_{(\gz)}$ the associated scale. In this scale, 
\begin{equation}
\cM\left(\bgo\right) = \bgs_{(w)}\;\text{Re}\bigg( d\wb d\wb \Big( \; \parD_{\wb} \parD_{\wb}  - \frac{1}{2}\parD_u C_{\wb \wb} \;\Big) \bigg) \go_{(w)} 
\end{equation}
and the transformation rules \eqref{Strong structure of null-infinity: News transformation rules} can be rewritten as
\begin{align}
-\frac{1}{2}\parDh_{\uh} C_{\wb \wb} = -\frac{1}{2}\parD_{u} C_{\gzb \gzb} - S_{\gz}(w)
\end{align}
where $S_{\gz}(w)$ is the standard Schwartzian derivative of $w$ with respect to $\gz$,
\begin{equation}
 S_{\gz}(w) = \left(\frac{w''}{w'} \right)'-\frac{1}{2}\left(\frac{w''}{w'} \right)^2.
\end{equation}
In particular, since $\cM\left(\bgs_{(\gz)} \right)=  \bgs_{(\gz)}\; \text{Re}\bigg( d\gzb d\gzb \Big(  - \frac{1}{2}\parD_u C_{\gzb \gzb} \;\Big) \bigg)$, we have the following:
\begin{Proposition}\label{Strong structure of null-infinity: Prop, complex projective coordinates1}
If $\gz$ and $w$ are local complex coordinates such that the associated local scales $\bgs_{(\gz)}$ and $\bgs_{(w)}$ are zeros of a Möbius operator $\cM$ then $S_{\gz}(w)=0$ and these complex coordinates must therefore be related by a Möbius transformation.
\end{Proposition}

\begin{Definition}\label{Strong structure of null-infinity: Def, integrable Mobius structures}
	A Möbius structure $\left(\gS , \bh_{AB} , \cM \right) $ is said to be \emph{integrable} if, for every point $x\in \gS$ there exists exists a neighbourhood $U_x$ and a zero of the Möbius operator $\bgs_x \in \So{L_{U_x}}$ which is the scale of a metric $(\bgs_x)^{-2} \bh_{AB}$ with constant scalar curvature.
\end{Definition}

Integrable Möbius structures are useful because they amount to giving $\gS$ a complex projective structure:
\begin{Proposition}{(from \cite{calderbank_mobius_2006}, \cite{burstall_conformal_2010})}\label{Strong structure of null-infinity: Prop, Mobius structures and complex projective coordinates}
	$\cM$ is integrable if and only if there exists a set of local complex coordinates $(\gz)_i$ covering $\gS$ such that $\cM\left(\bgs_{(\gz)_i}\right)=0$. By Proposition \ref{Strong structure of null-infinity: Prop, complex projective coordinates1} on each intersection these complex coordinates must be related by Möbius transformations. 
\end{Proposition}
\begin{proof}
Let $x\in \gS$ and $\gz$ be complex coordinates in a neighbourhood of $x$ such that $\cM\left(\bgs_{(\gz)}\right)=0$. If $\bgo$ is any scale then, in the trivialisation given by $\bgs_{(\gz)}$, we must have:
\begin{equation}
\cM\left(\bgo\right) = \bgs_{(\gz)}\;\text{Re}\bigg( d\gzb d\gzb  \; \parD_{\gzb} \parD_{\gzb}  \bigg) \go_{(\gz)}
\end{equation}
In particular local scales of the form $\go_{(\gz)} = \frac{1+ R \gz\gzb}{2}$ (with $R\in \bbR$) must be zeros of the Möbius operator. On the other hand, if $\bgo$ is a scale with constant scalar curvature defined in a neighbourhood of $x\in \gS$ there must exists local complex coordinates $\gz$ such that $\bgo^{-2}\bh_{AB} = \frac{4}{\left(1 +R \gz\gzb \right)^2 } d\gz d\gzb$. From eq \eqref{Strong structure of null-infinity: Proof, integrable Mobius operator1} we obtain $\cM\left(\bgo\right) =\bgs_{(\gz)}\; \text{Re}\bigg( d\gzb d\gzb \Big( - \frac{1}{2}\parD_u C_{\gzb \gzb} \;\Big) \bigg) \go_{(\gz)}  = \cM\left(\bgs_{(\gz)} \right) \go_{(\gz)} $ and therefore if $\bgo$ is a zero of the Möbius structure so is $\bgs_{(\gz)}$.
\end{proof}
 
Finally, integrability of a Möbius structure can be locally characterised by the vanishing of its ``curvature'':
\begin{Proposition}{(from \cite{calderbank_mobius_2006})}\mbox{}\\
Let $\cM= \left( \covD_A \covD_B\big|_0 -\frac{1}{2}N_{AB} \right)$ be a Möbius operator represented in a scale $\bgs \in \So{L}$. Its ``curvature''
\begin{equation}
	\bK_A\left(\cM\right)  = -\frac{1}{2}\left( \covD_C N_{A}{}^C + \frac{1}{2} \covD_{A} R \right)
\end{equation}	
is a genuine section of $T^*\gS \otimes L^{-2}$. What is more $\cM$ is integrable if and only if $\bK_A\left(\cM\right)=0$.
\end{Proposition}
\begin{proof}
	This will directly follow from the equivalence between Möbius operators and tractor connections.
\end{proof}

\subsection{Symmetries}

\begin{Definition}
Let $\left(\MNull \to \gS, \bh_{ab}, \bn^a, \cP\right)$ be a strong null-infinity structure and let $\Phi \from \MNull \to \MNull$ be a diffeomorphism of $\MNull$. We define $\Phi^* \cP$, the pull-back of $\cP$ by $\Phi$ as
\begin{equation}
\left(\Phi^* \cP\right) \left( \bl \right) \coloneqq \left( \Phi^* \circ \cP \circ (\Phi^{-1})^* \right)  \left( \bl \right).
\end{equation}

\end{Definition}

\begin{Definition}
We will say that a diffeomorphism $\Phi$ of $\MNull$ is a symmetry of the strong null-infinity structure $\left(\MNull\to \gS ,\bh_{ab}, \bn^a ,\cP\right) $ if and only if $\Phi \in BMS\left(\MNull \to \gS, \hConf, \bn\right)$ and $\Phi^* \cP = \cP$.
\end{Definition}
 
Let $\left(\MNull \to \gS, \bh_{ab}, \bn^a, \cP\right)$ be a strong null-infinity structure and let $\left(\bgs , u \right)$ be a well-adapted trivialisation. Let $\Phi$ be an element of the BMS group, by proposition \ref{Weak structure of null-infinity: Prop, BMS symmetry}, $\Phi$ sends $\left(\bgs , u \right)$ to another well-adapted trivialisation $\left( \bgsh = \Phi^* \bgs, \uh = \Phi^* u \right)$. There must exists $\gO$ and $\xi$ two functions on $\gS$ such that
 \begin{align}\label{Strong structure of null-infinity: BMS action on well-adapted trivialisation}
 \Phi^* \bgs = \gO^{-1} \bgs, \qquad \Phi^* u  = \gO\left(u - \xi\right).
 \end{align}

\begin{Proposition}\label{Strong structure of null-infinity: Prop, Symmetries}
Let $\left(\MNull \to \gS, \bh_{ab}, \bn^a, \cP\right)$ be a strong null-infinity structure and let $\left(\bgs , u \right)$ be a well-adapted trivialisation. 	

Let $\Phi$ be an element of the BMS group, $\Phi \in BMS\left(\MNull \to \gS, \hConf, \bn\right)$, we have
\begin{equation}
\Phi^* \cP  - \cP = \frac{1}{2} \bigg[ \gO^{-1}\; \Phi^* C_{AB} - C_{AB}+ 2\;\covD_A \covD_B \big|_0 \xi + 2\; \uh \;\covD_A \covD_B \big|_0 \gO^{-1} , \parD_u \bigg]
\end{equation}
where $\xi$ and $\gO$ are given by \eqref{Strong structure of null-infinity: BMS action on well-adapted trivialisation}.

In particular, if we suppose that  $\left(\bgs , u \right)$ is a \underline{flat} well-adapted trivialisation (see def \ref{Strong structure of null-infinity: Def, flat well-adapted trivialisations}) then $\Phi$ is a symmetry if and only if
\begin{align}\label{Strong structure of null-infinity: Symmetries eqs}
\covD_A \covD_B \big|_0 \xi &=0, & \covD_A \covD_B \big|_0 \gO^{-1}&=0.
\end{align}
\end{Proposition} 

Remarks: If $\bh_{ab}$ is taken to be conformally flat, the first of the two equations \eqref{Strong structure of null-infinity: Symmetries eqs} are the well-known conditions reducing super-translations $\cT \simeq \Co{\gS}$ to the group of usual translations $\bbR^{n+1}$ while the second is automatically satisfied in dimension $n\geq4$. In dimension $n=3$ and if $\bgs$ is taken to be a constant scalar curvature scale then these equations can be rewritten as a Schwartzian derivative and selects a (local) copy of $PSL\left(2,\bbC\right)$ inside the space of holomorphic transformations. Therefore a choice of Poincaré operator selects inside the infinite dimensional BMS group a copy of the Poincaré group (thus the name of the operator). We will in fact see that the conjoint requirements that $\bh_{ab}$ is conformally flat ( in dimension $n=3$ this is automatic and one instead require the integrability of the Möbius structure) and $\cP$ admits a zero is equivalent to the vanishing of the curvature of the associated Cartan connection.

 \begin{proof}
Let $\left(\bgs, u\right)$ be a well adapted trivialisation. If $\bl \in \CSoL{k}$ is a scale and $\bgs^{-1}\bl = \bl$ its coordinates we have
 \begin{equation}
\cP\left( \bl \right) = \bgs\left( \covD_A \covD_B \big|_0 + \frac{1}{2}\; \bigg[ C_{AB} , \parD_u \bigg] \ell \right).
 \end{equation}
Now since $\Phi^{*}\bgs = \bgsh$, we have $(\Phi^{-1})^* \bl = (\Phi^{-1})^*\left( \bgsh \ellh \right)=   \bgs \;(\Phi^{-1})^* \ellh$ and 
 \begin{align}
 	\Phi^* \cP\left(\bl\right)  &=  \Phi^*\left( \cP\left( (\Phi^{-1})^* \bell \right)\right) \\\nonumber
 	&= \Phi^*\left( \bgs \left(\covD_A \covD_B \big|_0 + \frac{1}{2}\Big[C_{AB} , \parD_u \Big] \right) (\Phi^{-1})^*\ellh \right)	\\ \nonumber
   &= \left( \bgsh \left(\covDh_A \covDh_B \big|_0 + \frac{1}{2}\Big[ \Phi^* C_{AB} , \parD_{\uh} \Big] \right) \ellh \right) \\  \nonumber
   &= \bgs \left( \covD_A \covD_B \big|_0+ \frac{1}{2}\; \bigg[ \gO^{-1}\Phi^* C_{AB} + 2\;\covD_A \covD_B \big|_0 \xi +  2\;\uh\;\covD_A \covD_B \big|_0 \gO^{-1} , \parD_u \bigg]  \right) \ell
 \end{align} 
where one goes from the second to the third line by making use of $\hh_{ab} = \Phi^* h_{ab}$, $\uh = \Phi^* u $ and from the second to the third by the conformal invariance of $\cP$ and transformation rules \eqref{Strong structure of null-infinity: C transformation rules} for $C_{AB}$. All in all we have
\begin{equation}
	 	\Phi^* \cP  - \cP = \frac{1}{2} \bigg[ \gO^{-1}\;\Phi^* C_{AB} - C_{AB}+ 2\; \covD_A \covD_B \big|_0 \xi + 2\; \uh \;\covD_A \covD_B \big|_0 \gO^{-1} , \parD_u \bigg]
\end{equation}
which concludes the proof.
 \end{proof}
 
 \subsection{Strong null-infinity structure in dimension \texorpdfstring{$n=2$}.}\label{ss: Generalised Laplace operators and Poincare operators in dimension n=2}

 Let $\left(\MNull \to \gS, \bh_{AB} ,\bn^a\right)$ be a null-infinity manifold of dimension $n=2$. It follows that $\gS$ has dimension $n=1$ and since a trace-free symmetric tensor on $\gS$ must be zero all definitions of the previous sections are vacuous. We therefore need to devise specific formulas adapted to this dimension.
 
 \subsubsection{Generalised Laplace operators}
 
 We recall from \cite{curry_introduction_2018} that, on any $d$-dimensional conformal manifold $\left( \gS , \bh_{AB} \right)$ with $d\geq2$, if $\bl$ is a section of $L^{1-\frac{d}{2}}$, we can define its \emph{conformal Laplacian} as follows: let  $\bgs \in \So{L_{\gS}}$ be a choice of scale and $l = \bgs^{-1 + \frac{d}{2}} \bl$ its coordinate, the conformal Laplacian is
 \begin{equation}\label{Strong structure of null-infinity: Conformal laplacian}
 \bgD \left| \begin{array}{ccc}
 \So{L^{1-\frac{d}{2}}} & \to & \So{L^{-1-\frac{d}{2}}} \\ \\
 \bl & \mapsto & \bgs^{-1-\frac{d}{2}}\left( \gD + \frac{2-d}{2} P \right) l,
 \end{array}\right.
 \end{equation}
 where $\gD = h^{ab}\covD_a\covD_b$ is the Laplacian of $h_{ab}=\bgs^{-2}\bh_{ab}$ and $P = \frac{1}{2(d-1)} R$ is the trace part of its Schouten tensor. If $\bgsh = \gO^{-1}\bgs$, the transformation rules for $P$,
 \begin{equation}
 \Ph = \gO^{-2}\left( P - \covD^A \gU_A + \frac{2-d}{2} \gU^2 \right) 
 \end{equation}
 (with $\gU_A = \gO^{-1}d_A \gO$) are such that \eqref{Strong structure of null-infinity: Conformal laplacian} does not depend on the choice of $\bgs$. In particular if $\left(\MNull \to \gS , \bh_{ab} , \bn^a \right)$ is a $n$-dimensional null-infinity manifold with $n\geq3$ it defines a differential operator
 \begin{equation}
 \bgD \from \Gamma_{0}\left[ L^{1-\frac{d}{2}} \right]  \to \Gamma_{0}\left[ L^{-1-\frac{d}{2}} \right]
 \end{equation}
 through the identification $\CSoL{0} = \pi^* \So{L_{\gS}}$.
 
 If $\gS$ is a $1$-dimensional manifold however ( equivalently if $\MNull \to \gS$ is a $(n=2)$-dimensional null-infinity manifold), the ``trace of the Schouten tensor'' is ill-defined and there is no canonical conformal Laplacian. Following \cite{calderbank_mobius_2006,burstall_conformal_2010} we define (generalised) Laplace structure as follows.
 \begin{Definition}{Generalised Laplace structure}\label{Strong structure of null-infinity: Def, generalised laplace structure}\mbox{}\\ 	
 	Let $\left(\MNull \to \gS , \bh_{ab} , \bn^a \right)$ be null-infinity manifold of dimension $2$. A compatible \emph{generalised Laplace operator} is a choice of linear differential operator of order two $\cL \from \Gamma_{0}\left[ L^{\frac{1}{2}} \right] \to \So{L^{-\frac{3}{2}}}$ such that in a well-adapted trivialisation $\left(\bgs, u\right)$ it takes the form
 	\begin{equation}
 	\cL\left(\bl\right) =  \bgs^{-\frac{3}{2}}\left(\gD -\frac{1}{4}M \right) l
 	\end{equation}
 	where $l = \bgs^{-\frac{1}{2}}\bl$ and $M$ is a function on $\MNull$.
 	A Laplace structure  $\left(\MNull \to \gS , \bh_{ab} , \bn^a , \cL \right)$ on a null-infinity manifold $\left(\MNull \to \gS , \bh_{ab} , \bn^a \right)$ is a choice of compatible generalised Laplace operator $\cL$.
 \end{Definition}
 
The only reason why generalised Laplace operators in definition \ref{Strong structure of null-infinity: Def, generalised laplace structure} differ from the Laplace operators of \cite{calderbank_mobius_2006,burstall_conformal_2010} is that in general the image of a generalised Laplace operators $\cL$ might not lie in $\CSoL{0}$. In fact the obstruction is
\begin{equation}
\covD_{\bn} \cL\left( \bl\right) =   -\;\bl \;\frac{1}{4}\parD_u M.
\end{equation}
However if $\parD_u M=0$, a generalised Laplace operator on $\MNull \to \gS$ induces a genuine Laplace structure $\cL \from \So{L_{\gS}^{\frac{1}{2}}} \to \So{L_{\gS}^{-\frac{3}{2}}}$ on $\gS$ through the identification $L_{\gS} = \CSoL{0}$.

Since (generalised) Laplace structure on null-infinity manifolds are a straightforward generalisation of the Laplace structures from \cite{calderbank_mobius_2006,burstall_conformal_2010} we have the
 \begin{Proposition}
 	Definition \ref{Strong structure of null-infinity: Def, generalised laplace structure} does not depend on the choice of well-adapted trivialisation: If $\left(\bgsh = \gO^{-1}\bgs , \uh= \gO\left( u -\xi  \right)\right) $ is any other well-adapted trivialisation, there exists $\Mh$ a function on $\MNull$ such that 
 	\begin{equation}
 	\cL\left(\bl\right) =  \bgsh^{-\frac{1}{2}}\left(\gDh -\frac{1}{4}\Mh \right) \lh.
 	\end{equation}
 	We have the transformation rules
 	\begin{equation}\label{Strong structure of null-infinity: 3D mass transformation rules}
\begin{array}{ccc}
	 	-\frac{1}{2}M  & \mapsto &	-\frac{1}{2}\Mh = \gO^{-2}\left( 	-\frac{1}{2}M - \covD^A \gU_A + \frac{1}{2} \gU^2 \right) 
\end{array}
 	\end{equation}
 	with $\gU_A = \gO^{-1}d_A \gO$.
 \end{Proposition}
 By construction the scalar function $-\frac{1}{2}M$ behaves like the ``missing'' trace of the Schouten tensor of $1$-dimensional manifolds, one can indeed check that its transformation rules \eqref{Strong structure of null-infinity: 3D mass transformation rules} mimics that of $P$ for $d=1$. Accordingly, any conformally invariant formula involving $P$ in dimension $n\geq3$ can be extended to $n=2$ at the cost of making a choice of generalised Laplace structure: then all the formula holds straightforwardly by replacing $P$ by $-\frac{1}{2}M$.

 In the end, comparing the transformation law of $M$ under change of well-adapted trivialisation with the transformation rules of the 3D mass-aspect of \cite{barnich_finite_2016} we have a version of theorem \ref{Strong structure of null-infinity: Thrm, News and Mobius structure} for $n=2$. 
 \begin{Theorem}{3D Mass-aspect and Generalised Laplace operators}\mbox{}\\ \label{Strong (or Radiative) BMS structure: Thrm, 3D Mass and Laplace structure}
 	Choices of mass-aspect $M$ for an asymptotically flat space-times of dimension $n+1 = 3$ are in one-to-one correspondence with choices of generalised Laplace operators on the null-infinity manifold at the conformal boundary. Stationary mass-aspect (i.e ``u independent'') are in one-to-one correspondence with genuine Laplace structures on the celestial sphere.
 \end{Theorem}
 
 \subsubsection{Poincaré operators in dimension \texorpdfstring{$n=2$}{2}}
 
 \begin{Definition}{Poincaré structure}\mbox{}\label{Strong structure of null-infinity: Def,Poincare operator n=2}\\
 	Let $\left( \MNull \to \gS, \hConf_{ab}, \bn^a ,\cL \right)$ be a null-infinity manifold of dimension $n=2$ equipped with a Laplace structure. A compatible Poincaré operator $\cP \from \CSoL{k} \to \So{ T^*\gS \otimes L^{-1}}$ is defined to be a linear differential operator of order three such that, in a well adapted trivialisation $\left ( \bgs , u \right)$, it takes of the form 
 	\begin{equation}\label{Strong structure of null-infinity: Poincare operator n=2}
 	\cP(\bl)_{A} =  \bgs^{-1} \Big(  \covD_A\; \gD  + N_A\; \parD_u - M \; \covD_A - \frac{1}{2}\covD_A M \Big) l
 	\end{equation}
 	where $l= \bgs^{-1} \bl$, $\covD_A$ is the Levi-Civita connection of $h_{AB} = \bgs^{-2}\hConf_{AB}$ and  $\gD = h^{AB} \covD_A \covD_B$ and $N_A$ is a 1-forms on $\MNull$ and $\cL = \left(\gD -\frac{1}{4}M \right)$.
 \end{Definition}
 
 \begin{Proposition}
 	Let $\cP \from \CSoL{k} \to \So{T^*\gS \otimes L^{-1}}$ be a differential operator such that in a well-adapted trivialisation $\left( \bgs , u\right)$ it takes the form  \eqref{Strong structure of null-infinity: Poincare operator n=2}. If $\left(\bgsh = \gO^{-1}\bgs , \uh = \gO\left(u-\xi\right) \right) $ is any other well-adapted trivialisation then there exists $\Nh_A$ and $\Mh$ such that
 	\begin{equation}
 	\cP(\bl)_{A} =  \bgsh^{-1} \left(  \covDh_A\; \gDh \; +\Nh_A\; \parDh_u - \Mh \;\covDh_A - \frac{1}{2}\covDh_A \Mh \right) \lh
 	\end{equation} 
We have the transformation rules 	
 	\begin{align}\label{Strong structure of null-infinity: M, N transformation rules} 
	 	-\frac{1}{2}M \quad &\mapsto  & \;-\frac{1}{2}\Mh &= \gO^{-2}\Big( -\frac{1}{2}M - 
	 		\covD^A \gU_A +\frac{1}{2}\gU^2 \Big)\\ 		
 	N_A  \quad&\mapsto  &\; \Nh_A &= \gO^{-1} \Bigg( N_A - \left(  \covD_A \gD - M \covD_A\right) \xi   +  \uh \Big(\left(\covD_A- 2 \gU_A \right)  \left( \covD^C \gU_C +\frac{1}{2} \gU^2 \right) +M \gU_A \Big) \Bigg)\nonumber
 	\end{align}
 \end{Proposition}  
 In order to prove this transformation rules it will be useful to remark that Poincaré operators can be recast the following alternative form:
 \begin{Proposition}\label{Strong structure of null-infinity: Prop, Rewritting the Poincare operator} 
 	In a well-adapted trivialisation $\left( \bgs , u \right)$ we have
 	\begin{equation}\label{Strong structure of null-infinity: Rewritting the Poincare operator} 
 	\cP_A\left(\bl\right) = 2\bl^{-1}\covD_A\left(  \bl^{\frac{3}{2}} \cL \big( \bl^{\frac{1}{2}}\big)  \right)  + \bgs^{-1} N_A \covD_{\bn} \bl
 	\end{equation}	
 \end{Proposition}
 \begin{proof}
 	Proposition \ref{Strong structure of null-infinity: Prop, Rewritting the Poincare operator} can be checked directly by expanding equation \eqref{Strong structure of null-infinity: Rewritting the Poincare operator} and comparing with equations \eqref{Strong structure of null-infinity: Poincare operator n=2}. 
 	
 	Let us now establish the transformation rules. Let $\bl \in \CSoL{k}$ be a scale with constant vertical derivative and let $\left(\bgs, u\right)$ be a well-adapted trivialisation we must have $l = ku + \ell$ and $\cP\left(\bl\right) $ can be rewritten as
 	\begin{equation}
 	\cP(\bl)_{A} =  \bgs^{-1} \Big(  \left( \covD_A\; \gD  - M \; \covD_A  -\frac{1}{2} \covD_A M \right)\ell +k \left(N_A -\frac{1}{2}u\covD_A M \right) \Big) 
 	\end{equation}
 	Let $\left( \bgsh = \gO^{-1} \bgs , \uh = \gO\left( u - \xi \right )\right)$ be another well-adapted trivialisation, making use of proposition \ref{Strong structure of null-infinity: Prop, Rewritting the Poincare operator} (and the covariance of the Laplace operator), the above can be rewritten as
 	\begin{align}
 	\cP(\bl)_{A} &=   2 \bell^{-1} \covD_A \left( \bell^{\frac{3}{2}} \cL\left( \bell^{\frac{1}{2}}\right)  \right)  +\bgs^{-1} \; k \left(N_A -\frac{1}{2} u\covD_A M \right) \\
 	& =\bgsh^{-1} \Bigg(  \left( \gDh \;\covDh_A - \Mh \; \covDh_A  -\frac{1}{2} \covDh_A \Mh \right)\left( \gO \ell\right) +k \;\gO^{-1} \left(N_A - u \frac{1}{2}\covD_A M \right) \Bigg) 
 	\end{align}	
 	making use of $\lh = k\uh +\ellh$ with $\gO \ell = \ellh -k \gO \xi$ (see eq \eqref{Weak structure of null-infinity: constant sections transformation rules}) and using proposition \ref{Strong structure of null-infinity: Prop, Rewritting the Poincare operator} again for $\xi$ one can derive the following transformation rule
 	\begin{align}
 	\Nh_A -\uh \frac{1}{2} \covDh_A \Mh= \gO^{-1}\bigg( N_A- u\; \frac{1}{2} \covD_A M\bigg) - \gO^{-1}\left( \covD_A  \gD - M \covD_A -\frac{1}{2} \covD_A M \right) \xi   .
 	\end{align}
 	Finally making use of the transformation rules \eqref{Strong structure of null-infinity: 3D mass transformation rules} for $M$ the above can be put in the form \eqref{Strong structure of null-infinity: M, N transformation rules}.
 \end{proof}

 \section{The tractor bundle associated to the weak structure of null-infinity}\label{s: The tractor bundle associated to the weak structure of null-infinity}
 
 In this section we show that null-infinity manifolds are canonically equipped with a vector bundle which is the equivalent of the tractor bundle of conformal geometry \cite{bailey_thomass_1994}. However, since we are dealing with degenerate conformal metrics, one cannot directly export results from the literature, rather one needs to make use of the structure at hand to adapt the construction. Just as its non-degenerate counterpart the tractor bundle will be useful to produce geometrical objects and operators which are manifestly conformally invariant. Note that this will only make use of the weak structure of null-infinity, the strong structure being related to a choice of tractor connection. 
 
 \subsection{The Tractor bundle at \texorpdfstring{$\MNull$}{scri}}
 
 We would like to define the tractor bundle associated to a null-infinity manifold $\left(\MNull \to \gS, \hConf_{ab}, \bn^a \right)$. Since $\hConf_{ab}$ is a degenerate conformal metric standard techniques from \cite{bailey_thomass_1994} do not apply but we can instead consider the following.
 
 Let $J^2 L$ be the bundle of $2$-jets of sections of $L \to \scrI$. Let $\bgs \in \CSoL{0}$ be a vertically constant scale, it gives local coordinates on $J^2 L$. Let $\left( l , \parD_{a} l, \parD_{a}\parD_{b} l \right)$ be these coordinates. 
 Let $ \widetilde{J^2 L} \subset J^2 L$ be the sub-bundle of the $2$-jet bundle given by
 \begin{equation}\label{Weak structure of null-infinity: tractor bundle def aux}
\bgs \bn^{a}  \parD_{a} l = cst, \qquad \bn^{a}  \parD_{a}\parD_{b} l = 0.
 \end{equation}
 A first remark is that the above conditions are meaningful, i.e do not depend on the choice of scale $\bgs \in \CSoL{0}$: Had we chosen $\bgsh = \gO^{-1}\bgs \in \CSoL{0}$ (implying $\bn^{a}\gO=0$) then $\hat{l} = \gO l$ would also satisfy the above constraints.
 
 Now there is a canonical injection $\pi^* \left( S^2T^*\gS \otimes L \right) \to \widetilde{J^2 L}$ given by $\bga_{AB} \mapsto \left(0,0, \bga_{AB}\right)$. The conformal metric $\hConf_{AB}$ allows to define $\pi^* \left( S^2T^*\gS\right)\big|_0$, the space of trace-free symmetric tensors and the injection
 \begin{equation}\label{Weak structure of null-infinity: tractor bundle def injection}
 \begin{array}{ccc}
 \pi^* \left( S^2T^*\gS \otimes L\right)\big|_0 & \to & \widetilde{J^2 L} \\
 \bga_{AB} & \mapsto & \left( 0 , 0 , \bga_{AB} \right) 
 \end{array}.
 \end{equation}
 We define the dual tractor bundle $\cT^* \to \scrI$ as the quotient of $\widetilde{J^2 L}$ by the image of \eqref{Weak structure of null-infinity: tractor bundle def injection}.
 \begin{Definition}\label{Weak structure of null-infinity: tractor bundle definition}
 	The dual tractor bundle $\cT^* \to \MNull$ on a null-infinity manifold $\left(\MNull, \hConf_{ab}, \bn^a \right)$ is
 	\begin{equation}
 	\cT^*= \Quotient{\widetilde{J^2 L}}{ \pi^* \left( S^2T^*\gS\big|_0 \otimes L \right)}.
 	\end{equation}
 	where $\widetilde{J^2 L}$ is the sub-bundle of the $2$-jet bundle defined by \eqref{Weak structure of null-infinity: tractor bundle def aux}.
 \end{Definition}
 By construction, the dual tractor bundle is a vector bundle over $\MNull$ with $(n+2)$-dimensional fibres. The tractor bundle $\cT \to \MNull$ is (evidently) taken to be the dual of $\cT^* \to \MNull$.
 
 Note that in this definition we do not use the extra structure of a null-infinity manifold and it therefore also make sense for any conformal Carroll manifolds. It is very likely that all the tractor construction presented in this article extends (may be with some generalisations) to any conformal Carroll manifolds but we will not try to achieve this here.
 
 \subsection{Well-adapted trivialisations}
 
 The preceding definition is conceptually useful but ill-suited for practical purpose. We however have the following proposition which can be considered as an alternative definition:
 \begin{Proposition}{Thomas' splitting}\mbox{}
 	
 	Let $\left( \bgs, u \right) $ be a well-adapted trivialisation on a $n$-dimensional null-infinity manifold $\left( \scrI , \hConf_{ab}, \bn^a \right)$ with $n\geq3$ (if $n=2$ we also need to suppose that $\MNull$ is equipped with a generalized Laplace structure, see section \ref{ss: Generalised Laplace operators and Poincare operators in dimension n=2}) it defines an isomorphism
 	\begin{equation}
 	T_{(\bgs, u)} \left| 
 	\begin{array}{ccccccccccc}
 	\cT & \to     &  L &\oplus&  T\gS \otimes L^{-1} &\oplus& L^{-1} &\oplus& \bbR \\
 	Y^I & \mapsto & \Big( \bY^{+} &,& \bY^A &,& \bY^{-} &,& Y^u \Big) 
 	\end{array}\right.
	\end{equation}
 	We will call this isomorphism the splitting associated with $\left( \bgs, u \right)$.
 	
 	If $\left( \bgsh = \gO^{-1} \bgs , \uh = \gO\left( u - \xi \right )\right)$ is any other well-adapted trivialisation on $\scrI$ we have the transformation rules
 	\begin{equation}\label{Weak structure of null-infinity: Standard transformation rules}
 	\Mtx{ \bYh^{+} \\  \bYh^{A} \\ \bYh^{-} \\ \Yh^u}  =
 	\Mtx{ 1 & 0 &0 &0\\
 		\bgU^{A}& \gd^{A}{}_{B} & 0 & 0 \\ 
 		-\frac{1}{2} \bgU^2 & -\gU_B & 1 &0 \\
 		\frac{\bgs}{n-1}\Big( \bgD\xi + (\bgU^2- \covD_C \bgU^C)\left(u-\xi\right)  \Big)  & \bgs\Big( \gU_B \left( u-\xi\right )  -d\xi_{B} \Big)  & 0 & 1	}
 	\Mtx{ \bY^{+} \\  \bY^{B} \\ \bY^{-} \\ Y^u}	
 	\end{equation}
 	where $\gU_{A} = \gO^{-1}(d\gO)_{A}$ and all upper-case Latin indices of the beginning of the alphabet $(A,B,etc)$ are raised and lowered with the conformal metric $\hConf_{AB}$.
 \end{Proposition}
 
 \begin{proof}
 	
 	Let $\bgs \in \CSoL{0}$ and $h_{AB} = \bgs^{-2} \hConf_{AB}$ the corresponding metric. It defines a connection on $L^k$ as $\bgs^{-k} \covD \bl \coloneqq d\left( \bgs^{-k}\bl\right)$. We use $\covD$ for tensor product of this connection with the Levi-Civita connection of $h_{AB}$ and $\gD = h^{AB} \covD_A\covD_B$ the associated Laplacian. We also note $R$ the scalar curvature of $h_{AB}$. The Thomas' splitting associated with $\left( \bgs,u\right)$ is given by
 	\begin{equation}\label{Weak structure of null-infinity: proof, Standard transformation rules}
 	T_{\left( \bgs , u\right)} \left| 
 	\begin{array}{ccc}
 	\cT^* & \to     & \bbR \oplus L \oplus T^*\gS \otimes L  \oplus L^{-1} \\
 	\Mtx{ l \\ \parD_{a} l  \\ h^{ab} \parD_{a}\parD_{b} l } & \mapsto & \Mtx{  \covD_{\bn} \bl \\ \bl \\ \covD_{A} \bl \\ -\bgs^{-2}\frac{1}{n-1}\left( \gD + P \right)\bl  } 
 	\end{array}\right.
 	\end{equation}	
 	Note that $\bgs$ is needed to make sense of both the coordinates on the left-hand side, the covariant derivatives and the trace of the Schouten tensor $P = \frac{1}{2(n-2)} R$ on the right-hand side, on the other hand $u \from \MNull \to \bbR$ (which amounts to a trivialisation of $\MNull \to \gS$) is needed to make sense of the ``horizontal'' derivative $\covD_{A}$. Finally in the case where $n=2$ the trace of the Schouten tensor tensor is ill-defined and so is the last line of equation \eqref{Weak structure of null-infinity: proof, Standard transformation rules}. However, if we suppose that $\MNull$ has been equipped with a Generalised Laplace structure $\cL = \gD - \frac{1}{4}M$ (see section \ref{ss: Generalised Laplace operators and Poincare operators in dimension n=2} for more on generalised Laplace structure) then one can make sense of the formula by replacing $P$ by $-\frac{1}{2}M$.
 	
 	The transformations rules for the dual tractor bundle $\cT^*$ are obtained by a direct computation.
 	\begin{equation}\label{Weak structure of null-infinity: Standard transformation rules2}
 	\Mtx{ \Yh_{u} \\  \bYh_{+} \\ \bYh_{A} \\ \bYh_{-}}  =
 	\Mtx{ 1 & 0 &0 &0\\
 		0 & 1 & 0 &0 \\	
 		\;\bgs\left(d\xi_A - \gU_A\left(u-\xi\right)\right)& \gU_A & \gd^{B}{}_{A}& 0 \\
 		-\frac{\bgs}{n-1}\left(  \bgD \xi + (n-1)\bgU^A d\xi_A -\left(\covD_A \bgU^A + (n-2)\bgU^2 \right)(u-\xi) \right)   & -\frac{1}{2} \bgU^2 & -\bgU^B & 1	
 	}
 	\Mtx{ Y_u \\ \bY_{+} \\ \bY_{B} \\ \bY_{-}}	
 	\end{equation}
 	In order to obtain these transformation rules it helps to make use of the fact that the definitions of the dual tractor bundle and well-adapted trivialisations implies that $ \bgs^{-1} \bl = k u + \ell$ (where $\parD_u \ell =0$) with the transformation rules \eqref{Weak structure of null-infinity: constant sections transformation rules}. It then considerably simplifies the computation to remark that for $k=0$ one is effectively brought back to computing the usual tractor transformation rules which was done in detail in \cite{bailey_thomass_1994} or \cite{gover_almost_2010}. 
 	
 	The transformation rules \eqref{Weak structure of null-infinity: Standard transformation rules} for the tractor bundle  are then obtained from \eqref{Weak structure of null-infinity: Standard transformation rules2} and the pairing
 	\begin{equation}
 	Y_I W^I = \bY_{-} \bm{W}^{+} + \bY_{+} \bm{W}^{-} + \bY_{A} \bm{W}^{A}  + Y_{u} W^u
 	\end{equation}
 	between a tractor $W^I$ and a dual tractor $Y_I$.
 \end{proof}	
 
 \subsection{Geometry of the tractor bundle}\label{ss: Geometry of the tractor bundle}
 
 Fibres of the tractor bundle effectively are copies of the hyper-surface $I^{\perp}$ of the flat model  (see equation \eqref{The flat model: null coordinates on Iperp} for coordinates on $I^{\perp}$ making this identification manifest).
 
In fact, the transformation rules \eqref{Weak structure of null-infinity: Standard transformation rules} for the tractor bundle can be seen to be matrix elements belonging to the Carroll group $Carr\left(n\right)$ (see equation \eqref{The flat model: Carr(n) action on scrI} for a concrete realisation of this group). The tractor bundle is thus an associated bundle for $Carr\left(n\right)\rtimes\bbR$. 
 
 We here discuss the geometrical structure that this induces.
 
 \paragraph{The tractor metric and the infinity tractor}
 The tractor bundle on $\MNull$ is equipped with a degenerate metric of signature $\left(n,1 \right)$. If $\left(\bgs, u\right)$ is a well-adapted trivialisation in the associated splinting the tractor metric reads
 \begin{equation}\label{Weak structure of null-infinity: tractor metric}
 g_{IJ} Y^I Y^J = 2\bY^{+}\bY^{-} + \bY_A \bY^A.
 \end{equation}
It also comes with a preferred tractor the ``infinity tractor'' $I^I$, spanning the degenerate direction of the tractor metric $g_{IJ} I^J =0$. In any splitting of the tractor bundle given by a well-adapted trivialisation,
 \begin{equation}\label{Weak structure of null-infinity: infinity tractor}
 I^{I} = \Mtx{0 \\ 0 \\ 0 \\ 1}.
 \end{equation}
 \begin{proof}
 	It is straightforward to see that the transformation rules \eqref{Weak structure of null-infinity: Standard transformation rules} preserve the infinity tractor \eqref{Weak structure of null-infinity: infinity tractor}. A longer but straightforward computation then also shows that \eqref{Weak structure of null-infinity: tractor metric} is invariant under \eqref{Weak structure of null-infinity: Standard transformation rules}. Alternatively, this follows from the properties of the flat model, see section \ref{ss: Homogenous space structure of null-infinity}.
 \end{proof}
 
 \paragraph{The reduced tractor bundle}
 
 Since $\cT$ is equipped with a preferred tractor $I^I$ we can consider the bundle $\cT/ I \to \MNull$ obtained by taking the quotient. We call this bundle the reduced tractor bundle on $\MNull$. By construction it is a vector bundle with $(n+1)$-dimensional fibres equipped with a (non-degenerate) metric of signature $\left(n,1\right)$. It is in fact canonically isomorphic with the pull-pack of $\cT_{\gS} \to \gS$, the tractor bundle of $\gS$:
 
 \begin{Proposition}
 	\begin{equation}
 	\Quotient{\cT}{\bbR I}  = \pi^*\left(\cT_{\gS} \right) 
 	\end{equation}
 \end{Proposition}
 \noindent and we thus have the short sequence of bundle
 \begin{equation}
 0 \to \bbR I \to \cT \to \pi^* \cT_{\gS} \to 0.
 \end{equation}
 
 \begin{proof}
 	The most immediate way to see the isomorphism is to note that the transformation rules of the reduced tractor bundle are precisely the standard transformation rules for the tractor bundle of a non-degenerate conformal metric  (see \cite{bailey_thomass_1994}). Alternatively, one can see the isomorphism from the jet bundle definition as follows: the dual of the reduced tractor bundle $\left(	\Quotient{\cT}{\bbR I} \right)^*$ identifies with the sub-bundle of $\cT^*$ such that $\parD_u \ell =0$. From definition \ref{Weak structure of null-infinity: tractor bundle definition} one then sees that $\left(	\Quotient{\cT}{\bbR I} \right)^*$ identifies with the pull-back of a quotient of the $2$-jet bundle on $\gS$. This quotient is just the tractor bundle on $\gS$ (as defined e.g in \cite{gover_almost_2010}).	
 \end{proof}

 \paragraph{Filtration}
 
 The tractor bundle is also equipped with a preferred ``position tractor'' $\bX^I \in \cT \otimes L$. In a splitting of the tractor bundle:
 \begin{equation}
 \bX^I = \Mtx{0 \\ 0 \\ 1 \\ 0}.
 \end{equation}
 The position tractor is null with respect to the tractor metric \eqref{Weak structure of null-infinity: tractor metric}, $\bX^I \bX_I =0$, and can be thought as a preferred inclusion $L^{-1} \to \cT$, we note $L^{-1}X \subset \cT$ the image of this inclusion. The position tractor also gives a preferred projection $\cT \to L$ obtained by contraction: if $Y^I \in \cT$ is any tractor, its projection is $Y^I \bX_I \in L$. We note $X^{\perp}$ the orthogonal subspace to $L^{-1}X$. $X^{\perp}$ is a vector bundle over $\MNull$ with $(n+1)$-dimensional fibres. Let $Y^I \in X^{\perp}$ be a section of this bundle, in a splitting of the tractor bundle we have:
 \begin{equation}\label{Weak structure of null-infinity: Xperp tractor subspace}
 Y^I = \Mtx{0 \\ \bY^A \\ \bY^{-} \\ Y^u} \in X^{\perp}.
 \end{equation}
 The tractor metric restricted to this sub-bundle is degenerate of rank $2$ with degenerate directions spanned by $\bX^I$ and $I^I$. We have the filtration:
 \begin{equation}
 L^{-1}X \subset X^{\perp} \subset \cT.
 \end{equation}
\begin{proof}
 	This immediately follows from the triangular form of the transformation rules \eqref{Weak structure of null-infinity: Standard transformation rules}.
\end{proof}

 Finally, the quotient of $X^{\perp}$ by $\bX$ is canonically isomorphic to $T\MNull \otimes L^{-1}$: Let $Y^I \in X^{\perp}$ be parametrised as in equation \eqref{Weak structure of null-infinity: Xperp tractor subspace} we will write,
 \begin{equation}
 \bY^a = \Mtx{\bY^A \\
 	Y^u} \in T\MNull \otimes L^{-1}.
 \end{equation}
 where
 \begin{align}
  \bY^A &\coloneqq \bY^a \gth^A_a, &Y^u &\coloneqq \bY^a du_a.
 \end{align}
Here $\gth^A_a$ is the projection from $T \MNull$  to the quotient $T \MNull / \bn  = \pi^* T\gS$
 
 \begin{proof}
 	It suffices to check that the transformation rules \eqref{Weak structure of null-infinity: Standard transformation rules}  restricted to $X^{\perp} / X$ are the transformation rules for the splitting of the tangent bundle $T\MNull \otimes L^{-1} \to \pi^*T\gS\otimes L^{-1} \oplus \bbR \bn$ given by $\bY^a = \bY^A \gth^a_A + Y^u \bn^a $ where $\gth^a_A$ is defined by $\gth_A^a \gth_a^b = \gd^B_A$, $\gth^a_A du_a=0$. 
 \end{proof}
 
 \paragraph{Thomas operator}
 
 Thomas' operator is a differential operator $T\from L \to \cT^*$ which associates to any scale $\bgt$ a dual tractor $T^{(\bgt)}_I $. It can be defined in a well-adapted trivialisation $\left(\bgs, u\right)$ as
 \begin{equation}\label{Weak structure of null-infinity: Thomas operator}
 T^{(\bgt)}_I = \Mtx{\covD_{\bn}\bgt \\ \bgt \\ \covD_A \bgt \\ -\bgs^{-2}\frac{1}{n-1}\left( \gD + P\right) \bgt }.
 \end{equation}
 Here the trace of the Schouten tensor $P = \frac{1}{2(n-2)}R$ only really make sense in dimension $n>2$. In dimension $n=2$, and once a generalised Laplace structure $\cL = \gD -\frac{1}{4}M$ has been chosen, one can make sense of the above formula by simply replacing $P$ by $-\frac{1}{2}M$ (see section \ref{ss: Generalised Laplace operators and Poincare operators in dimension n=2} for more on generalised Laplace structures).
 
\subsection{Some remarks on the notation}
Tractors are efficient to produce formulas which are conformally invariants, however the notation might appear intricate at first encounter. We here remind the reader of our conventions for abstract indices.

If $\left(\MNull \to \gS , \bh_{ab} , \bn^a \right)$ is a null-infinity manifold, lower case Latin indices of the beginning of the alphabet $a,b, etc$ will denote tensors on $\MNull$, upper case Latin indices of the beginning of the alphabet $A, B, etc$ will denote tensors on $\gS$ but also more frequently sections of the related pull-back bundles on $\MNull$. For example $V^a$ and $U_a$ are sections of $T\MNull$ and $T^*\MNull$ while $V^A$ and $U_A$ might represent a sections of $T\gS$ and $T^*\gS$ or ,more likely, sections of $\pi^*T\gS = T\MNull /\bn$ and $\pi^* T^*\gS$.
The projection operator from $T\MNull$ to the quotient $T\MNull / \bn$ will be denoted by $\gth_a^A$ and we will write for example:
\begin{equation}
	V^A =  V^a \gth^A_a.
\end{equation}
In the same vein, we will write
\begin{equation}
	\bm{U_u} = \bn^a U_a
\end{equation}
for the contraction of a 1-form $U_a$ with $\bn^a$.

If $u \from \MNull \to \bbR$ is a trivialisation of $\MNull \to \gS$ it defines a decomposition $\So{T\MNull} = \So{\pi^*T\gS} \oplus \So{L} $, we will then write
\begin{equation}
	V^a = \Mtx{ V^A \\ \bm{V}^u }
\end{equation}
and
\begin{equation}
	U_a = \gth^A_a U_A + \bm{U_u} \;\bgs du_a.
\end{equation}
In particular, with these conventions the identity on $T\MNull$ is
\begin{equation}
	\gd^a_b = \Mtx{ \gth^A_b \\ \bgs du_b}.
\end{equation}

A couple of remarks about ``lowering'' and ``raising'' indices: Upper case indices of the middle of the alphabet $I, J, etc$ denote tractor indices, e.g $V^I \in \So{\cT}$, $U_I \in \So{\cT^*}$, they are raised and lowered with the tractor metric $g_{IJ}$.  Upper case indices of the beginning of the alphabet $A, B, etc$ denoting sections of $\pi^*T\gS$ are raised and lowered using the conformal metric $\hConf_{AB}$ e.g $V^A \in \So{T\gS}$, $\bV_A \in \So{T^*\gS \otimes L^{2}}$. The only exception is for the Laplacian (when a choice of scale $\bgs \in \So{L}$ is understood); we will both use $\gD = h^{AB} \covD_A \covD_B$ and $\bgD = \hConf^{AB} \covD_A \covD_B$ depending on the context.

Finally, whenever a choice of scale $\bgs\in \So{L}$ is understood, $\covD$ denotes the tensor product of the Levi-Civita connection of $h_{AB} = \bgs^{-2} \hConf_{AB}$ with the connection on $L^k$ given by $\covD\left( \bl \right) \coloneqq \bgs^k d\left( \bgs^{-k} \bl \right)$. 
 
\section{The tractor connection associated to the strong structure of null-infinity} 

This section present our second main set of results: Poincaré operators are in one-to-one correspondence with a certain class of connection on the tractor bundle which we call null-normal Cartan connections.

\subsection{Tractor connections}

\subsubsection{Definition}

	Let $D\from T\MNull \otimes \So{\cT} \to \So{\cT}$ be a connection on the tractor bundle $\cT \to \MNull$. We suppose that it preserves both the tractor metric $g_{IJ}$ and the infinity tractor $I^I$. If $\left( \bgs , u \right)$ is a well-adapted trivialisation, we have in the associated splitting:
	\begin{equation}\label{Strong structure of null-infinity: DX in coordinates}
 \bgs D_b\left(\bgs^{-1}\bX^I\right) =  		\Mtx{ 0 \\  \phi^A_b \\ \bgs\;\covD^{(\gt)}_b\bgs^{-1} \\ \bphi_b^u } \in T^*\MNull \otimes X^{\perp} \otimes L.
	\end{equation}
Where $\phi^A_b \in \So{T^*\MNull \otimes \pi^*T\gS}$ is a 1-form with values in $T\MNull/ \bn$, $\covD^{(\gt)}_b$ is a connection on $L$ (extended to any $L^k$ by tensoriality) and $\bphi_b \in \So{T^*\MNull \otimes L}$ is 1-form with values in $L$.

We recall from section \ref{ss: Geometry of the tractor bundle} that we have a canonical identification of $X^{\perp} / \bbR X$ with $T\MNull \otimes L^{-1}$. Consequently, taking the quotient of \eqref{Strong structure of null-infinity: DX in coordinates} by $\bX^I$, we obtain an endomorphism of $T\MNull$:
\begin{equation}
\phi^a_b = \Mtx{ \gth^A_b \\  \bphi_b^u} \in \So{\End\left(T\MNull \right)}.
\end{equation}

\begin{Definition}{Tractor connection}\mbox{}\label{Strong structure of null-infinity: Def, tractor connection} \\
	Let $D\from T\MNull \otimes \So{\cT} \to \So{\cT}$ be a connection on the tractor bundle preserving both the tractor metric $g_{IJ}$ and the infinity tractor $I^I$. Let $\left( \bgs , u \right)$ be a well-adapted trivialisation, and let $D\left(\bgs^{-1} \bX^I\right) $ be given by \eqref{Strong structure of null-infinity: DX in coordinates}.
	We will say that $D$ is a \emph{tractor connection} if the two following properties hold:
	\begin{itemize}
		\item $\phi^a_b$ is the identity on $T\MNull$ i.e $\phi^a{}_b = \gd^a{}_b$
		\item $\covD^{(\gt)} \from T\MNull \otimes \So{L} \to \So{L}$ is the connection $\covD$ on $L$ given by the scale $\bgs$: $\covD \bl \coloneqq d\left( \bgs^{-1} \bl \right)$. In particular $\covD^{(\gt)}_b\bgs^{-1} =0$.
	\end{itemize}
\end{Definition}
\begin{Proposition}
The above definition does not depend on a choice of well-adapted trivialisation: if it holds for a particular well adapted trivialisation $\left( \bgs , u \right)$, it must holds for any other $\left( \bgsh , \uh \right)$.	
\end{Proposition}
\begin{proof}
The first property clearly does not depend on the choice of well-adapted trivialisation. One can check from the transformation rules \eqref{Weak structure of null-infinity: Standard transformation rules} that the second is invariant as well.
\end{proof}

If $D$ is a tractor connection and $Y^I \in \So{\cT}$ is any tractor field we have, in the splitting given by a well-adapted trivialisation $\left(\bgs , u \right)$,
\begin{equation}\label{The tractor connection associated to the strong structure of null-infinity: tractor connection in coordinates}
	D_b Y^I = \Mtx{ \covD_b & -\bgth_b{}_C  & 0 &0  \\
	- \bxi^A_b &	\covD^{(\go)}_b \gd^A_C & \gth^A_b & 0  \\
	0 &  \xi_b{}_C & \covD_b & 0\\
		-\bpsi_b & -\frac{1}{2}\bC_b{}_C & \bgs du_b	& \covD_b	
		} \Mtx{\bY^{+} \\ \bY^C \\  \bY^{-} \\ Y^u  }.
\end{equation}
Where $\covD^{(\go)}_a$ is a metric-compatible connection on $T\gS$, $\covD$ is the connection on $L$ given by $\bgs$, $\bC_a{}_B$ is a 1-form on $\MNull$ with values in $\pi^*\left(T^*\gS \right) \otimes L$, $\xi_b{}_A$ is a 1-form on $\MNull$ with values in $\pi^*\left(T^*\gS \right)$ and $\bpsi_b$ is a $L^{-1}$-valued 1-form on $\MNull$.

\subsubsection{Torsion}

Let $D_b$ be a tractor connection and $F^I{}_J{}_{cd}$ its curvature. In the splitting given by a well-adapted trivialisation $\left(\bgs , u \right)$, it reads:
\begin{equation}
	\frac{1}{2}F^I{}_J{}_{cd} = \Mtx{ \gth_{[c}{}^E \; \xi_{d]}{}_E &  - \covD{}_{[c}{}^{(\go)}\; \bgth_{d]}{}_B  & 0 &0 \\
		-\; \covD{}_{[c}{}^{(\go)} \bxi_{d]}{}^A & \frac{1}{2}\; F_{(\go)}{}^A{}_B{}_{cd} +  \bxi^A{}_{[c} \;\bgth_{d]}{}_B + \gth_{[c}{}^A \;\xi_{d]}{}_B & \covD^{\go}{}_{[c} \;\gth_{d]}{}^A &0 \\
		0 & \covD_{[c}{}^{(\go)} \; \xi_{d]}{}_B & -\; \gth_{[c}{}^E \; \xi_{d]}{}_E & 0 \\
		- \covD_{[c} \bpsi_{d]} + \frac{1}{2}\; \bC_{[c}{}^E \; \xi_{d]}{}_E  &  -\frac{1}{2}\covD_{[c}{}^{(\go)} \; \bC_{d]}{}_B + \bgs du_{[c} \; \xi_{d]}{}_B + \bpsi_{[c} \; \bgth_{d]}{}_B &  -\frac{1}{2}\; \bC_{[c}{}^E \; \bgth_{d]}{}_E & 0}.
\end{equation}

We will say that a tractor connection is torsion-free if $\bT^I \coloneqq F^I{}_J\bX^J$ vanishes. 
Let us write $\xi_{b}{}_A =  \xi_{BA} \; \gth_b{}^B + \bxi_{uA} \;\bgs du_b$ and $\bC_{b}{}_A = \bC_{BA}\;  \gth_b{}^B + C_{uA} \;\bgs du_b$, the torsion-free condition on $D$ is equivalent to
\begin{align}\label{Strong structure of null-infinity: torsion-free equations}
	\xi_{[AB]} &=0 & \bxi_{uA} &=0  &
	\bC_{[AB]} &=0 & C_{uA} &=0
\end{align}
together with
\begin{equation*}
	\covD_{[c}{}^{(\go)} \;\gth_{d]}{}^A = 0.
\end{equation*}
In particular this implies that $\covD^{(\go)}$ is the Levi-Civita connection $\covD$ of $h_{AB} = \bgs^{-2} \bh_{AB}$.

\subsubsection{Relation with Ashtekar/Geroch connections}\label{sss: Relation with Ashtekar/Geroch connections}

We here discuss how the equivalence class of connections that have been discussed in \cite{geroch_asymptotic_1977,ashtekar_radiative_1981, ashtekar_a._symplectic_1981,ashtekar_symplectic_1982,ashtekar_asymptotic_1987,ashtekar_geometry_2015,ashtekar_null_2018} 
naturally appear in our description as an equivalence class of coordinates for the tractor connection.

Let $\left(\bgs , u \right)$ be a well adapted trivialisation and let $D$ be a torsion-free tractor connection. This defines an affine connection on $\MNull$ as follows: Let $V^a = \Mtx{V^A \\ \bV^u} \in \So{T\MNull}$ be a vector field. Making use the expression \eqref{The tractor connection associated to the strong structure of null-infinity: tractor connection in coordinates} of the tractor connection and the torsion-free condition, we define a connection $D_b^{(\bgs, u)}$ on $T\MNull$ as:
\begin{equation}
D_b^{(\bgs, u)} V^a = \Mtx{\covD_b\gd^A_C & 0 \\ -\frac{1}{2} \gth_{b}^B \bC_{BC} & \covD_b} \Mtx{V^C \\ \bV^u}.
\end{equation}
This connection is torsion-free and preserves both $\bgs \bn^a$ and the metric $h_{ab}= \bgs^{-2} \bh_{ab}$.  The symmetric tensor $\bC_{AB}$ appearing as a coefficient of the tractor connection \eqref{The tractor connection associated to the strong structure of null-infinity: tractor connection in coordinates} therefore precisely correspond to a choice of connection $D_b^{(\bgs, u)}$ on $\MNull$ as in \cite{ashtekar_a._symplectic_1981,ashtekar_symplectic_1982,ashtekar_asymptotic_1987,ashtekar_geometry_2015,ashtekar_null_2018}. This construction however explicitly relies on the choice of well-adapted trivialisation $\left(\bgs , u\right)$ and therefore a tractor connection really defines an equivalence class of connections $\{D_b^{(\bgs, u)} \;\big|\; (\bgs, u)\in \;\text{well-adapted} \}$. For example, we leave as an exercise to the reader to verify that, by going from $\left( \bgs , u\right)$ to $\left( \bgs , \uh = u -\xi\right)$ and making use of \eqref{Weak structure of null-infinity: Standard transformation rules} one has the transformation rules
 \begin{align}
\Mtx{\Vh^A \\ \bVh^u} &= \Mtx{V^A \\ \bV^u - \bgs d\xi_{C} V^C},  &-\frac{1}{2}\Ch_{AB} &= -\frac{1}{2}C_{Ab} + \covD_A \covD_B \big|_0 \xi
\end{align}
 and that this gives
\begin{equation}
D_b^{(\bgs, u-\xi)} V^a = D_b^{(\bgs, u )} V^a + \frac{1}{n-1} \gD \xi \; h_{bc} V^c \; \bgs \bn^a.
\end{equation}

From the above discussion it should appear clearly that a torsion-free tractor connection defines an equivalence class of connections on $\MNull$ and that these correspond precisely to the equivalence class discussed in \cite{ashtekar_a._symplectic_1981,ashtekar_symplectic_1982,ashtekar_asymptotic_1987,ashtekar_geometry_2015,ashtekar_null_2018}. We will now show that by restricting to a reasonably natural class of tractor connections (which we call ``null-normal'') one can turn this relation into an equivalence.

 \subsection{Null-normal tractor connections}
 
Everywhere in this section $D$ is a torsion-free tractor connection.
  
 \subsubsection{Compatibility with Thomas operator}
 
  Let $Y_I$ be a dual tractor and $D_b Y_I$ be its covariant derivative.  Consider the equation $DY_I=0$. In the splitting given by well-adapted trivialisation $\left(\bgs, u\right) $:
\begin{equation}
	D_b Y_I = \Mtx{ \covD_{b} & 0 & 0 & 0 \\
		-\bgs du_b & \covD_b & -\gth_b{}^C & 0 \\
		\frac{1}{2} \bC_b{}_{A}& -\xi_b{}_{A} & \covD_b \;\gd^C_A & \bgth_b{}_{A} \\
		\bpsi_b & 0 & \bxi_b{}^C & \covD_b} \Mtx{ Y_u \\ \bY_{+} \\ \bY_C \\ \bY_{-} } =0.
\end{equation}
 The first two lines can readily be solved as
 \begin{align}
 	\bY_A &= \covD_A \bY_{+},  & Y_u &= \covD_{\bn} \bY_{+} = k
 \end{align}
 where $k\in \bbR$ is a constant. The third line then is equivalent to
 \begin{equation}
 \left(\covD_A \covD_B - \xi_{AB} \right) \bY_{+}  + \frac{1}{2} \bC_{AB} k + \bh_{AB} \bY_{-} =0
 \end{equation}
where we made use of the torsion-free conditions \eqref{Strong structure of null-infinity: torsion-free equations}. Taking the trace (with respect to $ \hConf^{AB}$) we obtain
\begin{equation}
	\bY_{-} = -\frac{1}{n-1} \left( \left(\bgD - \bxi^C{}_C \right) \bY_{+}  + \frac{1}{2} \bC^C{}_C k \right)
\end{equation}
In the end, we obtain a map from $\CSoL{k}$ to $\cT^*$,
\begin{equation}
	\bgt \mapsto \Mtx{ k \\ \bgt \\ \covD_A \bgt \\  -\frac{1}{n-1} \left( \left( \bgD - \bxi^C{}_C \right) \bgt  + \frac{1}{2} \bC^C{}_C k \right)  }
\end{equation}
this map is pretty close to Thomas operator \eqref{Weak structure of null-infinity: Thomas operator}. We will say that $D$ is compatible with the Thomas operator if these two maps coincide. This is equivalent to
\begin{align}\label{Strong structure of null-infinity: Thomas compatibility equations}
	\bxi^C{}_C &= -\bgs^{-2} P & \bC^C{}_C =0.
\end{align}

 \subsubsection{Reduced null-normal Cartan connection}
 
 Let $D$ be a torsion-free tractor connection compatible with the Thomas operator. Since the torsion $F^I{}_J \bX^J$ vanishes, we can meaningfully restrict the curvature to $X^{\perp}/X = T\MNull \otimes L^{-1}$. In a well-adapted trivialisation $\left(\bgs , u \right)$ this restriction reads
 \begin{equation}
 	F^a{}_b{}_{cd}  = \Mtx{  
 				R{}^A{}_{B}{}_{cd} + 2 \; \bxi_{[c}{}^A \; \bgth_{d]}{}_B + 2 \; \gth_{[c}{}^A \; \xi_{d]}{}_B & 0 \\
 		\covD_{[c} \; \bC_{d]}{}_B + 2 \; \bgs du_{[c} \; \xi_{d]}{}_B + 2\; \bpsi_{[c} \; \bgth_{d]}{}_B & 0 }.
 \end{equation}
where  $R{}^A{}_{B}{}_{cd} = R{}^A{}_{B}{}_{CD} \; \gth^C_c \; \gth^D_d $ is the Riemann tensor of $h_{AB}=\bgs^{-2} \bh_{AB}$.
 
 Let us consider the equation
  \begin{equation}\label{Strong structure of null-infinity: null-normal conditions1}
 	 \bn^c \; F^a{}_b{}_{cd}  = \Mtx{ 
 	 		0 & 0 \\
 	 		\left( -\frac{1}{2}\parD_u C_{DB} +\xi_{DB} + \bpsi_u \bh_{DB}\right) \gth_d{}^D & 0 } =0
 \end{equation}
 where $C_{AB} \coloneqq \bgs^{-1}\bC_{AB}$, $\parD_uC_{AB} \coloneqq \LieD_{\bn} \bC_{AB}$, $\bpsi_u \coloneqq \bn ^a \bpsi_a$ and we made use of the torsion-free equations \eqref{Strong structure of null-infinity: torsion-free equations}. We will call equations \eqref{Strong structure of null-infinity: null-normal conditions1} the first null-normal equations. They are solved as
 \begin{align}\label{Strong structure of null-infinity: null-normal equations1}
 	\xi_{AB}\big|_{0} &= \frac{1}{2}\parD_u C_{AB}, & \bpsi_u &= -\frac{1}{n-1} \bxi^C{}_C = \frac{1}{n-1}\bgs^{-2} P.
 \end{align}
 where $\big|_{0}$ means ``trace-free part of'' and we made use of equations \eqref{Strong structure of null-infinity: Thomas compatibility equations} given by the compatibility with Thomas operator.
 
 \begin{Definition}{Reduced Null-Normal Cartan connection}\label{Strong structure of null-infinity: Def, Reduced Null-Normal Cartan connection}\mbox{}
 	
 	Let $\left(\MNull \to \gS , \bh_{ab} , \bn^a \right)$ be a null-infinity manifold. We will say that a connection $\Dt$ on the reduced tractor bundle  $\cT/I=\pi^*\cT_{\gS}$ is a \emph{reduced null-normal Cartan connection} if it is the restriction of a connection $D$ on $\cT$ such that
 	\begin{itemize}
 		\item $D$ is a tractor connection, see definition \ref{Strong structure of null-infinity: Def, tractor connection}
 		\item $D$ is torsion free i.e satisfies equation \eqref{Strong structure of null-infinity: torsion-free equations}
 		\item $D$ is compatible with Thomas operator i.e satisfies equation \eqref{Strong structure of null-infinity: Thomas compatibility equations}
 		\item $D$ satisfies the first null-normality equations i.e \eqref{Strong structure of null-infinity: null-normal conditions1}.
 	\end{itemize}	
 	If $n\geq4$ we will say that $\Dt$ is the \emph{normal Cartan connection} if it is the pull-back of the normal Cartan connection on $\left(\gS , \bh \right)$ (see \cite{bailey_thomass_1994} or proposition below for the definition of the normal Cartan connection of a conformal manifold). 
 \end{Definition}
 
 \begin{Proposition}\label{Strong structure of null-infinity: Proposition, Reduced Null-Normal Cartan connection}
 	Let $\left(\MNull \to \gS , \bh_{ab} , \bn^a \right)$ be a $n$-dimensional null-infinity manifold. Then, in a well-adapted trivialisation, reduced null-normal Cartan connections $\Dt$ must be of the form
  			\begin{equation}
  				\Dt_b Y^I = \Mtx{ \covD_b & -\bgth_b{}_C  & 0  \\
  					- \bxi_b{}^A &	\covD_b & \gth_b{}^A  \\
  					0 & \xi_b{}_B & \covD_b \\
  				} \Mtx{\bY^{+} \\ \bY^C \\  \bY^{-}}
  			\end{equation}	
where $\covD$ is the tensor product of the Levi-Civita connection of $h_{AB} = \bgs^{-2}\bh_{AB}$ with the connection on $L$ given by the scale $\bgs$ and
 	\begin{itemize}
 		\item if $n\geq3$, \begin{equation*}
 		\xi_b{}_A = \left(\frac{1}{2}N_{AB} - \frac{P}{n-1}  h_{AB} \right)\gth_b{}^B
 	\end{equation*} where $P= \frac{1}{2(n-2)}R$ is the trace of the Schouten tensor and ``the Bondi news'' $N_{AB}$  is a symmetric trace-free tensor.
 			
 		\item if $n=2$ and splittings of the tractor bundle are given by a Laplace operator $\cL = \gD -\frac{1}{4}M$, \begin{equation*}
 		\xi_b{}_A =   \frac{1}{2}M \; h_{AB} \;\gth_b{}^B
 		\end{equation*} where ``the 3D mass aspect'' $M$  is a function on $\MNull$.
 	\end{itemize}
 If $n\geq4$ a reduced null-normal Cartan connection is the normal Cartan connection if $-\frac{1}{2}N_{AB}~=~P_{AB}\big|_0$ where $P_{AB}\big|_0 = \frac{1}{n-3} R_{AB}\big|_0$ is the trace-free Schouten tensor. \end{Proposition}

 \subsubsection{Null-normal Cartan connections}
 
 Let $D$ be a connection on $\cT$ satisfying all the conditions of definition \ref{Strong structure of null-infinity: Def, Reduced Null-Normal Cartan connection}. The curvature $F^I{}_J \big|_{X^{\perp}/X}$ must be of the form
 
  \begin{equation}
 F^a{}_b{}_{cd}  = \gth^C_{c} \; \gth^D_{d} \; \Mtx{  0  & 0 &0 \\
 	R^A{}_{B}{}_{CD} + 2 \; \bxi_{[C}{}^A \; \bh_{D]B} + 2 \; \gd_{[C}{}^A \; \xi_{D]}{}_B & 0 &0 \\
 		\covD_{[C} \; \bC_{D]}{}_B + 2\; \bpsi_{[C} \; \bh_{D]B} & 0 & 0}.
 \end{equation}
 contracting with $\bh^{bc}$ we obtain the second null-normal equations:
  \begin{equation}\label{Strong structure of null-infinity: null-normal condition2}
  \Mtx{ 0 & 0 &0 \\
  	 -\bgs^{-2} R^A{}_B -(n-3) \bxi^A{}_B -\bxi^C{}_C \; \gd^A{}_B& 0 &0 \\
  	-\frac{1}{2} \covD_C \bC^{C}{}_{B} -(n-2)\bpsi_B &0 &0 } =0.
  \end{equation}
Making use of \eqref{Strong structure of null-infinity: null-normal equations1}, this is solved as
\begin{align}\label{Strong structure of null-infinity: null-normal equations2}
(n-3)\left(\; -\frac{1}{2}\parD_u C_{AB}  \right)&= R_{AB}\big|_0, & (n-2)\; \bpsi_A = -\frac{1}{2} \covD_C \bC^{C}{}_{A}.
\end{align}
Note that for $n\ge 4$ equations \eqref{Strong structure of null-infinity: Thomas compatibility equations}, \eqref{Strong structure of null-infinity: null-normal equations1} and \eqref{Strong structure of null-infinity: null-normal equations2} imply that $\xi_{AB} = -P_{AB}$. In particular this implies that the restriction of the connection to the reduced tractor bundle $\cT/I = \pi^* \cT_{\gS}$ is the pull-back of the normal Cartan connection on $\left(\gS , \bh \right)$.

We now summarize the results of this section.

\begin{Definition}{Null-Normal Cartan connection}\label{Strong structure of null-infinity: Def, Null-Normal Cartan connection}\mbox{}\\
	Let $\left(\MNull \to \gS , \bh_{ab} , \bn^a \right)$ be a null-infinity manifold. We will say that a connection $D$ on the tractor bundle is a Null-Normal Cartan connection if it satisfies all the items listed in definition \ref{Strong structure of null-infinity: Def, Reduced Null-Normal Cartan connection} together with the second normality equations \eqref{Strong structure of null-infinity: null-normal condition2}.
\end{Definition}

\begin{Proposition}\label{Strong structure of null-infinity: Proposition, Null-Normal Cartan connection}
	Let $\left(\MNull \to \gS , \bh_{ab} , \bn^a \right)$ be a $n$-dimensional null-infinity manifold. Then, in the splitting given by a well-adapted trivialisation, a null-normal tractor connection must be of the form
	\begin{equation}
	D_bY^I = \Mtx{ \covD_b & -\bgth_{bC}  & 0 & 0  \\
		- \bxi_b{}^A &	\covD_b & \gth_b{}^A & 0  \\
		0 &  \xi_{bC} & \covD_b & 0\\
		-\bpsi_b & -\frac{1}{2}\bC_{bC} & du_b	& \covD_b	
	} \Mtx{\bY^{+} \\ \bY^C \\  \bY^{-} \\ Y^u  }.
	\end{equation}
	where $\covD$ is the tensor product of Levi-Civita connection of $h_{AB} = \bgs^{-2}\bh_{AB}$ with the connection on $L$ given by the scale $\bgs$ and	
	\begin{itemize}
		\item if $n\geq4$, 
		\begin{align}
		\bC_{bA}&=\bC_{AB}\;\gth^B_b,& \xi_{bA} &= -P_{AB}\;\gth^B_b,&  \bpsi_b &= \bgs^{-1}\left(  \frac{1}{n-1}P \; du_b - \frac{1}{2(n-2)}\covD^C C_{BC} \;\gth^B_b \right)\nonumber
		\end{align}
	\end{itemize}
		 where $P_{AB}$, $P$ are the Schouten tensor and its trace as defined by equations \eqref{Conformal geometry: Schouten tensor},\eqref{Conformal geometry: Schouten tensor trace} and ``the asymptotic shear'' $C_{AB}= \bgs^{-1}\bC_{AB}$ is a trace-free symmetric tensor satisfying $-\frac{1}{2}\parD_u C_{AB} = P_{AB} \big|_0$. In particular the restriction to the reduced tractor bundle is the normal Cartan connection.
	
	\begin{itemize}
		\item if $n=3$, 
		\begin{align}
		\bC_{bA}&=\bC_{AB}\;\gth^B_b,& \xi_{bA} &= \left(\frac{1}{2}\parD_u C_{AB} - \frac{P}{2} h_{AB}  \right)\gth^B_b,&  \bpsi_b &= \bgs^{-1} \left( \frac{1}{2}P \; du_b -\frac{1}{2}\covD^C C_{BC} \;\gth^B_b \right)\nonumber
		\end{align}
	\end{itemize}
where ``the asymptotic shear''  $C_{AB}= \bgs^{-1}\bC_{AB}$ is a trace-free symmetric tensor. In particular the restriction to the reduced tractor bundle is obtained by setting $N_{AB} = \parD_u C_{AB}$ in proposition \ref{Strong structure of null-infinity: Proposition, Reduced Null-Normal Cartan connection}.
	
	\begin{itemize}
		\item if $n=2$, and the splitting of the tractor bundle is defined in terms of a Laplace operator $\cL = \gD -\frac{1}{4}M$,
		\begin{align}
		\bC_{bA}&=0,&	\xi_{bA} &=  \frac{1}{2}M \;\gth_A{}_b,& \bpsi_b &= \bgs^{-1}\left( -\frac{1}{2}M \; du_b - N_A \;\gth^A_b \right)\nonumber
		\end{align}
	\end{itemize}
where ``the 3D mass aspect'' $M$ is a function on $\MNull$ and ``the 3D angular momentum aspect'' $N_A$ is a section of $\pi^* T^*\gS$.
\end{Proposition}

For the most physically relevant dimension $n+1=4$ this proposition implies that a null-normal connection is in one-to-one correspondence with a choice of ``asymptotic shear'' $C_{AB}$ and thus (from the discussion of section \ref{sss: Relation with Ashtekar/Geroch connections}) with a choice of equivalence class of connection as in \cite{geroch_asymptotic_1977,ashtekar_radiative_1981, ashtekar_a._symplectic_1981,ashtekar_symplectic_1982,ashtekar_asymptotic_1987,ashtekar_geometry_2015,ashtekar_null_2018}. The radiative degrees of freedom therefore exactly amounts to a choice of null-normal tractor connection. We will now prove that ``asymptotic shear'', ``mass aspect'' and ``angular momentum aspect'' correspond to the same objects as in section \ref{s: Radiative (or strong) structure of null-infinity} i.e the equivalence of null-normal tractor connections with Poincaré operators.

\subsection{Equivalence with the strong Null-infinity structure}

\subsubsection{Poincaré operators and null-normal Cartan connections, \texorpdfstring{$n\geq3$}{n>=3}}
Let $\left(\MNull \to \gS , \bh_{ab}, \bn^a \right)$ be a null-infinity manifold of dimension $n\geq3$. Let $D$ be a null-normal Cartan connection on its tractor bundle. Let $Y_I$ be a dual tractor and let $\bl\in \So{L}$ be the scale given by $\bl =\bX^I Y_I$. Let us consider the equation $D_b Y_I=0$, by proposition \ref{Strong structure of null-infinity: Proposition, Null-Normal Cartan connection} one has
\begin{equation}
D_b Y_I = \Mtx{ \covD_b & 0 & 0 & 0 \\
	-du_b & \covD_b & -\gth{}^C_b & 0 \\
	\frac{1}{2} \bC_{bA}& -\xi_{bA} & \covD_b \;\gd^C_A & \bgth_{bA} \\
	\bpsi_b & 0 & \bxi_b{}^C & \covD } \Mtx{ Y_u \\ \bl \\ \bY_C \\ \bY_{-} } =0.
\end{equation}
(where $\bC$, $\bxi$ and $\bpsi$ are functions of $h_{AB}$ and $C_{AB}$ as in proposition \ref{Strong structure of null-infinity: Proposition, Null-Normal Cartan connection})
solving for the first two lines and the trace of the third, one finds $Y_I = T_I\left(\bl \right)$ where $\bl \mapsto T_I\left(\bl \right)$ is Thomas operator, see equations \eqref{Weak structure of null-infinity: Thomas operator} and finally
\begin{equation}\label{Strong structure of null-infinity: covariant derivative of tractors}
D_bY_I = D_b T_I\left(\bl \right) = \Mtx{ 0 \\ 0 \\  \gth_b^B\; \left(\covD_{A}\covD_{B}\big|_0 +\frac{1}{2}\left[\bC_{AB} , \covD_{\bn} \right] \right)\bl \\ \star}.	
\end{equation}
Therefore $D_A T_B\left(\bl \right)$ coincides with a Poincaré operator \eqref{Strong structure of null-infinity: Poincare operator}. It follows that the symmetric tensor $\bC_{AB}$ from proposition \ref{Strong structure of null-infinity: Proposition, Null-Normal Cartan connection} coincides with the symmetric tensor parametrising Poincaré operators. In particular the transformations laws under change of well-adapted trivialisations must be the same.

\begin{Theorem}{Poincaré operators and Null-normal connection}\mbox{}\label{Strong structure of null-infinity: Thrm, equivalence between Poincare operator and Cartan connection}\\
	Let $\left(\MNull \to \gS , \bh_{ab} , \bn^a \right)$ be a null-infinity manifold of dimension $n$.	
	\begin{itemize}
		\item If $n =3$, choices of null-normal tractor connection are in one-to-one correspondence with choices of Poincaré operator. Choices of generalised Möbius structure are in one-to-one correspondence with choices of null-normal tractor connection on the reduced tractor bundle.
		\item If $n\geq4$, choices of null-normal Cartan connection are in one-to-one correspondence with choices of Poincaré operator inducing the canonical Möbius structure on $\left(\gS , \bh_{AB}\right)$.		
	\end{itemize}
\end{Theorem}	
	
	The discussion that lead to equation \eqref{Strong structure of null-infinity: covariant derivative of tractors} implies that if $Y_I$ is covariantly constant with respect to a null-normal Cartan connection $D$ then $\bl = \bX^I Y_I$ must be a zero of the related Poincaré operator. It is natural to ask about the converse statement. The answer is given by the proposition below and closes this section.
\begin{Proposition}	\label{Strong structure of null-infinity: Prop,equivalence between Poincare operator and Cartan connection}
Let $\left(\MNull \to \gS , \bh_{ab} , \bn^a \right)$ be a null-infinity manifold of dimension $n$. Let $\cP$ be a Poincaré operator on $\MNull$ and $D$ be the associated null-normal Cartan connection given by Theorem \ref{Strong structure of null-infinity: Thrm, equivalence between Poincare operator and Cartan connection}. Let $Y_I \in \So{\cT^*}$ be a dual tractor and $\bl = \bX^I Y_I$ the associated scale, then
\begin{itemize}
	\item if $n\geq4$, covariantly constant dual tractors are in one-to-one correspondence with scales $\bl \in \CSoL{k}$ which are zeros of $\cP$ i.e
	\begin{equation} \label{Strong structure of null-infinity: equivalence between Poincare operator and Cartan connection}
	DY_I=0 \qquad \Leftrightarrow \qquad \cP\left(\bl \right)=0.
	\end{equation}
	 	\item if $n=3$, covariantly constant dual tractor $DY_I =0$ such that $\bX^IY_I \in \CSoL{0}$ are in one-to-one correspondence with constant curvature scale $\bl \in \CSoL{0}$ which are zeros of $\cP$ i.e
	\begin{equation}
	 		DY_I=0 \quad\text{s.t}\quad X^IY_I \in \CSoL{0} \qquad \Leftrightarrow \qquad \cP\left(\bl \right)=0 \quad \text{and} \quad \covD_b R\left(\bl \right)=0.
	\end{equation}
	It there exists such a scale then the equivalence \eqref{Strong structure of null-infinity: equivalence between Poincare operator and Cartan connection} holds for all $k$.
\end{itemize}
\end{Proposition}
\begin{proof}
	
 The reasoning that lead to \eqref{Strong structure of null-infinity: covariant derivative of tractors} shows that $D_bY_I=0$ if and only if $Y_I = T_I\left(\bl \right)$ with $\bl \in \CSoL{k}$ and
 \begin{align}
D_AY_B = \cP_{AB}\left(\bl\right) &=0, & D_bY_{-} =\covD_b \bY_{-} + \bpsi_b k + \bxi_{b}{}^{C} \covD_C \bl &=0
 \end{align}
where $\xi_{bA}$ , $\bpsi_b$ are given by proposition \ref{Strong structure of null-infinity: Proposition, Null-Normal Cartan connection}. Since  $Y_I = T_I\left(\bl \right)$ implies \begin{equation}
 \bY_{-} = -\bgs^{-2}\frac{1}{n-1}\left( \gD + P \right) \bl
 \end{equation}
and $\bl \in \CSoL{k}$ one can check that $\bn \intD \left(DY_{-}\right) =0$ holds identically. We thus only need to understand the implication of
\begin{equation}
D_AY_{-} =\covD_A \bY_{-} + \bpsi_A k + \bxi_{A}{}^{B} \covD_B \bl =0.
\end{equation}
This will be a straightforward generalisation of standard results in tractor calculus (see \cite{bailey_thomass_1994} and \cite{curry_introduction_2018}): let us consider 
\begin{equation}
D_A Y_B = \covD_A \covD_B \bl +\frac{1}{2}\bC_{AB} k -\xi_{AB}\bl + \bh_{AB} \bY_{-} =0
\end{equation}
taking covariant derivative and trace we have
\begin{equation}
\covD^A D_A Y_B = \bgD \covD_B \bl +\frac{1}{2} \covD_{A}C^{A}{}_{B} k - \covD_A \bxi^A{}_{B} \;\bl- \bxi^A{}_B \covD_A \bl + \covD_B \bY_{-} =0
\end{equation}
\begin{equation}
\covD_B D^A Y_A =  \covD_B \bgD \bl - \covD_B \bxi^A{}_A\; \bl- \bxi^A{}_A \covD_B \bl + (n-1)\covD_B \bY_{-} =0
\end{equation}
taking the difference and making use of $\left[\covD_A ,\gD \right] = -R_{AB} \covD^B$ and $R_{AB} = (n-3)P_{AB} + P h_{AB}$ we obtain 
\begin{equation}\label{Strong structure of null-infinity: Proof of zeros of poincare operator}
(n-2)\left(\covD_B \bY_{-} - \bgs^{-2}P_{AB} \covD^A\bl \right) -\frac{1}{2} \covD_A \bC^A{}_{B}\; k + \bgs^{-2}\left(\xi_{AB} +P_{AB} -\left(\xi +P \right) h_{AB} \right) \covD^A \bl  + \left(\covD_A \bxi^A{}_{B} - \covD_B \bxi\right) \bl =0
\end{equation}	
where $\bxi = \bxi^C{}_C$.

If $n\geq4$, by making use of proposition \ref{Strong structure of null-infinity: Proposition, Null-Normal Cartan connection} this last equation can be recast as
\begin{equation}
D_AY_{-} = \bgs^{-2}\frac{1}{n-2}\left(\covD^B P_{AB} -\covD_A P\right)\bl
\end{equation}
however $\covD^B P_{AB} -\covD_A P$ holds identically (see equation \eqref{Conformal geometry: Cotton tensor, identities}) which concludes the proof for $n\geq4$.

If $n=3$, by making use of proposition \ref{Strong structure of null-infinity: Proposition, Null-Normal Cartan connection} we can rewrite equation \eqref{Strong structure of null-infinity: Proof of zeros of poincare operator} as 
\begin{equation}\label{Strong structure of null-infinity: Proof of zeros of poincare operator2}
	D_AY_{-} =  \bgs^{-2}\frac{1}{2}\left(\covD^B\parD_u C_{AB}  +\frac{1}{2}\covD_A R \right)\bl
\end{equation}
Let us suppose that $\bl \in \CSoL{0}$ and let us take $\left(\bgsh =\bl ,u \right)$ as well adapted trivialisation then $\cP\left(\bl \right)=0$ implies that we must have $\parD_u \Ch_{AB}=0$ and the vanishing of $D_A\Yh_{+}$ is equivalent to
\begin{equation}
	D_A\Yh_{-} = \left( \frac{1}{4}\covD_A \Rh \right)\; \bl^{-1}	= 0
\end{equation}
i.e $\bgsh=\bl$ is a scale of constant scalar curvature. Now, suppose that there exists a constant curvature scale $\bgs \in \So{L_{\gS}}$ such that $\cP\left(\bgs\right)=0$. In any well-adapted trivialisation of the form $\left(\bgs, u\right)$ we must have $\parD_u C_{AB} =0$ (because $\bgs$ is a zero of $\cP$) and $\covD_A R=0$ (because $\bgs$ has constant scale curvature). Therefore if such a scale $\bgs$ exists and $Y_I$ is any section such that $D_b Y_+ =0$, $D_b Y_u =0$, $D_bY_A =0$ (equivalently $\cP\left(\bl \right)=0$ with $\bl= Y_I \bX^I \in \CSoL{k}$) then $D_bY_{-}=0$ is automatically satisfied as a result of \eqref{Strong structure of null-infinity: Proof of zeros of poincare operator2}. This concludes the proof.
\end{proof}	
 
 \subsubsection{Poincaré operators and null-normal Cartan connections, \texorpdfstring{$n=2$}{n=2}.}
 
 Let $\left(\MNull \to \gS , \bh_{ab}, \bn^a \right)$ be a null-infinity manifold of dimension $n=2$. Let \begin{equation}
 	\cL = \gD - \frac{1}{4}M \from \gC_{k}\left[ L^{\frac{1}{2}} \right] \to \gC_{k}\left[ L^{-\frac{3}{2}} \right]
 \end{equation} be a generalised Laplace operator and let $D$ be a null-normal Cartan connection on its tractor bundle. Let $Y_I$ be a dual tractor and let $\bl\in \So{L}$ be the scale given by $\bl =\bX^I Y_I$. From proposition \ref{Strong structure of null-infinity: Proposition, Null-Normal Cartan connection} one has
 \begin{equation}
 	D_b Y_I = \Mtx{ \covD_b & 0 & 0 & 0 \\
 		-\bgs du_b & \covD_b & -\gth_b^C & 0 \\
 		0& -\bgs^{-2}\frac{1}{2} M \bgth_b{}_A & \covD_b \;\gd^C_A & \bgth_{bA} \\
 		\bgs^{-1}\left( -\frac{1}{2}M du_b - N_A \gth_b^A\right) & 0 & \bgs^{-2}\frac{1}{2}M \gth_b^C & \covD_b } \Mtx{ Y_u \\ \bl \\ \bY_C \\ \bY_{-} } =0.
 \end{equation}
 solving for the three first lines one finds $Y_I = T_I\left(\bl \right)$ where $\bl \mapsto T_I\left(\bl \right)$ is Thomas operator, see equations \eqref{Weak structure of null-infinity: Thomas operator} (Recall that in this context our convention is to replace $P$ by $-\frac{1}{2}M$) and finally
 \begin{equation}\label{Strong structure of null-infinity: covariant derivative of tractors n=2}
 	D_bY_I = D_b T_I\left(\bl \right) = \gth^A\;\Mtx{ 0 \\ 0 \\  0 \\ -\bgs^{-1}\Big(\left( \covD_A \gD - M \covD_A -\frac{1}{2}\covD_A M  \right)\bl + N_A k \Big)}.	
 \end{equation}
 i.e $D_A Y_{-} = -\cP\left(\bl\right)$. This therefore proves the following:

 \begin{Theorem}{Poincaré operators and Null-normal connections}\mbox{}\label{Strong structure of null-infinity: Thrm, equivalence between Poincare operator and Cartan connection, dimension 2}.	
 	
 Let $\left(\MNull \to \gS , \bh_{ab} , \bn^a \right)$ be a null-infinity manifold of dimension $2$. Choices of null-normal Cartan connections are in one-to-one correspondence with choices of Poincaré operators. 
 \end{Theorem}
 
 \begin{Proposition}	
 Let $\cP$ be a Poincaré operator on $\MNull$ and $D$ be the associated null-normal Cartan connection given by the Theorem \ref{Strong structure of null-infinity: Thrm, equivalence between Poincare operator and Cartan connection, dimension 2}. Let $Y_I \in \So{\cT^*}$ be a dual tractor and $\bl = \bX^I Y_I$ the associated scale, then covariantly constant dual tractors are in one-to-one correspondence with scales $\bl \in \CSoL{k}$ which are zeros of $\cP$ i.e
 		\begin{equation} \label{Strong structure of null-infinity: equivalence between Poincare operator and Cartan connection n=2}
 			DY_I=0 \qquad \Leftrightarrow \qquad \cP\left(\bl \right)=0.
 		\end{equation}
 \end{Proposition}
 
 \subsection{Tractor curvature}

\subsubsection{Curvature tensor,  \texorpdfstring{$n\geq 4$}{n>=4}.} 

From proposition \ref{Strong structure of null-infinity: Proposition, Null-Normal Cartan connection} one obtains that, for any dimension $n\geq4$ and for any well adapted trivialisation $\left(\bgs, u\right)$, the curvature of a null-normal tractor connection must be of the form
\begin{equation*}
 \hspace*{-1cm} F^I{}_{Jcd} = \gth_c{}^C \gth_d{}^D 
\Mtx{
0 & 0& 0 &0 \\
\bC_{CD}{}^A & W^A{}_{BCD} & 0 & 0 \\
0 & -C_{CD}{}_B & 0 &0 \\
\frac{1}{n-2} \covD_{[C} \covD^E \bC^{(0)}{}_{D] E} + \bC^{(0)}{}_{[C}{}^E \;P_{D] E} & -\covD_{[C} \bC^{(0)}{}_{D]B} - \frac{1}{n-2} \covD^E \bC^{(0)}{}_{E[C} h_{D]B} +\bgs^{-1} u \;C_{CDB}&0 &0
}
\end{equation*}
Where the Cotton tensor $C_{AB}{}_C$ and the Weyl tensor $W^A{}_{BCD}$ of $h_{AB}$ are defined as in section \ref{ss: Curvature tensors} and the ``zero mode'' of the asymptotic shear $\bC^{(0)}_{AB}$ is a tensor on $\gS$ that we take to be defined by the relation \begin{equation}
\bC_{AB} = \bC^{(0)}_{AB} -2 \bgs\;u\; P_{AB}.
\end{equation}
 
The first three lines of the curvature tensor are the coefficients of the tractor curvature of the normal Cartan connection of $\left(\gS , \bh_{AB} \right)$ and their vanishing is equivalent to local conformal flatness of $\bh_{AB}$. The last line could be called ``the curvature of the zero mode'' and encodes the information which is specific of the null-normal connection.

\subsubsection{Curvature tensor,  \texorpdfstring{$n= 3$}{n=3}.} 

From proposition \ref{Strong structure of null-infinity: Proposition, Null-Normal Cartan connection} one can derive that, for $n=3$ and for any well adapted trivialisation $\left(\bgs, u\right)$, the curvature of a null-normal tractor connection must be of the form
\begin{equation*}
F^I{}_{Jcd} = 
\Mtx{
	0 & 0& 0 &0 \\
	\beps_{cd} \; \bK_E \; \beps^{AE}  + 2 \;du_{[c}\; \gth_{d]}{}^D\; \parD_u \bN_{D}{}^A & 0 & 0 & 0 \\
	0 & -\beps_{cd} \;  \bK_E \; \beps_B{}^E -2 \;du_{[c}\; \gth_{d]}{}^D\; \parD_u N_{D}{}_B & 0 &0 \\
	 \frac{1}{2}\; \beps_{cd} \; \bK -  2 \bgs\;du_{[c}\; \gth_{d]}{}^D \; \bK_D & 0&0 &0
}
\end{equation*}
Here $\beps_{cd} = \gth_c{}^C \; \gth_d{}^D \beps_{CD}$ where $\beps_{AB}\in \gL^2 T^*\gS \otimes L^2_{\gS}$ is the volume form on $\gS$ and $\beps^{AB}$ its inverse. The news tensor $N_{AB}$ is defined as $N_{AB} = \parD_u C_{AB}$. Finally the curvature elements $\bK_A \in \So{T^* \gS \otimes L^{-2}}$ and $\bK\in \So{L^{-3}}$ are given by
\begin{equation}
\bK_A = -\bgs^{-2}\frac{1}{2}\left( \covD_C N_{A}{}^C + \frac{1}{2} \covD_{A} R \right)
\end{equation}
and
\begin{equation}
\bK = \beps^{CD}\left( \covD_{[C} \covD^E \bC_{D] E} - \frac{1}{2}\bC_{[C}{}^E \;N_{D] E}   \right) 
\end{equation}

Both $\parD_u N_{A}{}_B$ and $\bK_A$ parametrize the curvature of the generalised Möbius operator (equivalently of the reduced null-normal tractor connection): as we already discussed, $\parD_u N_{A}{}_B$ is the obstruction for this operator to be a genuine Möbius operator on $\left(\gS , \bh_{AB}\right)$ and $\bK_A$ is the obstruction for this Möbius structure to be integrable. When these two tensors vanish, the whole tractor curvature is parametrized by a section of $L^{-3}$ only:
\begin{equation}\label{Strong structure of null-infinity: curvature of the zero mode n=3}
\bK = \beps^{CD} \left(\; \covD_{[C} \covD^E \bC^{(0)}{}_{D] E} - \frac{1}{2}\bC^{(0)}{}_{[C}{}^E \;N_{D] E} \right) 
\end{equation}
here $\bC^{(0)}{}_{AB}$ is the ``zero mode of the asymptotic shear'' $\bC_{AB} = \bC^{(0)}{}_{AB} - \frac{1}{2}\bgs u \; N_{AB}$. The ``curvature of the zero mode'' \eqref{Strong structure of null-infinity: curvature of the zero mode n=3} is then the obstruction to existence of solutions to the good-cut equations. 
  
\subsubsection{Curvature tensor,  \texorpdfstring{$n= 2$}{n=2}.}   
 Finally, from proposition \ref{Strong structure of null-infinity: Proposition, Null-Normal Cartan connection} one can derive that, for $n=2$ and for any well adapted trivialisation $\left(\bgs, u\right)$, the curvature of a null-normal tractor connection must be of the form
 \begin{equation}
 F^I{}_{Jcd} = 
 \Mtx{
 	0 & 0& 0 &0 \\
   - \bgs^{-2}\;du_{[c}\; \gth_{d]}{}^A\; \parD_u M & 0 & 0 & 0 \\
 	0 &\bgs^{-2}\; du_{[c}\;\bgth_{d]}{}_B \; \parD_u M  & 0 &0 \\
4\bgs^{-1}\; du_{[c} \; \gth_{d]}{}^E \left( \parD_u N_E - \frac{1}{2} \covD_E M  \right)    & 0&0 &0
 }
 \end{equation} 
 As we previously discussed $\parD_u M$ is the obstruction for the associated generalised Laplace operator $\cL$ to induce a genuine Laplace operator on $\gS$. 

The vanishing of this curvature has a nice physical interpretation: Let us suppose that $\MNull \to \gS$ is the boundary of a 3D asymptotically flat space-times then Einstein's equations implies (see e.g \cite{barnich_aspects_2010}) that the 3D mass aspect and angular momentum aspect must satisfy
\begin{align}
\parD_u M &=0,  &  \parD_u N_A &= \frac{1}{2}\covD_A M.
\end{align}
 From the intrinsic perspective that we take in this article, one sees that these equations amount to the flatness of the associated tractor connection.

\section{Gravity vacua degeneracy, ``soft modes'' and memory effect}

In this section we restrict to the most physically relevant case where $\left( \MNull \to \gS , \bh_{ab} , \bn^a\right)$ is a null-infinity manifold over the conformal sphere $\left(\gS , \bh_{AB} \right)= \left(  \S^{n-1} , \bh_{AB}{}(\S^{n-1} )\right)$. We wish to show how the formalism that we presented in this article allows to reproduce, in a completely geometrical manner, the results from \cite{geroch_asymptotic_1977,ashtekar_radiative_1981, ashtekar_a._symplectic_1981,ashtekar_symplectic_1982,ashtekar_asymptotic_1987,ashtekar_geometry_2015,ashtekar_null_2018}.

\subsection{Gravity vacua}

Let $\left( \MNull \to \S^{n-1} , \bh_{ab} , \bn^a \right)$ be a null-infinity manifold over the conformal sphere, we will call the space of compatible \emph{flat} strong null-infinity structures the space of ``gravity vacua''. Here ``flat strong null-infinity structures'' are defined, through theorem \ref*{Strong structure of null-infinity: Thrm, equivalence between Poincare operator and Cartan connection}, as strong null-infinity structures corresponding to flat null-normal tractor connections.

The space of gravity vacua is not a point, there is a whole moduli which is closely related to the BMS group. In particular, a choice of gravity vacuum amounts to choosing a copy of the Poincaré group inside the BMS group:
\begin{Proposition}\label{Gravity vacua degeneracy and soft modes: Prop, BMS action}
Let $\left( \MNull \to \S^{n-1} , \bh_{ab} , \bn^a \right)$ be a null-infinity manifold. The BMS group acts transitively on the space of flat null-normal tractor connections with stabilisers isomorphic to the Poincaré group $\Iso\left(n,1\right) = \bbR^{n+1} \rtimes \SO\left(n,1\right)$.
\end{Proposition}
\begin{proof}
Let $D_1$ and $D_2$ be two flat null-normal tractor connections. They induce on $\S^{n-1}$ a normal tractor connection, however there is a unique normal tractor connection defined on the whole of $\S^{n-1}$ so these connections must coincides. In particular, we can pick the scale $\bgs$ corresponding to the round sphere metric on $\S^{n-1}$ and it is a zero of the Poincaré operators associated to both $D_1$ and $D_2$. Since the connections are flat their must exists flat well-adapted trivialisation of the form $\left(\bgs, u_1\right)$ and	$\left(\bgs, u_2\right)$ for each of these connections. Let us note $u_2 = u_1 - \xi$ and let us consider the automorphism of $\MNull \to \S^{n-1}$ inducing the identity on the base and such that $\Phi^* u_1 = u_1 - \xi = u_2$. Let us write $\bl_1 = \bgs u_1$, $\bl_2 = \bgs u_2$, in particular $\Phi^* \bl_1 = \bl_2$. Let $\cP_1$,$\cP_2$ be the Poincaré operator associated with $D_1$, $D_2$. We have
\begin{align}
\Phi^*\cP_1 \left( \bl_2 \right)  &:= \Phi^*\left( \cP_1 \left( (\Phi^{-1})^* \bl_2  \right)   \right) \\ \nonumber 
&= \Phi^*\left( \cP_1 \left( \bl_1  \right)   \right) \\ \nonumber
&=0.
\end{align}
Thus $\left(\bgs , u_2\right) $ is a flat well-adapted trivialisation for both $\Phi^*\cP_1$ and $\cP_2$. This implies that the operator must be the same (in this well-adapted trivialisation they both have zero asymptotic shear) and thus
\begin{equation}
\Phi^*D_1 = D_2.
\end{equation}
This proves the transitivity of the action of the BMS group.

Elements of the BMS group stabilising a gravity vacua correspond to symmetries of the Poincaré operator: By proposition \ref{Strong structure of null-infinity: Prop, Symmetries} these are isomorphic to the Poincaré group.
\end{proof}

An interesting property of the space of gravity vacua is that their behaviour on a well-chosen open set defines them completely:
\begin{Proposition}\label{Gravity vacua degeneracy and soft modes: Prop, global behavior of vacua}
	Let $\left( \MNull \to \S^{n-1} , \bh_{ab} , \bn^a \right)$ be a null-infinity manifold. Let $U$ be a connected open set of $\MNull$ such that $\pi \left(U \right) = \S^{n-1}$ and let $D_U$ be a flat null-normal tractor connection on $U$. Then there exists a unique flat null-normal tractor connection $D$ on $\MNull$ which coincides with $D_{U}$ on $U$.
\end{Proposition}
\begin{proof}
The existence of a connection $D$ on $\MNull$ extending $D_{U}$ follows from the fact that vanishing of the curvature is a PDE. Suppose that $D_1$ and $D_2$ are two flat null-normal connections extending $D_U$. By proposition \ref{Gravity vacua degeneracy and soft modes: Prop, BMS action} there must exist a diffeomorphism $\Phi \from \MNull \to \MNull$ such that $\Phi^* D_1 = D_2$. Let $\cP_1$ and $\cP_2$ be Poincaré operators associated with $D_1$ and $D_2$. Let $\left(\bgs , u\right)$ be a flat well-adapted trivialisation for $\cP_1$ and let $\gO$, $\xi$ be two functions on $\S^{n-1}$ given by $\left(\Phi^* \bgs = \gO^{-1}\bgs , \Phi^* u = \gO\left(u-\xi\right)   \right) $. By proposition \ref{Strong structure of null-infinity: Prop, Symmetries} we must have
\begin{align}
\cP_2 -\cP_1 &= \Phi^* \cP_1  - \cP_1 \\ \nonumber 
&= 	 \frac{1}{2}\gO^{-1}\; \bigg[ 2\;\covD_A \covD_B \big|_0 \xi + 2\; \uh  \;\covD_A \covD_B \big|_0 \gO^{-1} , \parD_u \bigg].
\end{align}
It follows from this last relation that if $\cP_2 -\cP_1$ coincide on $U = (\ga, \gb) \times \S^{n-1}$ with $\ga,\gb \in \bbR$ they must coincide everywhere.	
\end{proof}

\subsection{Memory effect}

We now want to study simple properties of a class of null-normal tractor connections related to the passage of a finite gravitational wave. Let us consider a $(n+1)$-dimensional asymptotically flat space-time with conformal boundary $\left( \MNull \to \S^{n-1} , \bh_{ab} , \bn^a \right)$. 

As we already discussed, in dimension $n+1=3$ Einstein's equations imply that null-normal tractor connections on $\MNull$ must be flat, this correspond to the fact that there is no gravitational waves in three dimension. In dimension $n+1 \ge 4$ it is known that the information about the radiation is not encoded in the asymptotic shear and thus not in the tractor connection at null-infinity. From the intrinsic picture presented here this appears in the following way: by proposition \ref{Strong structure of null-infinity: Proposition, Null-Normal Cartan connection} the only free parameter in the choice of null-normal tractor connection is the zero mode of the asymptotic shear: In these dimensions, there is therefore no dynamics that could be related to the presence of gravitational waves.

 However, in the most physically relevant dimension $n+1 =4$, the asymptotic shear is completely unconstrained and by proposition \ref{Strong structure of null-infinity: Proposition, Null-Normal Cartan connection} parametrizes null-normal tractor connections. The presence of curvature, where the second derivatives in $u$ for the asymptotic shear appear, correspond to the presence of gravitational waves: If $T^I{}_J{}_{cd}$ is a tensor parametrizing the presence of radiations, then the equation
 \begin{equation}
 F^I{}_{J}{}_{cd} = T^I{}_J{}_{cd}
 \end{equation}
 can be effectively thought of as an evolution equation for the asymptotic shear.
 
  In particular, tractor connections corresponding to a burst of gravitational waves must be such that ``far in the future'' and ``far in the past'' the tractor connection is flat. For this reason we will now restricts ourselves to tractor connections which are such that their curvature has compact support. We will call ``neighbourhood of $\MNull_{\infty}$'' (resp of $\MNull_{-\infty}$) open subsets $U$ of $\MNull$ such that in a trivialisation of $\MNull \to \gS$ we have $U = (\eps, \infty) \times \gS$ (resp $(-\infty, \eps) \times \gS$) with $\eps \in \bbR$. If $D$ is a null-normal tractor connection with a curvature of compact support then there must exists $U_{\pm \infty}$ two neighbourhoods of $\MNull_{\pm\infty}$ such that the curvature of $D$ vanishes. By proposition \ref{Gravity vacua degeneracy and soft modes: Prop, global behavior of vacua} this defines two flat null-normal connections $D_{(\pm \infty)}$ on the whole of $\MNull$ such that
\begin{equation}
 \left( D- D_{(\pm \infty)} \right)\big|_{U_{\pm \infty}} =0.
\end{equation}
In this sense a burst of gravitational waves realizes the transition between two gravity vacua: a ``soft mode'' has been excited by the passage of the radiation. By proposition \ref{Gravity vacua degeneracy and soft modes: Prop, BMS action} there must exists an element $\Phi$ of the BMS group such that
\begin{equation}
\Phi^* D_{(-\infty)} = D_{(\infty)}.
\end{equation}
Finally let $\left(\bgs ,u\right)$ be a flat well-adapted trivialisation for $D_{(-\infty)}$, by proposition \ref{Strong structure of null-infinity: Prop, Symmetries} their must exits two functions $\xi$ and $\gO$ on $\S^{n-1}$ such that

\begin{equation}
\Phi^* D_{(-\infty)}{}_{b}  - D_{(\infty)}{}_{b}  =  \gth^B_b\; \Mtx{ 0 & 0  & 0 & 0  \\
	0 &	0& 0& 0  \\
	0 &  0 & 0 & 0\\
	-\frac{1}{2}\covD_C \bm{\gd C^C{}_{B}}\;  & -\frac{1}{2}\bm{\gd C_{AB}} & 0	& 0
}
\end{equation}
with $\bm{\gd C_{AB}} = \bgs \; \gO^{-1}\; \covD_A \covD_B \big|_0 \xi$.

\section{Conclusion and outlook}

We described a tractor calculus for null-infinity manifolds $\left( \MNull \to \gS, \bh_{ab}, \bn^a \right)$ as a generalisation of the standard tractor calculus of conformal geometry. The essential difference is that there is no equivalent of the normal Cartan connection, rather we found a infinite family of null-normal tractor connections in one-to-one correspondence with a particular class of second order differential that we called ``Poincaré operators''. 

From the physics point of view, i.e when $\left( \MNull \to \gS, \bh_{ab}, \bn^a\right )$ is taken to be the conformal boundary of an asymptotically flat space-times, null-normal tractor connections then correspond to all possible choices of ``asymptotic shear''. In the particular case of four dimensional asymptotically flat space-times, the asymptotic shear is known to encode the radiative degrees of freedom at null-infinity. Therefore, in this physically relevant dimension, the picture is especially compelling: the gravitational radiation then precisely correspond to a choice of null-normal tractor connection.

Tractor calculus is an especially powerful tool to construct conformal invariants \cite{gover_invariant_2001,cap_parabolic_2009,curry_introduction_2018}: as an example, it was explained in \cite{gover_conformally_2003} how to obtain tractor expressions (i.e manifestly conformally conformally invariants expressions) for the conformally invariant powers of the Laplacian known as GJMS operators (after \cite{graham_conformally_1992}) and similarly for Q-curvature. Since the tractor calculus described in this work is in all respect similar one can therefore expects to produce, by applying the same algorithms, similar invariants. In particular, when applied to three-dimensional null-infinity manifolds, this should allow to obtain invariants of the radiative phase space of gravity which would otherwise be very difficult to guess. As a distinctive example of this phenomenon, we remark that the Poincaré operators described in this article have, to the best knowledge of the author, never been discussed previously in the literature, this is despite the fact that this should be one of the simplest possible differential invariant. From this perspective it also seems that the symplectic form on the phase space and the BMS charges that have been extensively discussed in the literature \cite{barnich_bms_2011,henneaux_bms_2018,de_paoli_gauge-invariant_2018,henneaux_asymptotic_2019-1, barnich_bms_2020} should also have tractor expressions (in particular these would be manifestly conformally invariant). It would also be very illuminating to relate the intrinsic geometrical description that was given in this article to the geometric action described in \cite{barnich_geometric_2017}.

Another aspect which is worth emphasising is the flexibility of the formalism: it can be formulated in a generic dimension and, in particular, naturally encompass the geometry of both 3D and 4D asymptotically flat space-times. Considering how different these two examples might naively appear, this is a nice surprise. What is more it is known that the (usual) tractor calculus is very well adapted to studying asymptotically AdS and asymptotically dS space-times (see e.g \cite{curry_introduction_2018}), it therefore suggests that results described in this work could serve as a unifying framework to deal with the difficult question of understanding gravitational wave physics when the cosmological constant is non-zero (see \cite{ashtekar_asymptotics_2014,ashtekar_asymptotics_2015-1,ashtekar_asymptotics_2015,ashtekar_gravitational_2016,ashtekar_implications_2017,ashtekar_asymptotics_2019,compere_superboost_2018,compere_lambda-bms_4_2019}). It also seems very likely that a similar geometrical description of the radiative phase-space should exists for Cauchy surface at finite distance as is for example considered in the ``edge mode'' literature \cite{donnelly_local_2016,gomes_observers_2017,hopfmuller_gravity_2017,hopfmuller_null_2018,chandrasekaran_symmetries_2018,speranza_local_2018,gomes_unified_2018,freidel_gravitational_2019,riello_soft_2020} or when studying symmetries of black-hole horizons \cite{donnay_supertranslations_2016,donnay_extended_2016,donnay_carrollian_2019}. There is also certainly a sense in which the result described in this paper must encompass as a particular case the geometry of Carroll manifold \cite{duval_carroll_2014,hartong_gauging_2015,bekaert_connections_2018,morand_embedding_2018,ciambelli_carroll_2019}. Finally, another intriguing possibility is to formulate a quantum theory of gravity's radiative degrees of freedom as a quantum theory on  null-infinity, see \cite{ashtekar_asymptotic_1987,adamo_perturbative_2014,geyer_ambitwistor_2015,bagchi_field_2019,gupta_constructing_2020} for works in this direction.
  
 \section*{Acknowledgements}
 
 The author is indebted to Tim Adamo, Romain Ruzziconi and Carlos Scarinci for crucial discussions and sharing their extensive knowledge on (respectively) good-cuts and their relations to asymptotic shear, the BMS formalism and asymptotic symmetries, Riemann surfaces and 3D gravity. The author also would like to thank Glenn Barnich, Joel Fine and Kirill Krasnov for their comments on a first version of this article.
 
 This research was funded by a FNRS ``chargé de recherche'' grant.
  
\bibliographystyle{ieeetr}
\bibliography{Biblio,BiblioExtra}
 
\end{document}